\documentclass[11pt]{article}

\usepackage{amsthm}
\usepackage{graphicx} 
\usepackage{array} 

\usepackage{amsmath, amssymb, amsfonts, verbatim}
\usepackage{hyphenat,epsfig,subcaption,multirow}
\usepackage[font=small,labelfont=bf]{caption}
\usepackage{nicefrac}
\usepackage{paralist}

\usepackage[usenames,dvipsnames]{xcolor}
\usepackage[ruled]{algorithm2e}

\DeclareFontFamily{U}{mathx}{\hyphenchar\font45}
\DeclareFontShape{U}{mathx}{m}{n}{
      <5> <6> <7> <8> <9> <10>
      <10.95> <12> <14.4> <17.28> <20.74> <24.88>
      mathx10
      }{}
\DeclareSymbolFont{mathx}{U}{mathx}{m}{n}
\DeclareMathSymbol{\bigtimes}{1}{mathx}{"91}

\usepackage{tcolorbox}
\tcbuselibrary{skins,breakable}
\tcbset{enhanced jigsaw}

\usepackage[normalem]{ulem}
\usepackage[compact]{titlesec}

\definecolor{DarkRed}{rgb}{0.5,0.1,0.1}
\definecolor{DarkBlue}{rgb}{0.1,0.1,0.5}

\usepackage{nameref}
\definecolor{ForestGreen}{rgb}{0.1333,0.5451,0.1333}
\definecolor{Red}{rgb}{0.9,0,0}
\usepackage[linktocpage=true,
	pagebackref=true,colorlinks,
	linkcolor=DarkRed,citecolor=ForestGreen,
	bookmarks,bookmarksopen,bookmarksnumbered]
	{hyperref}
\usepackage[noabbrev,nameinlink]{cleveref}
\crefname{property}{property}{Property}
\creflabelformat{property}{(#1)#2#3}
\crefname{equation}{eq}{Eq}
\creflabelformat{equation}{(#1)#2#3}

\usepackage{bm}
\usepackage{url}
\usepackage{xspace}
\usepackage[mathscr]{euscript}

\usepackage{tikz}
\usetikzlibrary{arrows}
\usetikzlibrary{arrows.meta}
\usetikzlibrary{shapes}
\usetikzlibrary{backgrounds}
\usetikzlibrary{positioning}
\usetikzlibrary{decorations.markings}
\usetikzlibrary{decorations.pathreplacing} 
\usetikzlibrary{patterns}
\usetikzlibrary{calc}
\usetikzlibrary{fit}
\usetikzlibrary{decorations}

\usepackage[framemethod=TikZ]{mdframed}

\usepackage[noend]{algpseudocode}
\makeatletter
\def\BState{\State\hskip-\ALG@thistlm}
\makeatother

\usepackage{cite}
\usepackage{enumitem}
\setlist[itemize]{leftmargin=20pt}
\setlist[enumerate]{leftmargin=20pt}

\usepackage[margin=1in]{geometry}

\usepackage{thmtools}
\usepackage{thm-restate}

\newtheorem{theorem}{Theorem}
\newtheorem{lemma}{Lemma}[section]
\newtheorem{proposition}[lemma]{Proposition}

\newtheorem{claim}[lemma]{Claim}
\newtheorem{fact}[lemma]{Fact}

\newtheorem{definition}[lemma]{Definition}
\newtheorem{problem}{Problem}

\newtheorem*{claim*}{Claim}
\newtheorem*{assumption*}{Assumption}
\newtheorem*{proposition*}{Proposition}
\newtheorem*{lemma*}{Lemma}

\newtheorem{observation}[lemma]{Observation}

\newtheorem*{theorem*}{Theorem}

\crefname{lemma}{Lemma}{Lemmas}
\crefname{claim}{claim}{claims}
\crefname{property}{Property}{Properties}
\crefname{invariant}{Invariant}{Invariants}

\newtheorem{mdresult}{Result}
\newenvironment{result}{\begin{mdframed}[backgroundcolor=lightgray!40,topline=false,rightline=false,leftline=false,bottomline=false,innertopmargin=2pt]\begin{mdresult}}{\end{mdresult}\end{mdframed}}

\theoremstyle{definition}

\newtheorem{remark}[lemma]{Remark}

\newtheorem*{mdproblem*}{Problem}
\newenvironment{Problem*}{\begin{mdframed}[hidealllines=false,innerleftmargin=10pt,backgroundcolor=gray!10,innertopmargin=5pt,innerbottommargin=5pt,roundcorner=10pt]\begin{mdproblem*}}{\end{mdproblem*}\end{mdframed}}
\newtheorem{mddefinition}[lemma]{Definition}

\newtheorem*{mddefinition*}{Definition}
\newenvironment{Definition*}{\begin{mdframed}[hidealllines=false,innerleftmargin=10pt,backgroundcolor=white!10,innertopmargin=5pt,innerbottommargin=5pt,roundcorner=10pt]\begin{mddefinition*}}{\end{mddefinition*}\end{mdframed}}
\newtheorem{mdremark}{Remark}

\newtheorem{mddistribution}{Distribution}
\newenvironment{Distribution}{\begin{tbox}\begin{mddistribution}}{\end{mddistribution}\end{tbox}}

\newtheorem{mdprotocol}{Protocol}
\newenvironment{Protocol}{\begin{tbox}\begin{mdprotocol}}{\end{mdprotocol}\end{tbox}}

\newenvironment{ourbox}{\begin{mdframed}[hidealllines=false,innerleftmargin=10pt,backgroundcolor=white!10,innertopmargin=2pt,innerbottommargin=5pt,roundcorner=10pt]}{\end{mdframed}}

\newtheorem{mdalgorithm}{Algorithm}

\allowdisplaybreaks

\renewcommand{\qed}{\nobreak \ifvmode \relax \else
      \ifdim\lastskip<1.5em \hskip-\lastskip
      \hskip1.5em plus0em minus0.5em \fi \nobreak
      \vrule height0.75em width0.5em depth0.25em\fi}

\renewcommand{\leq}{\leqslant}
\renewcommand{\geq}{\geqslant}

\renewcommand{\ge}{\geq}

\newcommand{\tvd}[2]{\ensuremath{\norm{#1 - #2}_{\mathrm{tvd}}}}
\newcommand{\Ot}{\ensuremath{\widetilde{O}}}
\newcommand{\eps}{\ensuremath{\varepsilon}}
\newcommand{\Paren}[1]{\Big(#1\Big)}
\newcommand{\Bracket}[1]{\Big[#1\Big]}
\newcommand{\bracket}[1]{\left[#1\right]}
\newcommand{\paren}[1]{\ensuremath{\left(#1\right)}\xspace}
\newcommand{\card}[1]{\left\vert{#1}\right\vert}
\newcommand{\Omgt}{\ensuremath{\widetilde{\Omega}}}

\newcommand{\IN}{\ensuremath{\mathbb{N}}}

\newcommand{\norm}[1]{\ensuremath{\|#1\|}}

\newcommand{\set}[1]{\ensuremath{\left\{ #1 \right\}}}
\newcommand{\poly}{\mbox{\rm poly}}

\DeclareMathOperator*{\Exp}{\ensuremath{{\mathbb{E}}}}
\DeclareMathOperator*{\Prob}{\ensuremath{\textnormal{Pr}}}
\renewcommand{\Pr}{\Prob}

\newcommand{\Ex}{\Exp}

\newenvironment{tbox}{\begin{tcolorbox}[
		enlarge top by=5pt,
		enlarge bottom by=5pt,
		 breakable,
		 boxsep=0pt,
                  left=4pt,
                  right=4pt,
                  top=10pt,
                  arc=4pt,
                  boxrule=1pt,toprule=1pt,
                  colback=white
                  ]
	}
{\end{tcolorbox}}

\newcommand{\event}{\ensuremath{\mathcal{E}}}

\newcommand{\rv}[1]{\ensuremath{{\mathsf{#1}}}\xspace}
\newcommand{\rA}{\rv{A}}
\newcommand{\rB}{\rv{B}}
\newcommand{\rC}{\rv{C}}
\newcommand{\rD}{\rv{D}}

\newcommand{\rX}{\rv{X}}
\newcommand{\rY}{\rv{Y}}
\newcommand{\rZ}{\rv{Z}}

\newcommand{\supp}[1]{\ensuremath{\textnormal{\text{supp}}(#1)}}
\newcommand{\distribution}[1]{\ensuremath{\textnormal{dist}(#1)}\xspace}

\newcommand{\kl}[2]{\ensuremath{\mathbb{D}(#1~||~#2)}}
\newcommand{\II}{\ensuremath{\mathbb{I}}}
\newcommand{\HH}{\ensuremath{\mathbb{H}}}
\newcommand{\mi}[2]{\ensuremath{\def\mione{#1}\def\mitwo{#2}\mireal}}
\newcommand{\mireal}[1][]{
  \ifx\relax#1\relax%
    \II(\mione \,; \mitwo)%
  \else%
    \II(\mione \,; \mitwo\mid #1)%
  \fi
}
\newcommand{\en}[1]{\ensuremath{\HH(#1)}}

\newcommand{\itfacts}[1]{\Cref{fact:it-facts}-(\ref{part:#1})\xspace}

\newcounter{mybox}

\crefname{mybox}{Box}{Boxes}

\newcommand{\prot}{\ensuremath{\pi}}
\newcommand{\suc}{\ensuremath{\textnormal{\textsf{suc}}}}
\newcommand{\round}{\ensuremath{\textnormal{\textsf{round}}}}
\newcommand{\bwidth}{\ensuremath{\textnormal{\textsf{bw}}}}

\newcommand{\cC}{\ensuremath{\mathcal{C}}}

\newcommand{\typeval}[1]{\ensuremath{type(#1)}}
\newcommand{\rtypeval}[1]{\ensuremath{\rv{type}(#1)}}

\newcommand{\xstar}{\ensuremath{x^{*}}}
\newcommand{\Xstar}{\ensuremath{X^{*}}}
\newcommand{\ystar}{\ensuremath{y^{*}}}
\newcommand{\Ystar}{\ensuremath{Y^{*}}}
\newcommand{\zstar}{\ensuremath{z^{*}}}
\newcommand{\Zstar}{\ensuremath{Z^{*}}}

\newcommand{\astar}{\ensuremath{a^*}}
\newcommand{\bstar}{\ensuremath{b^*}}
\newcommand{\cstar}{\ensuremath{c^*}}

\newcommand{\Nbtogether}[2]{\ensuremath{\mathcal{N}^{#1_{#2}}}}
\newcommand{\rNbtogether}[2]{\ensuremath{\bm{\mathcal{N}}^{#1_{#2}}}}

\newcommand{\Nbhood}[3]{\ensuremath{\mathcal{N}^{#1_{#3} \to #2}}}

\newcommand{\Ninner}[2]{\ensuremath{\mathcal{N}_{\textnormal{in}}^{#1_{#2}}}}
\newcommand{\rNinner}[2]{\ensuremath{\bm{\mathcal{N}}_{\textnormal{in}}^{#1_{#2}}}}

\newcommand{\cH}{\ensuremath{\mathcal{H}}}
\newcommand{\cD}{\ensuremath{\mathcal{D}}}
\newcommand{\cG}{\ensuremath{\mathcal{G}}}
\newcommand{\cGtil}{\ensuremath{\widetilde{\cG}}}
\newcommand{\cDreal}{\ensuremath{\cD^{\textnormal{real}}}}
\newcommand{\cDrealtil}{\ensuremath{\widetilde{\cDreal}}}
\newcommand{\cDfake}{\ensuremath{\cD^{\textnormal{fake}}}}

\newcommand{\cDmarg}{\ensuremath{\cD_{\textnormal{in}}}}

\newcommand{\idlayer}[1]{\ensuremath{id^{#1}}}
\newcommand{\ridlayer}[1]{\ensuremath{\rv{id}^{#1}}}

\newcommand{\id}[1]{\ensuremath{id(#1)}}

\newcommand{\ids}{\ensuremath{ids}}

\newcommand{\aux}{\ensuremath{aux}}
\newcommand{\raux}{\ensuremath{\rv{aux}}}

\newcommand{\cJ}{\ensuremath{\mathcal{J}}}
\newcommand{\cK}{\ensuremath{\mathcal{K}}}

\newcommand{\cJup}[1]{\ensuremath{\cJ^{#1}}}
\newcommand{\Jupdown}[2]{\ensuremath{J^{#1}_{#2}}}
\newcommand{\cKupdown}[2]{\ensuremath{\cK^{#1}_{#2}}}
\newcommand{\Kupdown}[2]{\ensuremath{K^{#1}_{#2}}}
\newcommand{\Lupdown}[2]{\ensuremath{L^{#1}_{#2}}}

\newcommand{\rcJup}[1]{\ensuremath{\bm{\cJ}^{#1}}}

\newcommand{\rcKupdown}[2]{\ensuremath{\bm{\cK}^{#1}_{#2}}}

\newcommand{\rLupdown}[2]{\ensuremath{\rv{L}^{#1}_{#2}}}

\newcommand{\cJall}{\ensuremath{\mathcal{J}_{\textnormal{all}}}}
\newcommand{\cKall}{\ensuremath{\mathcal{K_{\textnormal{all}}}}}
\newcommand{\Lall}{\ensuremath{L_{\textnormal{all}}}}

\newcommand{\rcJall}{\ensuremath{\bm{\mathcal{J}}_{\textnormal{all}}}}
\newcommand{\rcKall}{\ensuremath{\bm{\mathcal{K}}_{\textnormal{all}}}}
\newcommand{\rLall}{\ensuremath{\bm{L}_{\textnormal{all}}}}

\newcommand{\Npub}[1]{\ensuremath{\mathcal{N}_{\textnormal{pub}}^{#1}}}
\newcommand{\rNpub}[1]{\ensuremath{\bm{\mathcal{N}}_{\textnormal{pub}}^{#1}}}
\newcommand{\Mpub}[1]{\ensuremath{\mathcal{M}_{\textnormal{pub}}^{#1}}}
\newcommand{\rMpub}[1]{\ensuremath{\bm{\mathcal{M}}_{\textnormal{pub}}^{#1}}}

\newcommand{\Minner}[1]{\ensuremath{\mathcal{M}_{\textnormal{in}}^{#1}}}
\newcommand{\rMinner}[1]{\ensuremath{\bm{\mathcal{M}}_{\textnormal{in}}^{#1}}}

\newcommand{\vecx}{\ensuremath{\vec{x}}}

\newcommand{\msg}[3]{\ensuremath{\mathcal{M}^{#1_{#3} \to #2}}}

\newcommand{\msgrest}[2]{\ensuremath{\mathcal{M}^{\textnormal{out} \to #1_{#2} }}}
\newcommand{\msgall}{\ensuremath{\mathcal{M}_{\textnormal{all}}}}

\newcommand{\Nrest}[1]{\ensuremath{\mathcal{N}^{#1}_{\textnormal{rest}}}}
\newcommand{\rNrest}[1]{\ensuremath{\bm{\mathcal{N}}^{#1}_{\textnormal{rest}}}}

\newcommand{\rG}{\ensuremath{\rv{G}}}
\newcommand{\rid}[1]{\ensuremath{\rv{id}(#1)}}
\newcommand{\rids}{\ensuremath{\rv{ids}}}

\newcommand{\rNbhood}[3]{\ensuremath{\bm{\mathcal{N}}^{#1_{#3} \to #2}}}
\newcommand{\rmsgall}{\ensuremath{\bm{\mathcal{M}}_{\textnormal{all}}}}
\newcommand{\rmsg}[3]{\ensuremath{\bm{\mathcal{M}}^{#1_{#3} \to #2}}}
\newcommand{\rmsgrest}[2]{\ensuremath{\bm{\mathcal{M}}^{\textnormal{out} \to #1_{#2}}}}

\newcommand{\Euniq}{\ensuremath{\mathcal{E}_{\textnormal{uniq}}}}

\newcommand{\rIuniq}{\ensuremath{\rv{I}_{\textnormal{uniq}}}}

\newcommand{\rWcond}{\ensuremath{\rv{W}_{\textnormal{cond}}}}
\newcommand{\rWrest}{\ensuremath{\rv{W}_{\textnormal{rest}}}}

\newcommand{\rmsgup}[1]{\ensuremath{\bm{\mathcal{M}}^{#1}}}
\newcommand{\rNbhoodup}[1]{\ensuremath{\bm{\mathcal{N}}^{#1}}}

\newcommand{\rmsgnos}{\ensuremath{\bm{\mathcal{M}}}}

\newcommand{\Econd}{\ensuremath{\mathcal{E}_{\textnormal{cond}}}}
\newcommand{\Et}{\ensuremath{\mathcal{E}_t}}

\newcommand{\rw}{\rv{w}}
\newcommand{\rwstart}{\rv{w}^{\textnormal{start}}}
\newcommand{\rwend}{\rv{w}^{\textnormal{end}}}

\title{Distributed Triangle Detection is Hard in Few Rounds} 
\author{Sepehr Assadi\footnote{(sepehr@assadi.info) Supported in part by a  Sloan Research Fellowship, an NSERC
Discovery Grant, and a Faculty of Math Research Chair grant. \smallskip} 
\\ {\small University of Waterloo} \and
Janani Sundaresan\footnote{(jsundaresan@uwaterloo.ca) Supported in part by a Cheriton Scholarship from the School of Computer Science, Faculty of Math Graduate Research Excellence Award, and Sepehr Assadi's NSERC Discovery Grant. \smallskip} \\ {\small University of Waterloo}
}

\date{}

\begin{document}
\maketitle


\begin{abstract}

\medskip

In the distributed triangle detection problem, we have an $n$-vertex network $G=(V,E)$ with one player for each vertex of the graph who sees the edges incident on the vertex. The players communicate in synchronous rounds using the 
edges of this network and have a limited bandwidth of $O(\log{n})$ bits over each edge. The goal is to detect whether or not $G$ contains a triangle as a subgraph in a minimal number of rounds. 

\medskip

We prove that any protocol (deterministic or randomized) for distributed triangle detection requires $\Omega(\log\log{n})$ rounds of communication. Prior to our work, only 
one-round lower bounds were known for this problem. 

\medskip

The primary technique for proving these types of distributed lower bounds is via reductions from two-party communication complexity. However, it has been known for a while that this approach is provably incapable of establishing 
any meaningful lower bounds for distributed triangle detection. Our main technical contribution is a new information theoretic argument which combines recent advances on multi-pass graph streaming lower bounds 
with the point-to-point communication aspects of distributed models, and can be of independent interest.

\end{abstract}

\pagenumbering{roman}

\clearpage

\setcounter{tocdepth}{3}
\tableofcontents
\clearpage
\pagenumbering{arabic}
\setcounter{page}{1}



\section{Introduction}\label{sec:intro}

In the distributed triangle detection problem, we have an $n$-vertex graph $G=(V,E)$ representing a distributed network with one player for each vertex of the graph who sees the edges incident on the vertex. Players communicate in synchronous rounds 
by sending $O(\log{n})$ bit messages to each of their neighbors per round. The goal is to detect if $G$ contains a triangle as a subgraph in a small number of rounds. Detecting in this problem means that if $G$ contains no triangles, \emph{all} players should output \emph{No}; and, if $G$ contains a triangle, at least \emph{some} player should output \emph{Yes}. 

This model of distributed communication is referred
to as the CONGEST model~\cite{Peleg00} to contrast it with the LOCAL model~\cite{Linial92} that ignores bandwidth limitations and focuses solely on locality, imposed by communicating only over the edges of the graph. Triangle detection
can be trivially solved in a single round of LOCAL by each vertex collecting the neighborhood of its own neighbors. But, this approach incurs a significant communication bottleneck and misrepresents the true complexity of 
this problem in distributed networks. As a result, triangle detection (among other subgraphs) has been studied extensively in the CONGEST model in recent years, leading to the following state of affairs (see also~\Cref{sec:related} for more on the related work): 
\begin{itemize}
	\item From the upper bound side, the first algorithm for distributed triangle detection is due to~\cite{IzumiG17} and achieves $\Ot(n^{2/3})$ rounds. This was subsequently improved 
	in~\cite{ChangPZ19} to $\Ot(n^{1/2})$ rounds using expander decompositions, and further refined in~\cite{ChangS19,ChangPSZ21} to $\Ot(n^{1/3})$ rounds. See also \cite{HuangPZZ21} for an  algorithm with guarantees based on maximum
	 degree of the graph. These algorithms are randomized
	and their guarantees hold with high probability.
	Subsequent work ~\cite{CHLV22}, building on ~\cite{ChangS20}, obtained deterministic $n^{1/3+o(1)}$ round algorithms for this problem. Plugging in 
	the more recent deterministic expander routing algorithm of~\cite{ChangHS24} in the framework of~\cite{CHLV22}, leads to an $\Ot(n^{1/3})$ round deterministic algorithm. 
	
	
	\item On the lower bound side, not much is known about this problem. Basically, the only lower bounds for this problem are due to~\cite{AbboudCKL20} and~\cite{FischerGKO18} who, respectively, 
	showed that deterministic or randomized algorithms cannot solve this problem in a  \emph{single} round. 
\end{itemize}
We refer the reader to the excellent survey of~\cite{CH22} for a detailed overview of the literature. 

The lack of progress on lower bounds can be attributed to two factors: 
(1) the main lower bound technique in this model, namely reductions from
two-player communication complexity, is provably incapable of establishing any meaningful lower bounds for triangle detection~\cite{AbboudCKL20,FischerGKO18,CH22} (see also~\Cref{sec:congest-lower-background}); and, (2) there is a provable barrier for establishing 
strong lower bounds:~\cite{EdenFFKO22} shows that a lower bound of $n^{\delta}$ rounds for any constant $\delta > 0$ implies circuit complexity lower bounds that are entirely out of reach of existing techniques (see also~\Cref{sec:cktbounds}).

Given the above challenges, proving any lower bounds for distributed triangle detection has been considered a challenging task~\cite{AbboudCKL20,FischerGKO18,CH22}, and prior
work has only ruled out one-round algorithms. In this work, we make further progress on this tantalizing open question. 


\begin{result}\label{res:main}
	Any algorithm for distributed triangle detection in CONGEST on $n$-vertex graphs requires $\Omega(\log\log{n})$ rounds to succeed with high constant probability. 
\end{result}

\vspace{-5pt}

\Cref{res:main} provides a partial answer to the question of determining the round complexity of distributed triangle detection---from the lower bound side---raised repeatedly in the literature, e.g., 
in~\cite[Open Question 1]{AbboudCKL20} and~\cite[Open Problem 2.2]{CH22}, among others.

To put our main result in perspective, we should highlight other lower bounds for distributed triangle detection in related models. The lack of lower bounds in CONGEST motivated~\cite{AbboudCKL20} to consider a more limited model of 
communication wherein each player sends a single bit on each of its edges per round (simulating a round of CONGEST requires $\Omega(\log{n})$ rounds in this model). In this model,~\cite{AbboudCKL20}
proved an $\Omega(\log^*\!{n})$ lower bound for \emph{deterministic} algorithms and~\cite{FischerGKO18} extended the lower bound to $\Omega(\log{n})$ rounds. 
Similarly,~\cite{DruckerKO14} proved that any deterministic algorithm that broadcasts the \emph{same} $O(\log{n})$-bit message to all its neighbors in each round, needs $n^{1-o(1)}$ rounds\footnote{The lower bound of~\cite{DruckerKO14} also holds
	in the more general model of \emph{broadcast} CLIQUE wherein each player can send the same message to all vertices of the graph, not only its neighbors.}. To our knowledge, for randomized algorithms, 
even in these computationally weaker models, no lower bounds better than a single round were known. 


\paragraph{Technical perspective.} As mentioned earlier, a key bottleneck in proving lower bounds for distributed triangle detection is the inadequacy of standard applications of communication complexity. 
This is reminiscent of another (loosely) related area of research: \emph{multi-pass graph streaming lower bounds}, wherein direct applications of communication complexity are often limited (see,~\cite{AssadiCK19} or~\cite{Assadi23}
for a discussion of this topic). However, recent years have witnessed a flurry of breakthroughs in proving lower bounds in this model, e.g., in~\cite{GuruswamiO16,AssadiR20,ChenKPSSY21,AssadiKNS24} (see~\cite{Assadi23} for a quick summary) using more direct arguments tailored to the model and problems at hand. 

Our work is inspired by these successes and our conceptual contribution is to draw a connection between some of these techniques and the distributed triangle detection problem. At a technical level however, 
this requires bypassing the inherent difference between the point-to-point communication aspect of CONGEST model versus the broadcast nature of all these lower bounds in
graph streaming and related models (e.g.~\cite{AlonNRW15,AssadiKZ22}). Our main technical contribution is thus bridging this gap, which we hope can be of independent interest
for proving other distributed lower bounds in scenarios which are not amenable to standard communication complexity arguments.


\paragraph{Our techniques.} We shall go over our techniques in detail in the streamlined overview of our approach in~\Cref{sec:overview}. For now, we only mention the high-level bits of our techniques.

Our proof is based on \emph{round elimination}, specifically the types introduced in~\cite{AlonNRW15}\footnote{\label{footnote:welfare}Their work addresses the welfare maximization problem in mechanism design and bipartite matching in a broadcast communication model between (active) players on
	one side and (passive) items on the other side.}, and generalized more recently in~\cite{AssadiKZ22} and~\cite{AssadiKNS24,AssadiBKNS25} to prove lower bounds in
\emph{broadcast} CLIQUE and graph streaming models, respectively\footnote{We note that while conceptually related, this technique is entirely different from the round elimination in LOCAL lower bounds; see, e.g.,~\cite{Suomela20} for
	more details on the LOCAL round elimination technique.}. Prior work employ this technique by creating ``large'' $r$-round hard instances via combining many $(r-1)$-round hard ``small'' instances together, in a way that solving the large instance, requires solving one or a few of small instances as well. The construction ensures that
the identities of these few small instances are \emph{hidden} from the players in the first round. The analysis then shows that the first-round messages can be ``eliminated'', 
leaving the protocol with solving a hard (small) $(r-1)$-instance in $(r-1)$ rounds, which, inductively, is impossible. As we will argue later, direct applications of this approach inherently fail for us given CONGEST allows for a very large communication overall, and the complexity only comes from the limited and ``inconsistent'' view of each individual player. We thus show how to implement and analyze round elimination at a ``per player'' level instead of ``per (small) instance'' level. 

It is worth mentioning that graph streaming applications of these techniques in~\cite{AssadiKNS24,AssadiBKNS25} rely on complex combinatorial constructions to facilitate the lower bound arguments. In contrast, our hard input distributions are quite elementary from a combinatorial point of view and all the challenge is in the information-theoretic analysis of these distributions. 

\paragraph{Barriers for improvement.} As stated earlier, the barrier identified in~\cite{EdenFFKO22} suggests that, modulo breaking notoriously difficult circuit complexity barriers, 
the best lower bound we can hope for distributed triangle detection is some $n^{o(1)}$ rounds. This is precisely what our~\Cref{res:main} achieves, although at the ``lower end'' of the $n^{o(1)}$ range. 
A natural question now is whether we can improve our lower bound further to ``higher ends'' of $n^{o(1)}$; we provide some context for this question here.

A slight optimization of the parameters in \cite{EdenFFKO22} gives us that even an $n^{\Omega(1/\log \log n)}$ lower bound for distributed triangle detection breaks circuit complexity barriers; we present this proof in \Cref{sec:cktbounds} for completeness. Moreover, the proof of \cite{EdenFFKO22} relies on expander decompositions and routing schemes in CONGEST model. Expander decompositions now run in $\poly\!\log{\!(n)}$ rounds \cite{ChenMGS25} and it is \emph{plausible} that the current algorithms for expander routings can be improved dramatically to also run in $\poly\!\log{\!(n)}$ rounds. If such a fast algorithms for routing exists, even an $\Omega(\log^c{(n)})$-round lower bound for some large constant $c > 1$ for triangle detection could lead to breaking circuit complexity barriers (see \Cref{rem:cktbounds}). In light of all these, we believe that realistically, the best lower bound one can hope to prove for our problem is an $\Omega(\log{n})$ bound. We leave closing the gap of $\Omega(\log\log{n})$ in our~\Cref{res:main} to this ``realistic target'' of $\Omega(\log{n})$ as an interesting open question. 

Finally, we remark that our techniques, adapted from \cite{AssadiKNS24,AssadiBKNS25} (and originated from~\cite{AlonNRW15}), are heavily tailored to proving $\Omega(\log \log n)$ lower bounds -- in fact, 
both prior lower bounds of \cite{AssadiKNS24,AssadiBKNS25} are \emph{optimal} for their respective problems, and as such cannot be extended beyond these bounds. 
We believe proving a lower bound of $\Omega(\log n)$ for triangle detection in CONGEST likely requires new techniques.

\subsection{Related Work}\label{sec:related}

\paragraph{Listing vs detection.} A more general version of our problem is triangle \emph{listing} where the goal is that every triangle in the network is output by at least one player. 
All algorithms mentioned earlier for distributed triangle detection~\cite{IzumiG17,ChangPSZ21,ChangS20,CHLV22,ChangHS24}, with the exception of~\cite{IzumiG17}, 
also solve the listing problem within the same round complexity. As such, this problem can also be solved in~$\Ot(n^{1/3})$ rounds of CONGEST~\cite{ChangS19,ChangHS24}. 
In addition,~\cite{IzumiG17,PanduranganRS18} proved that this $\widetilde{\Omega}(n^{1/3})$ rounds is nearly-optimal for the listing problem. 
Their proof is based on a clever but relatively simple information-theoretic argument on random graphs which shows some vertex ``learns'' $\Omega(n^{4/3})$ bits
about the graph (outside its neighborhood) based solely on the answer; since each vertex 
can only receive $O(n\log{n})$ bits per round, this implies an $\Omega(n^{1/3}/\log{n})$ round lower bound. 
This technique is entirely disjoint from our work for triangle detection. 


\paragraph{CONGEST vs CLIQUE models.} The lower bounds of~\cite{IzumiG17,PanduranganRS18} hold in the more general (Congested) CLIQUE model~\cite{LotkerPPP05} 
wherein \emph{every} pair of players communicates $O(\log{n})$ bits per round. 
Interestingly, better CLIQUE algorithms are known for triangle \emph{detection}:~\cite{CensorHillelKK15} 
designed an $O(n^{1-2/\omega})$ rounds algorithm for triangle detection where $\omega$ is the matrix multiplication exponent %
which currently stands at $\omega \sim 2.371339$~\cite{AlmanDWXXZ25}. 
Unlike the lower bounds of~\cite{IzumiG17,PanduranganRS18} for triangle listing, proving \emph{any} super-constant round lower bounds for \emph{any} decision problem in CLIQUE---so, in particular, triangle detection---%
is a highly challenging task, as it implies circuit lower bounds that are beyond the reach of existing techniques~\cite{DruckerKO14}. 


CONGEST algorithms for triangle listing or detection in~\cite{ChangPSZ21,ChangS20,CHLV22,ChangHS24} all work by using expander decompositions and routings, 
to decompose the graph into well-connected components and simulate triangle \emph{listing} algorithms in the CLIQUE model on these components. The circuit complexity barrier of~\cite{EdenFFKO22} 
for proving triangle detection lower bound in CONGEST follows a similar pattern and relies on earlier work of~\cite{DruckerKO14} in the CLIQUE model. 


\paragraph{Other subgraphs.} There is also a large body of work on distributed subgraph detection beyond triangles. 
One example is the $4$-clique listing CONGEST algorithm of~\cite{EdenFFKO22} that runs in $n^{3/4+o(1)}$ rounds, which was subsequently improved to $\Ot(n^{1-2/p})$-round algorithms by~\cite{CensorHillelCG21,CHLV22,ChangHS24}
for all $p$-cliques for $p \geq 3$. In addition,~\cite{CzumajK18} proved a nearly-optimal $\Omgt(n^{1/2})$ lower on the round complexity of detecting $4$-cliques in CONGEST. Unlike for triangles, 
the lower bound for $4$-clique can be proven using the standard two-party communication complexity approach. 
We refer the interested reader to the excellent survey of~\cite{CH22} for a thorough review of the literature on distributed subgraph listing and detection, as well as more details on the 
uniqueness of \emph{triangles} among all other subgraphs when it comes to proving lower bounds. 


\paragraph{Multi-pass graph streaming lower bounds.} Reviewing the large body of work in this area is beyond the scope of this paper and 
we instead refer the reader to~\cite{GuruswamiO16,AssadiKSY20,AssadiR20,ChenKPSSY21,KolPSY23,ChenKPSSY23,AssadiS23,AssadiGLMM24,KonradN24,AssadiKNS24,AssadiKZ24,AssadiBKNS25} and references therein (see~\cite{Assadi23} for a short survey). But, we note that most relevant to us are the $\Omega(\log\log{n})$ pass lower bounds of~\cite{AssadiKNS24} and~\cite{AssadiBKNS25} for, respectively, maximal independent sets (in insertion-only streams) and approximate matchings (in dynamic streams), which are both also known to be \emph{optimal}. 

Two other related results---although not exactly in the streaming model---are the $\Omega(\log\log{n})$ round lower bounds of~\cite{AlonNRW15} for bipartite matching (see~\Cref{footnote:welfare})
and~\cite{AssadiKZ22} for maximal independent sets and matchings in \emph{broadcast} CLIQUE models (these results are \emph{not} known to be optimal and in fact the former one is improved to $\Omega(\log{n})$ rounds in~\cite{BravermanO17}); we discuss
 similarities and differences of the techniques in~\cite{AlonNRW15,AssadiKNS24,AssadiBKNS25} with ours in~\Cref{sec:overview}.



\newcommand{\eqcolor}[1]{\textcolor{blue}{#1}}

\section{Preliminaries}\label{sec:prelim}

\subsection{Notation}
We use $[n]$ to refer to the set $\{1, 2, \ldots, n\}$ for any $n \in\mathbb{N}$. For a tuple $x=(x_1,\ldots,x_n)$ and $i \in [n]$, we define $x_{<i}:=(x_1,\ldots,x_{i-1})$; we define 
$x_{>i}$ and $x_{-i}$ analogously. 

We use sans-serif font for random variables (for calligraphic letters, we instead use bold font to represent random variables). When it is clear from context, we may use random variables to refer to their distributions also. 
For random variables $\rA$ and $\rB$, we use $\en{\rA}$ and $\mi{\rA}{\rB}$ to denote the Shannon entropy of $\rA$ and mutual information between $\rA$ and $\rB$. Moreover, $\tvd{\rA}{\rB}$ denotes 
the total variation distance of $\rA$ and $\rB$.~\Cref{app:info} reviews information theory background we use. 

At certain places with consecutive lengthy equations, we may \eqcolor{color} some parts of the text to highlight the differences between nearby equations; these highlights are at no place necessary for parsing
the equations and can always be ignored entirely. 

\paragraph{Tripartite graphs.}
We work with tripartite graphs $G = (V, E)$ with vertex set partitioned into $V = A \sqcup B \sqcup C$, and $\card{A} = \card{B} = \card{C} = n$ and no edges with both end points in $A$ or $B$ or $C$. 
We use $A = \{a_1, a_2, \ldots, a_n\}$ to denote the elements of set $A$, and similarly for $B, C$. 
For any $i \in [n]$, we use $\{< a_i\}$ to refer to the set $\{a_1, a_2, \ldots, a_{i-1}\}$. This is defined analogously for $\{< b_i\}$ and $\{< c_i\}$. 

\noindent
\emph{\textbf{An important note on the notation:}} to avoid the repetition and clutter in the notation, we use $x, y, z$ to iterate over vertices of the graph without referring to each particular part of $A,B,C$. When we use $x \in \{a, b, c\}$, we use $X$ to refer to the set among $\{A,B,C\}$ that $x$ belongs to and vice-versa, and use $Y$ and $Z$ for the other two sets (e.g., if $x=a$, then $X=A$ and $Y \neq Z \in \set{B,C}$). 

Another example is to write $x_i$ for $x \in \set{a,b,c}$, to refer to $a_i$ or $b_i$ or $c_i$, correspondingly. That is, when, say, $x=a$, any other mention of $x_i$ or $\xstar_i$ should be interpreted as $a_i$ or $\astar_i$, respectively. 

\paragraph{Vectors.}
Given a vector $L$ of $\ell$ elements, we use $L[i]$ for $i \in [\ell]$ to refer to the $i^{\textnormal{th}}$ element of $L$. We use $L[S]$ for any $S \subseteq [\ell]$ to refer to all the elements at positions from $S$ in vector $L$.

\subsection{Model of Communication}\label{subsec:prelim-model}

For technical reasons, we work with a stronger model than CONGEST, wherein vertices are allowed to communicate to some select other vertices in certain rounds even without having an edge to them. Given we are proving a lower bound, 
this can only strengthen our result. 
Throughout the paper, we use $G=(V,E)$ to denote the input graph and use the terms `players' and `vertices' interchangeably.

\paragraph{Channels.} We consider communication as happening over \textbf{channels} between pairs of vertices. When we say a channel exists between $u, v \in V$ in some particular round, vertices $u$ and $v$ can send messages to each other in this round. 
Edge $(u,v)$ may not exist in the input graph, and the channels are a superset of the edges. The sets of channels available may vary between the rounds. 

Let $r$ be the total number of rounds in the protocol. 
We use $\cC_j \subseteq V \times V$ to denote the set of channels available at round $(r-j+1)$ for $j \in [r]$. The channels and the graph satisfy: 
\[
	\cC_{r} \supseteq \cC_{r-1} \supseteq \ldots \supseteq \cC_1  \supseteq \cC_0 :=  E,
\]
where recall that $E$ is the set of edges in the input graph (for notational convenience, it is easier to define the channels in this reverse order, e.g., have $\cC_{r}$ be the channels in round $1$ instead of $r$). 

\paragraph{Type of vertex pairs.}
To any pair of vertices $u,v \in V$, we assign a \textbf{type}, denoted by $type(u,v)$, which is an integer in $[r+1] \cup \{0\}$ to determine until when the channel is available for communication: 
\begin{itemize}
	\item If $(u, v) \in E$, then $type(u,v) = 0$; 
	\item If $(u, v) \in \cC_j \setminus \cC_{j-1}$ for some $j \in [r]$, then $type(u,v) = j$;
	\item If $(u, v) \notin \cC_r$, then $type(u,v) = r+1$ and $(u, v)$ is not a channel for any round in the protocol. 
\end{itemize}

The \textbf{channel-degree} of any vertex $u$ will denote the total number of other vertices $v$ such that $\typeval{u, v} < r+1$, i.e., the total number of vertices $v$ such that $(u, v) \in \cC_r$.

\paragraph{Sources of randomness.}
Any protocol has three distinct sources of randomness:

\begin{itemize}
	\item \underline{Public} randomness: this is a global source of random bits visible to all vertices. 
	\item \underline{Pair} randomness: for each pair $u,v \in V$, there is a tape of random bits visible only to the two vertices $u$ and $v$.  This is independent of whether $(u,v)$ form any channel in the graph. 
	\item \underline{Private} randomness: each vertex $u \in V$ has a private tape of random bits only visible to itself. 
\end{itemize}

\paragraph{Input of players.} We work with tripartite graphs with a \emph{known} partition of vertices into three parts $A,B,C$, each of size $n$ (players also know $n$). The vertices in each layer 
are identified with elements from $V = A \sqcup B \sqcup C$; we refer to this element as the \textbf{identity} of the vertex. 
Each vertex $x \in X$ is given two vectors of length $n$, one for each of layers; in the vector for a layer $Y$, for each $y \in Y$, the type of pair $(x,y)$ which is an integer in $[r+1] \cup \set{0}$ is specified. 
Note that this uniquely identifies the input graph and all the communication channels. 

\paragraph{Communication in a protocol.} Communication happens in rounds. In round $i \in [r]$, each vertex $u \in V$ simultaneously sends messages to all $v \in V$ such that $(u,v) \in \cC_{r+1-i}$
 i.e., in the first round the vertices can send messages over all pairs $(u,v) \in \cC_r$, in the second round, over $\cC_{r-1}$, and so on (as $\cC_0 = E$, the players
 can always communicate over the edges of $G$). 

\paragraph{Output in a protocol:} 
When no triangle exists in the input graph $G$, we want \emph{all} vertices to declare that no triangle exists at the end of the last round. When a triangle exists, we want \emph{at least one} vertex to declare that a triangle exists 
(the vertex may not be part of any triangle). 

\paragraph{Protocol parameters.} 
There are three relevant parameters associated with a protocol $\prot$:
\begin{itemize}
	\item $\round(\prot)$: the total number of rounds used by protocol $\prot$. 
	\item $\bwidth(\prot)$: the bandwidth of the protocol, or the maximum of the number of bits sent by any vertex over any channel, taken over all rounds (in CONGEST, we have $\bwidth(\prot)=O(\log{n})$).
	\item $\suc(\prot)$: the minimum success probability of the protocol $\prot$ over all input instances where the probability is taken over all the sources of randomness. 
\end{itemize}
For deterministic protocols $\prot$, we also have the following parameter for any distribution $\mu$ of inputs:
\begin{itemize}
	\item $\suc(\prot, \mu)$: probability of success of $\prot$ on average over inputs drawn from the distribution $\mu$. 
\end{itemize}

Our model is \textbf{more powerful than the CONGEST} model, as any protocol $\prot$ in CONGEST can be implemented in our model by only 
communicating messages over channels $\cC_0=E$ (and ignoring other channels). It is also worth pointing out---even though this will not be related to our paper---that this model \emph{in the limit} can even capture the CLIQUE model by allowing a channel between \emph{all} pairs of vertices throughout all the rounds (i.e., $\cC_r = \ldots = \cC_1 = {{V}\choose{2}}$ and $\cC_0=E$ as before). Although for our lower bounds, we certainly do not allow all channels to exist.



\newcommand{\GG}{\ensuremath{\mathcal{G}}}
\newcommand{\jstar}{\ensuremath{j}^{\star}}

\newcommand{\rM}{\rv{M}}
\newcommand{\rN}{\rv{N}}

\newcommand{\rI}{\rv{I}}

\newcommand{\Astar}{\ensuremath{A^*}}
\newcommand{\Bstar}{\ensuremath{B^*}}
\newcommand{\Cstar}{\ensuremath{C^*}}

\newcommand{\Gstar}{\ensuremath{G^*}}

\newcommand{\rL}{\rv{L}}
\newcommand{\rK}{\rv{K}}
\newcommand{\rW}{\rv{W}}

\newcommand{\rtype}[1]{\rv{type}(#1)}

\section{Technical Overview}\label{sec:overview}

Our proof of~\Cref{res:main} is quite technical and involves various information-theoretic maneuvers that can be unintuitive or daunting to parse. Thus, we 
use this section to unpack our main ideas and give a streamlined overview of our approach. We emphasize that this section oversimplifies many details and the discussions will be informal for the sake of intuition.
The rest of the paper is written in an independent way so the reader can skip this part and directly jump to technical arguments.

We start with two background subsections on: (1) distributed lower bounds in CONGEST via communication complexity and (2) 
round elimination arguments in~\cite{AlonNRW15,AssadiKZ22,AssadiKNS24,AssadiBKNS25} to review these techniques and be able to pinpoint the obstacles 
in extending them for our purpose. A reader familiar with these works can safely skip these subsections. We then move to the main part of this section that reviews our own approach in establishing~\Cref{res:main}. 

\subsection{Background I: CONGEST Lower Bounds via Communication Complexity}\label{sec:congest-lower-background}

Communication complexity provides us with a wide array of tools for proving lower bounds in distributed computing and beyond. 
There are however two aspects that differentiate the CONGEST model from typical scenarios wherein one can apply communication complexity arguments directly: 
\vspace{-15pt}
\begin{itemize}[leftmargin=10pt]
	\item \textbf{Point-to-point communication}: the communication bottleneck in CONGEST model is \emph{not} based on a limit on the \emph{total} communication of players: the total communication across all vertices in each round
	is quite large and is sufficient to specify the \emph{entire} input of all players to an external party that sees all the messages. Instead, the bottleneck is imposed by the \emph{individual} view of each player who can only receive 
	limited information from each of its own neighbors. Thus, here, the players not only have an inconsistent view of the input, but also, the communicated messages (this is in sharp contrast with all broadcast 
	or ``shared blackboard'' communication models used for establishing similar lower bounds in other models). 
	\item \textbf{Input-sharing}: there is a limited but non-negligible degree of ``input sharing'' between the players in this model as each edge of the input is seen by two players. This makes
	this model to move one (small) step away from the friendly number-in-hand communication models (with no input sharing) towards the notorious number-on-forehead models (with arbitrary input sharing). We 
	refer the reader to~\cite{NelsonY19,AssadiKO20,AssadiKZ22} that discuss this aspect in more detail. 
\end{itemize}
\vspace{-5pt}

Starting from~\cite{PelegR00}, the vast majority of lower bounds in CONGEST have found an interesting way that simultaneously addresses both these aspects and allows for obtaining reductions from two-party communication complexity 
lower bounds. They work by carefully splitting the input graph $G=(V,E)$ between two \emph{induced} subgraphs $G_A$ and $G_B$ with only \emph{few} edges between $G_A$ and $G_B$. They then show that a 
\emph{hard} two-party communication problem (almost always \emph{set disjointness}) can be embedded in the edges of $G_A$ and $G_B$, while keeping the few edges between $G_A$ and $G_B$ fixed and input-independent. 
Assuming one needs $t$ bits of communication for solving the two-party communication problem and there are only $k$ edges between $G_A$ and $G_B$, this construction immediately implies 
an $\Omega(t/(k\log{n}))$ round lower bound for the original problem; see~\Cref{fig:alice-bob}.

This way, the point-to-point communication aspect is reduced to the overall communication between vertices of $G_A$ on one side and vertices of $G_B$ on the other side (via the few edges between them). Input-sharing is bypassed entirely 
because the only shared edges are now fixed and input-independent. We refer the reader to~\cite{AbboudCKP21} for many successful applications of this technique. 

\begin{figure}[t!]
	\centering
	\includegraphics[scale=0.30]{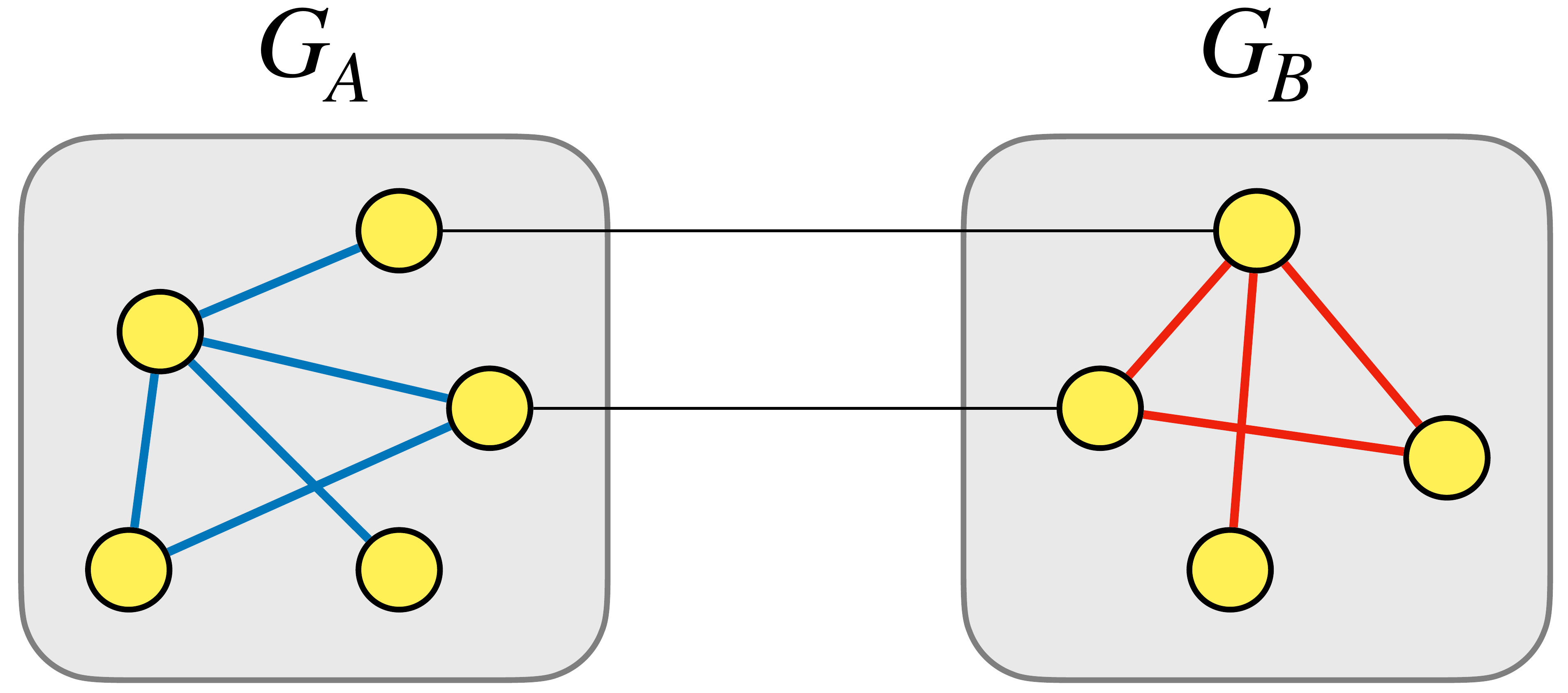}
	\caption{An illustration of CONGEST lower bounds via two-party communication complexity. The edges in the left subgraph $G_A$ (resp. right subgraph $G_B$) are known only to Alice (resp. Bob), and 
		the edges between the two subgraphs are input-independent and fixed. Alice and Bob can run any CONGEST algorithm on this graph by Alice simulating vertices in $G_A$ and Bob doing the same for $G_B$. The only communication between Alice and Bob 
		is the messages of the CONGEST algorithm that cross the middle (input-independent) edges. Thus, any $r$-round CONGEST protocol over a graph with $k$ edges between $G_A$ and $G_B$, implies
		a communication protocol with $O(r \cdot k \cdot \log{n})$ communication.}\label{fig:alice-bob}
\end{figure}

Unfortunately, it is easy to see that this approach is incapable of addressing triangle detection: no matter what graph we use and how we partition it, as long as two vertices of any
triangle are handled by a single player, that player gets to see all edges of the triangle and solve the problem with no communication. As such, any lower bound for distributed triangle detection 
effectively needs to handle each vertex as its own individual player and cannot shortcut the aforementioned unique aspects of the CONGEST model. This is perhaps why distributed triangle detection 
is considered ``one of the best illustrations of the lack and necessity of new techniques for proving lower bounds in distributed computing''~\cite{AbboudCKL20}. Indeed, the only existing lower bounds for this problem
in~\cite{AbboudCKL20,FischerGKO18} for one-round algorithms precisely target the problem in this way. 

\subsection{Background II: Round Elimination Arguments in Prior Work}\label{sec:re-background} 

We now switch to discussing the type of round elimination ideas that have proven quite successful in establishing distributed lower bounds in certain broadcast models~\cite{AlonNRW15,AssadiKZ22} and more recently graph streaming lower bounds~\cite{AssadiKNS24,AssadiBKNS25}\footnote{We shall note that origins of these ideas also directly traces back to the two-party communication complexity model and the ``round elimination lemma'' of~\cite{MiltersenNSW95} in that model.}. While there are substantial differences between applications of these ideas across these works, at a sufficiently high level (and for the purpose of comparison with our own work), we can characterize all of them as follows. 

We inductively create a family of distributions $\set{\GG_r}_{r \geq 0}$ where $\GG_r$ generates hard inputs for $r$-round algorithms (with the base case $\GG_0$ being simply some non-trivial distribution as $0$-round algorithms 
are typically very easy to analyze directly). For any $r \geq 1$, the distribution $\GG_r$ is constructed by sampling a ``large'' number of $(r-1)$-round independent instances $I_1,I_2,\ldots,I_{k_r}$ from $\GG_{r-1}$. 
These instances are then ``packed'' together carefully to ensure that for some randomly chosen $\jstar \in [k_r]$, (1) the answer to the instance $I$ sampled from $\GG_r$ is determined by the answer of the inner instance $I_{\jstar}$ from $\GG_{r-1}$, and yet, (2) in the
first round of the protocol, the players are oblivious to the value of $\jstar$. We note that this way, often the inputs and number of players in $r$-round instances in $\GG_r$ are polynomially larger than those in $(r-1)$-round instances. 

Suppose we have an $r$-round protocol $\prot_r$ for solving instances $I \sim \GG_r$ with success probability $\suc({\prot_r})$. We construct an $(r-1)$-round protocol $\prot_{r-1}$ 
for solving instances $I' \sim \GG_{r-1}$ using $\prot_r$: 
\begin{tbox}
	\begin{enumerate}
		\item The players of $\prot_{r-1}$ use \underline{public randomness} to sample an index $\jstar \in [k_r]$ as well as the first round of messages $M^{(1)}$ of $\prot_r$ from its distributions. They assume the inner 
		$(r-1)$-round instance $I_{\jstar}$ of the outer $r$-round instance $I$ is their input $I'$. 
		\item They then use a combination of \underline{private and public randomness} to sample the remainder of $I$ conditioned on $M^{(1)}$ and $I_{\jstar} = I'$ to create a complete $r$-round instance\footnote{This step
			is highly non-trivial and specialized as each player of $I'$ only knows its own input in $I'$ and not the entire $I'$ and so it is not clear (nor always possible) for the players to sample $I$ conditioned on $M^{(1)}$ and $I_{\jstar}=I'$.}. 
		\item Finally, the players in $\prot_{r-1}$ already have access to the first-round message $M^{(1)}$ and thus only need to run $\prot_r$ from its second round onwards---while simulating remaining parts of $I$ outside $I_{\jstar}=I'$---and
		output the same answer as $\prot_r$. 
	\end{enumerate}
\end{tbox}

To argue the correctness of the protocol $\prot_{r-1}$, we only need to show that the distribution of inputs and messages to $\prot_r$ induced by $\prot_{r-1}$ is within $o(1)$ total variation distance of the correct distribution;
the rest follows by construction, since the answer to $I \sim \GG_r$ is the same as $I_{\jstar}$ which is chosen to be $I'$, namely, the instance $\prot_{r-1}$ wants to solve in the first place. Specifically, 
\begin{itemize}
	\item The \emph{right} distribution of all random variables in $\prot_r$, for a fixed $\jstar$, can be expressed as:
	\[
	\rM^{(1)} \times \eqcolor{\paren{\rI_{\jstar} \mid \rM^{(1)}}} \times \paren{\rI_{-\jstar} \mid \rM^{(1)},\rI_{\jstar}}, 
	\]
	where $\rI_{-\jstar}$ denotes the random variable for all $(r-1)$-round sub-instances except for $I_{\jstar}$. 
	\item The distribution sampled from in the protocol $\prot_{r-1}$ on the other hand is: 
	\[
	\underbrace{\rM^{(1)}}_{\textnormal{publicly}} \times \underbrace{\eqcolor{\rI_{\jstar}}}_{\textnormal{input}} \times \underbrace{\paren{\rI_{-\jstar} \mid \rM^{(1)},\rI_{\jstar}}}_{\textnormal{mix of publicly and privately}},
	\]
	namely, here, $M^{(1)}$ and $I_{\jstar}$ are sampled independently of each other. 
\end{itemize}
Assuming the third step of sampling can be done---which, to emphasize again, contains the bulk of efforts in many of these works ---the distance between the two distributions is only a function of their second arguments. 
This distance, for a random choice of $\jstar \in [k_r]$, can be upper bounded as 
\[
\Exp_{\jstar \in [k_r]}\tvd{\rI_{\jstar}}{(\rI_{\jstar} \mid \rM^{(1)})}^2 \leq \Exp_{\jstar \in [k_r]}\mi{\rI_{\jstar}}{\rM^{(1)}} \leq \frac{1}{k_r} \cdot \mi{\rI_1,\ldots,\rI_{k_r}}{\rM^{(1)}} = o(1);
\]
here, the first inequality is standard (see~\Cref{fact:pinskers} and~\Cref{fact:kl-info}) and the main inequality is the second one: it holds roughly because the players in $\prot_r$ (but of course not $\prot_{r-1}$) are 
\emph{oblivious} to the choice of $\jstar$ in their first round and thus the information revealed by their first-round messages is ``spread'' over all $k_r$ sub-instances. The final equality also holds
by ensuring that $k_r$ is much larger than the length of sampled messages $M^{(1)}$ (and applying~\itfacts{uniform}). 

All in all, this means protocol $\prot_{r-1}$ has success probability $\suc(\prot_{r-1},\GG_{r-1}) \geq \suc(\prot_r,\GG_r)-o(1)$ and a similar communication bandwidth as $\prot_r$, but with one fewer round. 
Continuing like this then leaves us with a $0$-round protocol for $\GG_0$ with a non-trivial probability of success, a contradiction. 

Before moving on, we should caution the reader that while the above discussion captures the common theme of arguments across the aforementioned work, none of those work follow this recipe directly. For instance,~\cite{AlonNRW15} samples
the message $M^{(1)}$ to be only the ones originating from players in $I_{\jstar}$ (to ensure its size is small with respect to $k_r$);~\cite{AssadiKZ22} does not even sample the remainder of $I$ and follows
the simulation through sampling relevant messages from $I_{-\jstar}$ to $I_{\jstar}$; and finally, for~\cite{AssadiKNS24,AssadiBKNS25}, the number of players is a lot fewer (in fact, just two players in~\cite{AssadiBKNS25}) and 
instead the players are forced to solve many albeit a small fraction of inner sub-instances at the same time. This discussion also only focused on the information-theoretic aspects of the lower bounds and entirely 
neglected the combinatorial parts of how one can ``pack'' so many $(r-1)$-round independent instances inside a single $r$-round instance. However, we hope our description provides a big picture of these prior arguments.

\subsection{Our Approach}\label{sec:our-ideas}

We now start presenting the main ideas behind our own work. 

\subsubsection{A Hard Input Distribution}

We generate a recursive family of hard input distributions $\set{\GG_r}_{r \geq 0}$ where $\GG_r$ samples tripartite graphs with $n_r$ vertices in each layer for $r$-round protocols. 
For reasons that will become clear soon, unlike in~\Cref{sec:re-background}, we cannot specify our $r$-round instances as a combination of many $(r-1)$-round instances. Instead, they are obtained by sampling a 
single $(r-1)$-round instance plus a very large individualized ``noise'' for each vertex in the inner $(r-1)$-round instance to hide its inner edges among the outer graph. 
 
 We use parameters $n_r$ and $d_r$ in the description of the distribution and they are defined recursively as follows\footnote{We specify these parameters explicitly in the actual proof, but here keep them
 	at this level to focus on their connections as opposed to their actual values.}:
 \begin{align}
 	d_r = \poly(n_{r-1}) \qquad \text{and} \qquad n_{r} = \poly(d_r).  \label{eq:overview-para}
 \end{align}
 
The base case $\GG_0$ has $n_0 = 1$ and simply consists of three vertices $a,b,c$ with an edge between each pair appearing with probability $1/2$ (and thus a triangle with probability $1/8$).
The distribution for $r$-rounds for $r>0$ is as follows. 

\begin{Distribution}\label{dist:dist-overview}
	A hard distribution $\GG_r$ of graphs $G=(A \sqcup B \sqcup C , E)$ for $r$-round protocols:  
	\begin{enumerate}
		\item Sample three sets of vertices $\Astar \subseteq A, \Bstar \subseteq B$, and $\Cstar \subseteq C$ each of size $n_{r-1}$ independently. 
		\item Sample an $(r-1)$-round instance $\Gstar$ over $(\Astar,\Bstar,\Cstar)$ and let the induced subgraph of $G$ on these vertices be $\Gstar$. The vertices in different layers have a channel in $\cC_r$ to each other, whose type 
		is determined by the channel-type in $\Gstar$ (not having a channel in $\Gstar$ translates to a type $r$ in $G$). 
		\item For any $x$ in some layer $\Xstar$, any layer $Y \neq X \in \set{A,B,C}$, and any type $t \in [r] \cup \set{0}$ of channel, sample enough vertices for $v$ uniformly from $Y \setminus \Ystar$
		such that $x$ has $d_r$ type-$t$ channels in total to $Y$. We require all sampled vertices to be unique across all vertices of $\Gstar$. 
	\end{enumerate}
\end{Distribution}
See~\Cref{fig:hard-dist} for an illustration.

\begin{figure}[t!]
	\centering
	\includegraphics[scale=0.20]{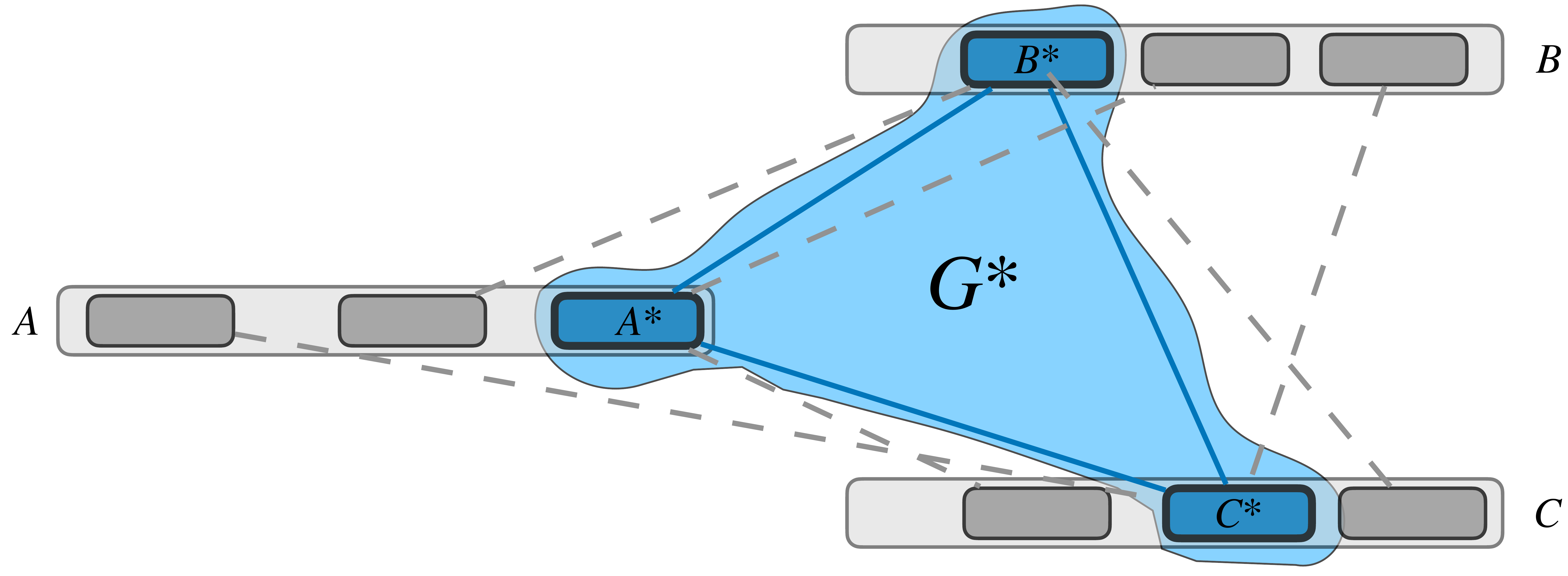}
	\caption{An illustration of our hard input distribution. The solid (blue) edges show the $(r-1)$-round hard instance inside an $r$-round hard instance and the dashed (gray) edges show the extra unique channels sampled
		for each vertex of the inner $(r-1)$-round instance. }\label{fig:hard-dist}
\end{figure}

It is easy to see that a graph $G$ sampled from $\GG_r$ for $r \geq 1$ has a triangle iff its inner graph $\Gstar$ has one. Thus, the task of players in solving an $r$-round instance $G$
reduces to that of solving the inner $(r-1)$-round instance $\Gstar$. Moreover, while a player in $G$ can determine whether it belongs to $\Gstar$ or not (its channel-degree is only $1$ in the latter case), 
it cannot determine this information about its neighbors on its own. Intuitively, this suggests that the messages between two players $x,y \in \Gstar$ cannot reveal almost any information
about $\Gstar$ as either of these vertices have $> d_r$ channels, but only $n_{r-1} \ll d_r$ of them are in $\Gstar$, and their identities are unknown to the sender. Hence, one expects the first
round messages to be ``wasted'' until the players find their channels in $\Gstar$, and then have to solve a hard $(r-1)$-round instance in the remaining $(r-1)$ rounds. 

As expected, turning this intuition into an actual proof is quite challenging. For instance, our distribution $\GG_1$ for one-round protocols is almost identical to the ones for the one-round lower
bound of~\cite{FischerGKO18} (modulo the notion of channels) and based on the same principle as the deterministic one-round lower bound of~\cite{AbboudCKL20}.  But, even for one-round protocols, quite a lot of technical work has been done in~\cite{FischerGKO18,AbboudCKL20}
to establish the lower bound. As such, all our effort in this paper is dedicated to formalizing this basic intuition. It is also worth pointing out that while many one-round lower bounds can be cast as round elimination arguments to zero-round instances, this 
is \emph{not} the case for the arguments in~\cite{FischerGKO18,AbboudCKL20}, and thus from a technical point of view, our arguments already deviate from these prior approaches even for one-round protocols. 

\subsubsection{First Attempt on Round Elimination: Public Sampling of Messages}\label{sec:attempt1}
Following prior work in~\Cref{sec:re-background}, let us consider creating an $(r-1)$-round protocol $\prot_{r-1}$ for $G_{r-1} \sim \GG_{r-1}$ from an $r$-round protocol $\prot_r$ for $G_r \sim \GG_r$: 
\begin{tbox}
	\begin{enumerate}
		\item The players of $\prot_{r-1}$ use \underline{public randomness} to sample vertices $(A^*,B^*,C^*)$ and map $[n_{r-1}]$, namely, vertices of their input $G_{r-1}$ in each layer,
		to these sets accordingly. They then define the induced subgraph $G^*$ in $G_r$ to be $G_{r-1}$ after this mapping of vertices. 
		
		\item The players of $\prot_{r-1}$ sample the first-round messages $M^{(1)}$ of $\prot_r$ using \underline{public randomness} conditioned on $(\Astar,\Bstar,\Cstar)$ but independent of the types of channels in $\Gstar$.  \\
		...
	\end{enumerate}
	We will already run into a serious problem!
\end{tbox}

For this approach to have any chance of succeeding, we need to be able to say that the joint distribution of $(M^{(1)},\Gstar) \mid (\Astar,\Bstar,\Cstar)$ in the protocol $\prot_{r-1}$ is close to being a product distribution, so that 
sampling them independently in $\prot_{r-1}$ does not change their distribution too much. But certainly distribution of $M^{(1)}$ highly correlates with that of $G^*$: for any channel in $G^*$, 
the message of $M^{(1)}$ can easily reveal the type of this channel with very small communication (and if we consider the original CONGEST model, then simply having a message between a pair of vertices reveals the existence of the edge
between them as well!). Thus, the protocol above is doomed to fail. 

This issue stems from the inherent difference of point-to-point communication versus broadcast ones targeted in~\Cref{sec:re-background}. 
Basically, sampling the message $M^{(1)}$ publicly is effectively the same as revealing the first round messages to everyone (in the second round); 
but unlike a broadcast model, for us, this is way too much information. Fortunately, to simulate a CONGEST protocol $\prot_r$, each vertex only needs to know the messages
communicated to and from it and not all messages. Hence, the \textbf{first lesson} for our round elimination protocol is that sampling first-round messages should ``respect'' the 
point-to-point communication pattern as well. 

\subsubsection{Second Attempt on Round Elimination: Pair Sampling of Messages}\label{sec:attempt2}
Equipped with our previous lesson, we make a second attempt at creating the protocol $\prot_{r-1}$. In addition to public randomness,  we have equipped the vertices with pair randomness, to facilitate the sampling of first round messages. See \Cref{sec:prelim} for the different sources of randomness given to the vertices.
\begin{tbox}
	\begin{enumerate}
		\item The players of $\prot_{r-1}$ use \underline{public randomness} to sample vertices $(A^*,B^*,C^*)$ and define $\Gstar$ of $G_r$ based on their input $G_{r-1}$, as before. 
		
		\item Then, every pair of vertices in $x,y \in G^*$, use \underline{pair randomness} to sample the messages to and from each other, namely, $M^{(1)}_{x \rightarrow y}$ and $M^{(1)}_{y \rightarrow x}$, 
		conditioned on $(\Astar,\Bstar,\Cstar)$ and $type(x,y)$ but independent of the rest of $G^*$. 
		\\
		...
	\end{enumerate}
	Again, we stop the description of the protocol here already. 
\end{tbox}

Let us examine the two distributions involved here: 
\begin{itemize}
	\item The right distribution of involved variables is: 
	\[
	\paren{\rA^*,\rB^*,\rC^*} \times \paren{\rG^* \mid \rA^*,\rB^*,\rC^*} \times \bigtimes_{(x,y) \in \rX^{*} \neq \rY^*} \paren{\rM^{(1)}_{x,y} \mid \eqcolor{\rM^{(1)}_{<(x,y)}, \rG^*}, \rA^*,\rB^*,\rC^*}. 
	\]
	where $\rM^{(1)}_{x,y} = (\rM^{(1)}_{x \rightarrow y},\rM^{(1)}_{y \rightarrow x})$ denotes both messages communicated over a pair $(x,y)$, and, under some arbitrary ordering between
	all pairs of vertices $(x,y) \in \Gstar$, variable $\rM^{(1)}_{<(x,y)}$ collects these messages for all pairs before $(x,y)$.
	\item On the other hand, the distribution of variables in $\prot_{r-1}$ are: 
	\[
	\underbrace{\paren{\rA^*,\rB^*,\rC^*}}_{\textnormal{public randomness}} \times \underbrace{\paren{\rG^* \mid \rA^*,\rB^*,\rC^*}}_{\textnormal{input plus renaming}} \times 
	\bigtimes_{(x,y) \in \rX^{*} \neq \rY^*} \underbrace{\paren{\rM^{(1)}_{x,y} \mid \eqcolor{\rG^*_{x,y}}, \rA^*,\rB^*,\rC^*}}_{\textnormal{pair randomness}}, 
	\]
	where $\rG^*_{x,y}$ is the random variable for $type(x,y)$ (conditioned on vertices $(\Astar,\Bstar,\Cstar)$). 
\end{itemize}
Bounding the difference between these distributions thus boils down to bounding the following for every pair $(x,y) \in \Gstar$ (we use $G^*_{-(x,y)}$ to denote the other types in $G $ apart from $(x,y)$ and we say $\rN^* = (\rA^*,\rB^*,\rC^*)$ to avoid the clutter):
\begin{align*}
	\tvd{\paren{\rM^{(1)}_{x,y} \mid \rM^{(1)}_{<(x,y)}, \rG^*, \rN^*}}{\paren{\rM^{(1)}_{x,y} \mid \rG^*_{x,y}, \rN^*}}^2 &\leq \mi{\rM^{(1)}_{x,y}}{\rG^*_{-(x,y)},\rM^{(1)}_{<(x,y)} \mid  \rG^*_{x,y}, \rN^*}. \tag{by~\Cref{fact:pinskers} and~\Cref{fact:kl-info}} \\
\end{align*}

\vspace{-5mm}

\noindent
Using the chain rule of mutual information (\itfacts{chain-rule}), 
we can further write the RHS above as 
\begin{align}
	\mi{\rM^{(1)}_{x,y}}{\rG^*_{-(x,y)},\rM^{(1)}_{<(x,y)} \mid  \rG^*_{x,y}, \rN^*} =  \mi{\rM^{(1)}_{x,y}}{\rG^*_{-(x,y)} \mid  \rG^*_{x,y},\rN^*} + \mi{\rM^{(1)}_{x,y}}{\rM^{(1)}_{<(x,y)} \mid \rG^*, \rN^*}. \label{eq:overview-1} 
\end{align}
Let us now consider each term in the RHS separately. 

\paragraph{First term in RHS of~\Cref{eq:overview-1}.} This term measures the information revealed by the messages between $x,y$ about \emph{other} pairs inside $\Gstar$ (conditioned on $type(x,y)$). Our intuition is that this information
should be quite low because vertices $x$ and $y$ cannot distinguish between their channels in or out of $G^*$, and thus should not be able to tell each other about the type of their other edges inside the graph $\Gstar$. 

At this point, it is tempting to follow a similar strategy as in~\Cref{sec:re-background} and re-write the distribution $\GG_r$ so that it is a combination of, say, $k_r$ many $(r-1)$-round instances. We can then say
that given $x$ participates in $k_r$ many instances and only one of them is special (namely, is $\Gstar$), $x$ will not be able to reveal much information about this instance in its message. The problem is again the unicast versus broadcast: 
if the partitioning of the input into separate instances are known to the players, even if they do not know which one is special, we still have a problem: vertex $x$ can reveal a lot of information to its neighbor $y$ about the \emph{same} instance
they are both part of. 

To bypass this, we instead use an \emph{inconsistent} way of creating these instances. Suppose in the distribution $\GG_r$, 
for the vertex $x \in \Gstar$ we sample $k_r$ \emph{partial} $(r-1)$-round sub-instances $K_1,\ldots,K_{k_r}$ on $2n_{r-1}-1$ vertices  such that each one is sampled from a distribution that combined with the edge $(x,y)$ can form the input of $x$ in a $\GG_{r-1}$ instance. 
See~\Cref{fig:partial} for an illustration. 

\begin{figure}[h!]
	\centering
	\includegraphics[scale=0.25]{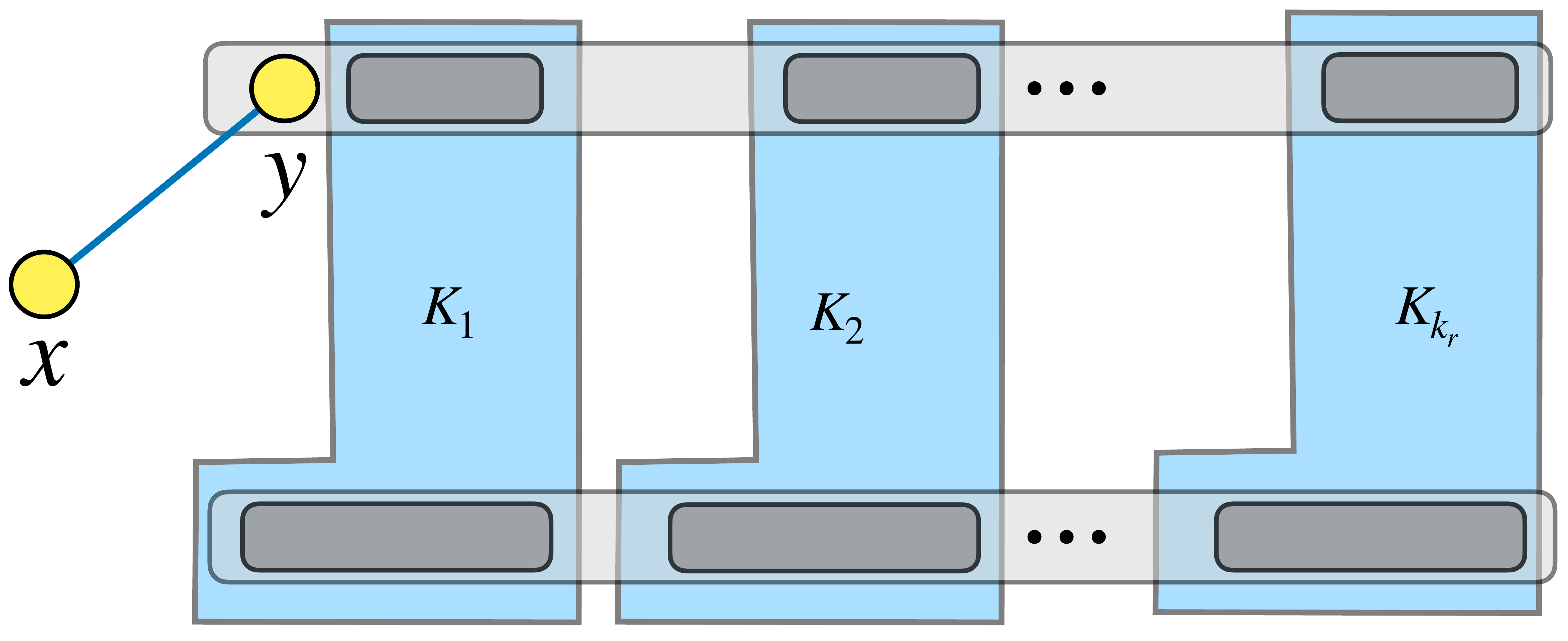}
	\caption{An illustration of sampling partial $(r-1)$-round sub-instances. Here, from the perspective of $x$, the edge $(x,y)$ plus any of $K_i$'s for $i \in [k_r]$ form a proper input of $x$ in an instance sampled from $\GG_{r-1}$. 
		One of these instances corresponds to the special inner instance $G^*$, while for all other instances, the instance is not complete, as in there are no edges between $Y$- and $Z$-part of the instance.
	}\label{fig:partial}
\end{figure}

By setting $k_r = d_r/(n_{r-1})^2$, we can show that for inputs sampled from $\GG_r$, with high probability, we can find such sub-instances for every vertex $x$ and any of the $r$ possible types of channels in $G^*$ (basically, $x$ has sampled $d_r$
many channels of each type so one can re-order them to create $k_r$ many partial sub-instances). Moreover, from the perspective of $x$ in protocol $\prot_r$, 
the special part of the input, namely, $G^*_{-(x,y)}$ can be any of these $k_r$ many partial sub-instances. This, with some more technical but standard ideas that we omit here, allows us to essentially write:  
\[
\mi{\rM^{(1)}_{x,y}}{\rG^*_{-(x,y)} \mid  \rG^*_{x,y},\rN^*} \leq \mi{\rM^{(1)}_{x,y}}{\rG^*_{-(x,y)} \mid \rK_1,\ldots,\rK_{k_r},\rG^*_{x,y},\rN^*}  \leq \frac{1}{k_r} \cdot \bwidth(\prot_r), 
\]
by effectively arguing that the information revealed by $\rM^{(1)}_{x \rightarrow y}$ is spread over $k_r$ many possible choices for $\rG^*_{-(x,y)}$ among $K_1,\ldots,K_{k_r}$ from the perspective of $x$ (and similarly for $\rM^{(1)}_{y \rightarrow x}$ and $y$). 
Here, with an abuse of notation, conditioning on $\rK_1,\ldots,\rK_{k_r}$ means conditioning on the existence of $\cC_r$ channels to them, but the type of those channels are not fully determined (only that they are sampled from the marginal distribution of $\GG_{r-1} \mid type(x,y)$). 

Given the choice of parameters in~\Cref{eq:overview-para}, across all choices of $x,y \in G^*$, the total contribution of the first term of RHS of~\Cref{eq:overview-1} to the distances between distributions is
\[
\paren{n_{r-1}}^2 \cdot \frac{1}{k_r} \cdot \bwidth(\prot_r) = \frac{({n_{r-1})}^4}{d_r} \cdot \bwidth(\prot_r) \leq \frac{1}{\poly(n_{r-1})} \cdot \bwidth(\prot_r) \leq 1/\poly(n_{r-1}). 
\]

\paragraph{Second term in RHS of~\Cref{eq:overview-1}.} This term measures the information revealed by the messages between $x,y$ about some of \emph{other messages}, 
conditioned on the entire special sub-instance $\Gstar$. This is where we run into another serious problem. 

Consider a protocol wherein every vertex $x$ samples a random coin privately and sends its value as part of the 
message to all its neighbors. In such a protocol, knowing the message $\rM^{(1)}_{x,y}$ for any $y$ reveals (at least) one bit of information about $\rM^{(1)}_{x,z}$ for any other $z \neq y$ as well, making the RHS
of~\Cref{eq:overview-1} at least $1$ which is way too large for our purpose. 

Thus, our second attempt at designing protocol $\prot_{r-1}$ also fails, although this time we made some further progress. The \textbf{second lesson} is that 
we need to break the correlation between messages originating from a single vertex to other vertices in $\Gstar$ (but certainly not all of $G$). 

\subsubsection{Third Attempt on Round Elimination: Public and Pair Sampling of Messages}\label{sec:attempt3}

We now use yet another sampling process for obtaining $\prot_{r-1}$ from $\prot_r$. 
To simplify the exposition, we only write this sampling from the perspective of a single pair $(x,y) \in \Gstar$ and only the message $M^{(1)}_{x \rightarrow y}$ from $x$ to $y$, 
\emph{assuming} other messages $M^{(1)}_{x \rightarrow -y}$ of $x$ to $\Gstar \setminus \set{y}$ have already been (somehow) sampled using their own pair randomness. 
Doing this for all pairs of vertices has several additional challenges but we skip those in this discussion to provide the high level intuition (see also~\Cref{fig:K-L}). 
\begin{tbox}
	\begin{enumerate}
		\item The players of $\prot_{r-1}$ use \underline{public randomness} to sample vertices $(A^*,B^*,C^*)$ and define $\Gstar$ of $G_r$ based on their input $G_{r-1}$, as before.
		We also assume messages $M^{(1)}_{x \rightarrow -y}$ have been sampled using \underline{pair randomness} at this point (so, in particular, are unknown to $y$). 
		\item Player $x$ uses \underline{public randomness} to sample 
		\[
		K_1,\ldots,K_{k_r} \subseteq{Y \cup Z},\qquad \qquad   L \subseteq Y, \qquad \text{and} \qquad M^{(1)}_{x}[L \cap \{<id(y)\}];
		\]
		\begin{itemize}[leftmargin=5pt]
			\item $K_i$ is a disjoint set of $(2n_{r-1}-1)$ vertices in $G$, which together with $(x,y)$ can form an $(r-1)$-round instance from $\GG_{r-1}$; we let $x$ have $\cC_r$-channels to all these vertices 
			but do not sample their types yet (also, value of $k_r$ will be determined later);
			\item $L$ is a disjoint set of $\ell_r$ (to be determined later) random vertices in $Y$, each of which has a $type(x,y)$-channel to $x$; 
			\item $M^{(1)}_{x}[L \cap \{<id(y)\}]$ is the messages $x$ sends to its neighbors in $L \cap \{<id(y)\}$ in $\prot_r$ conditioned on all the \emph{publicly} sampled information ($id(y) \in \Ystar$ is the id assigned to $y \in G_{r-1}$). 
		\end{itemize}
		\item Finally, $x$ and $y$ sample $M^{(1)}_{x \rightarrow y}$ via \underline{pair randomness} conditioned on $type(x,y)$ and all the \emph{publicly} sampled variables (so specifically, independent of $M^{(1)}_{x \rightarrow -y}$). 
	\end{enumerate}
\end{tbox}

\begin{figure}[t!]
	\centering
	\includegraphics[scale=0.25]{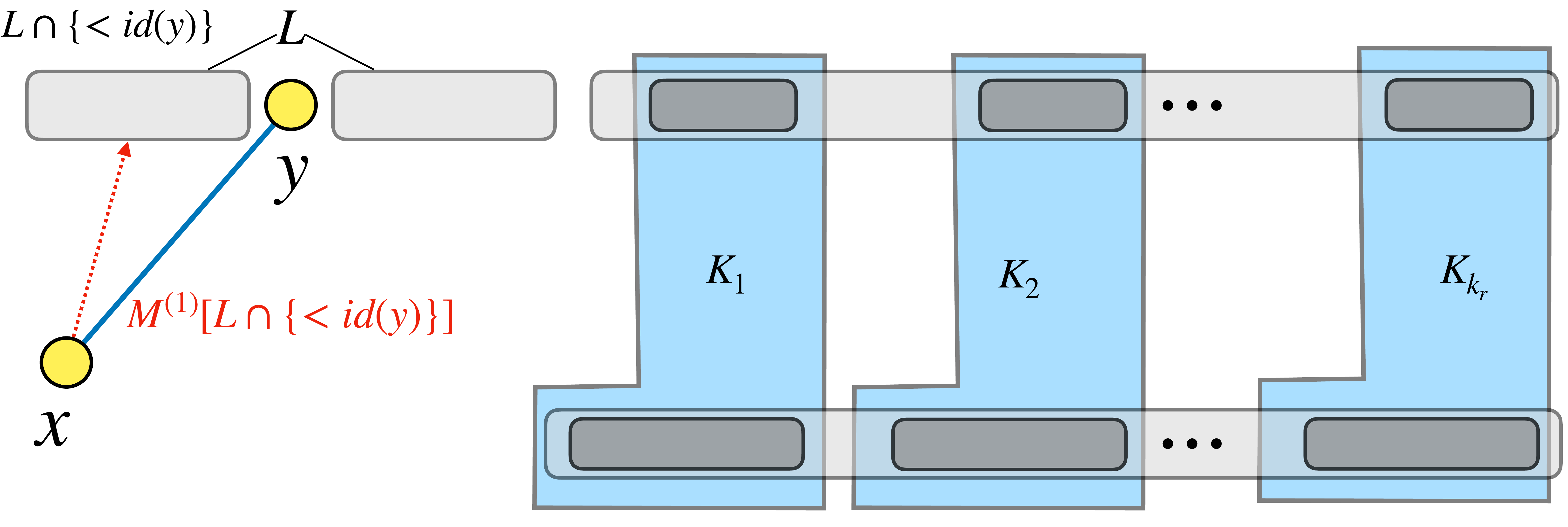}
	\caption{An illustration of sampling messages via a mixture of public and pair randomness. The type of the channel between $(x,y)$ any $(x,y')$ for any $y' \in L$ is the same.}\label{fig:K-L}
\end{figure}

As for the parameters $\ell_r$ and $k_r$, we pick them such that 
\begin{align}
	\ell_r = \poly(n_{r-1}), \qquad \qquad k_r = \poly(\ell_r), \qquad \text{and} \qquad d_r = \poly(k_r). \label{eq:overview-para2}
\end{align}
\vspace{-20pt}

We note that the entire goal of these sampling steps is to generate an input $G_r$ and first-round messages $M^{(1)}$ for the protocol $\prot_r$ so that players in $\prot_{r-1}$ can run this protocol on the sampled input. 
The above process only specifies some part of the input to one of the players $x$, with respect to one of the pairs $(x,y)$ in the original input of $\prot_{r-1}$. We will need to repeat this process for all pairs of 
vertices in $\prot_{r-1}$ and complete the remainder of inputs and messages of players (including $x$) before we can run $\prot_r$. However, we can already start analyzing the difference between variables involved in this process 
and the actual distribution of $\GG_r$ and $\prot_r$. In the following, it helps to think of $y$ as a vertex in the original input $G_{r-1}$ and $id(y)$ as 
the vertex $y$ maps to in $G_r$. 
\vspace{-7pt}
\begin{itemize}
	\item The right distribution of involved variables is (we define $\rN^* := (\rA^*,\rB^*,\rC^*)$, $\rK := (\rK_1,\ldots,\rK_{k_r})$, and $\rW := (\rN^*,\rK,\rL)$ to reduce the clutter): 
	\[
	\hspace{-20pt}	\paren{\rW,\rG^*,\rM^{(1)}_{x \rightarrow -y}} \times \paren{\rM^{(1)}_{x}[\rL \cap \{<\rid{y}\}] \mid \eqcolor{\rW,\rG^*,\rM^{(1)}_{x \rightarrow -y}}} \times \paren{\rM^{(1)}_{x \rightarrow y} \mid \rM^{(1)}_{x}[\rL \cap \{<\rid{y}\}] , \rW,\eqcolor{\rG^*_x,\rM^{(1)}_{x \rightarrow -y}}},
	\]
	where  $\rG^*_x$ in the last term denotes types of all channels incident on $x$ to vertices in $G^*$. We note that technically, here, we should have also
	conditioned on all of $\rG^*$ and not only $\rG^*_x$. However, it is easy to see that the distribution of messages sent by $x$ is independent of the rest of $\Gstar$ once we condition on $\Gstar_x$. This is because $\Gstar_x$ provides
	the input to $x$, thus conditioning on $\rG^*_x$ or $\rG^*$ is the same. 
	\item And, the distribution of variables in $\prot_{r-1}$ is: 
	\[
	\hspace{-20pt}	\underbrace{\paren{\rW,\rG^*,\rM^{(1)}_{x \rightarrow -y}}}_{\textnormal{input plus public and pair randomness}} \times \underbrace{\paren{\rM^{(1)}_{x}[\rL \cap \{<\rid{y}\}] \mid \eqcolor{\rW}}}_{\textnormal{public randomness}} \times 
	\underbrace{\paren{\rM^{(1)}_{x \rightarrow y} \mid \rM^{(1)}_{x}[\rL \cap \{<\rid{y}\}] , \rW,\eqcolor{\rG_{x,y}^*}}}_{\textnormal{pair randomness}}
	\]
	where $\rG_{x,y}^*$ (after conditioning on $\rN^*$) will simply be the type of the pair $(x,y)$. 
\end{itemize}
\vspace{-5pt}
Let us focus on the third terms for now. Bounding the difference between these terms (as before using~\Cref{fact:pinskers} and~\Cref{fact:kl-info}) boils down to upper bounding 
\begin{align}
	&\hspace{60pt}\mi{\rM^{(1)}_{x \rightarrow y}}{\rG^*_{x,-y},\rM^{(1)}_{x \rightarrow -y} \mid \rW,\rM^{(1)}_{x}[\rL \cap \{<\rid{y}\}],\rG^*_{x,y}} = \tag{$\rG^*_{x,-y}$ is $\rG^*_x$ minus $\rG^*_{x,y}$} \\
	&\hspace{-5pt}\mi{\rM^{(1)}_{x \rightarrow y}}{\rG^*_{x,-y} \mid \rW,\rG^*_{x,y},\rM^{(1)}_{x}[\rL \cap \{<\rid{y}\}]}  + \mi{\rM^{(1)}_{x \rightarrow y}}{\rM^{(1)}_{x \rightarrow -y} \mid \rW,\rM^{(1)}_{x}[\rL \cap \{<\rid{y}\}] ,\rG^*_x},  \label{eq:overview-2}
\end{align}
where the equality is a direct application of chain rule (\itfacts{chain-rule}). Notice that the LHS here is quite similar to the LHS of~\Cref{eq:overview-1}: modulo the extra conditioning on some new random variables, 
LHS of~\Cref{eq:overview-2} also measures the correlation of message of $x$ to $y$ with its remaining input $\Gstar_{x,-y}$ as well as $x$'s other messages (in~\Cref{eq:overview-1} also, we could have replaced $G_{-(x,y)}$ by $G_{x,-y}$). 
While we (provably) could have not bound the 
RHS of~\Cref{eq:overview-1} (in particular its second term), we can bound the RHS of~\Cref{eq:overview-2} with the help of these extra conditionings. 

Before getting to the sketch of the proof, it is worth briefly checking the protocol outlined at the end of the last subsection when discussing the second term of~\Cref{eq:overview-1}. For that specific protocol, 
the moment we condition on any message of $x$ in $M^{(1)}_{x}[L \cap\{<id(y)\}]$ (as in the RHS of~\Cref{eq:overview-2}), we already have fixed the random coin toss of the protocol across its different messages and 
thus can at least hope to say messages $M^{(1)}_{x \rightarrow y}$ and $\rM^{(1)}_{x \rightarrow -y}$ have low correlation. 

\paragraph{First term in RHS of~\Cref{eq:overview-2}.} We have 
\begin{align*}
	&\mi{\rM^{(1)}_{x \rightarrow y}}{\rG^*_{x,-y} \mid \rW,\rG^*_{x,y}, \rM^{(1)}_{x}[\rL \cap\{<\rid{y}\}]} \\
	&\hspace{30pt} = \mi{\rM^{(1)}_{x \rightarrow y}}{\rG^*_{x,-y} \mid \rW \setminus \rK,\textcolor{blue}{\rK_1,\ldots,\rK_{k_r}},\rG^*_{x,y},  \rM^{(1)}_{x}[\rL \cap \{<\rid{y}\}]} \tag{as $\rW$ contains $\rK = (\rK_1,\ldots,\rK_{k_r})$} \\
	&\hspace{30pt} \leq ~~\ddots \tag{skipping some technical steps} \\
	&\hspace{30pt} \leq \frac{1}{k_r} \cdot \mi{\rM^{(1)}_{x \rightarrow y}}{\eqcolor{\rK_1,\ldots,\rK_{k_r}} \mid \rW \setminus \rK,\rG^*_{x,y},  \rM^{(1)}_{x}[\rL \cap \{<\rid{y}]}  \tag{following~\Cref{sec:attempt2} as $\rG^*_{x,-y}$ could be any of $\rK_1,\ldots,\rK_{k_r}$ from the perspective of $x \mid \Gstar_{x,y}$} \\
	&\hspace{30pt} \leq \frac{1}{k_r} \cdot \bwidth(\prot_r) \tag{by~\itfacts{uniform}}, 
\end{align*}
given that $\rM^{(1)}_{x \rightarrow y}$ is message with at most $\bwidth(\prot_r)$ bits. By the choice of $k_r$ in~\Cref{eq:overview-para2}, 
we can bound the contribution of this term, across all $x,y \in \Gstar$, by $1/\poly(n_{r-1})$. 

\paragraph{Second term in RHS of~\Cref{eq:overview-2}.} We have, 
\begin{align*}
	\mi{\rM^{(1)}_{x \rightarrow y}}{\rM^{(1)}_{x \rightarrow -y} \mid \rW,\rG^*_{x}} &= \mi{\rM^{(1)}_{x \rightarrow y}}{\rM^{(1)}_{x \rightarrow -y} \mid \rM^{(1)}_{x}[\rL \cap\{ <\rid{y}\}], \textcolor{blue}{\rW_{-}, \rL, \rid{y}},\rG^*_{x}} \tag{by defining $\rW = (\rW_{-},\rL,\rid{y})$ since $\rid{y}$ 
		is fixed in $\rY^*$} \\
	&=  \mi{\rM^{(1)}_{x \rightarrow y}}{\rM^{(1)}_{x \rightarrow -y} \mid \rW_{-}, \rM^{(1)}_{x}[\rL \cap \{<\rid{y}\}], \textcolor{blue}{\rL \cup \rid{y}, \rid{y}}, \rG^*_{x}} \tag{since $\rid{y}$ is disjoint from $\rL$ and thus $(\rL \cup \rid{y}, \rid{y}) \equiv  (\rL, \rid{y})$}. 
\end{align*}
But now the good part is that the choice of $\rid{y}$ is uniform over $\rL \cup \rid{y}$, conditioned on all other variables; in particular, since all neighbors of $x$ in $L$ have the same type as the $(x,y)$ edge, this 
is still true even conditioned on $\rG^*_{x,y}$. In other words, from the perspective of the vertex $x$ in the protocol $\prot_r$ (but certainly not $\prot_{r-1}$), the edge $(x,y)$ is ``just another'' one of its edges of this type
among the set $L \cup id(y)$. Thus, when it is sending the message $M^{(1)}_{x \rightarrow -y}$, it cannot particularly correlate it with a specific message to vertices in $L$ as opposed to all messages to $L$. 
Using this intuition and an 
application of chain rule (\itfacts{chain-rule}) allows us to bound the RHS above as
\begin{align*}
	&\mi{\rM^{(1)}_{x \rightarrow y}}{\rM^{(1)}_{x \rightarrow -y} \mid \rW_{-}, \rM^{(1)}_{x}[\rL \cap \{<\rid{y}\}], {\rL \cup \rid{y}, \textcolor{blue}{\rid{y}}}, \rG^*_{x}} \\
	&\hspace{30pt}= \Exp_{y_j \in [L \cup id(y)]} \mi{\textcolor{blue}{\rM^{(1)}_{x}[y_j]}}{\rM^{(1)}_{x \rightarrow -y} \mid \rW_{-}, \rM^{(1)}_{x}[\rL \cap \textcolor{blue}{\{<y_j\}}], {\rL \cup \rid{y}}, \rid{y}=y_j, \rG^*_{x}} \tag{think of $y_j$ here 
		as going over all vertices in $L \cup id(y)$ and choosing which one is $id(y)$} \\
	&\hspace{30pt}\leq ~~\ddots \tag{skipping some technical arguments to drop the conditioning on $\rid{y}=y_j$} \\
	&\hspace{30pt} = \frac{1}{\ell_r+1} \cdot \mi{\textcolor{blue}{\rM^{(1)}[\rL \cup \rid{y}]}}{\rM^{(1)}_{x \rightarrow -y} \mid \rW_{-}, {\rL \cup \rid{y}}, \rG^*_{x}} \tag{by chain rule in~\itfacts{chain-rule}} \\
	&\hspace{30pt} \leq \frac{1}{\ell_r+1} \cdot (n_{r-1}-1) \cdot \bwidth(\prot_r) \tag{by~\itfacts{uniform}}, 
\end{align*}
given that $\rM^{(1)}_{x \rightarrow -y}$ contains $(n_{r-1}-1)$ messages with $\bwidth(\prot_r)$ bits each. 
Finally, by the choice of $\ell_r$ in~\Cref{eq:overview-para2}, we can bound the contribution of this term, across all $x,y \in \Gstar$, by $1/\poly(n_{r-1})$. 

\medskip

\paragraph{Back to the distributions.} We are not done yet however as we only compared the third terms in the distributions of variables induced by $\prot_{r-1}$ versus their actual distribution in $\prot_r$. 
We now need to compare the second terms also, namely, bound 
\[
\mi{\rM^{(1)}_{x}[\rL \cap \{<\rid{y}\}]}{\rG^*,\rM^{(1)}_{x \rightarrow -y} \mid \rW}, 
\]
which is the information revealed by some of the messages of $x$ about the inner graph $G^*$ (and some other messages). But this seems to bring us to the very beginning: we again 
have a collection of publicly sampled messages and need to bound their correlation with the inner special sub-instance. The \textbf{key observation} here is that we moved away from messages communicated 
over channels of $\Gstar$ and now we can indeed hope that at least messages $M^{(1)}_{x}[L]$ sampled for $x$ (and eventually other vertices in $\Gstar$) which are going to channels \emph{outside} of $G^*$ may not be 
too correlated with $\Gstar$. So, while the current protocol does not handle this part, we can hope to fix this issue using an argument similar to~\Cref{sec:re-background} which, unlike in~\Cref{sec:attempt1}, is no 
longer doomed to fail.

\subsubsection{The Final Attempt: Our Round Elimination Protocol} 

At this point, hopefully, we have provided enough intuition about the need for the rather peculiar sequence of sampling of variables in our final round elimination protocol. Describing this protocol at the level of details of previous subsections 
requires us to provide another lengthy set of definitions. Instead, equipped with the lessons and observations of previous subsections, we only state our final protocol $\prot_{r-1}$ at a very high-level as follows.  

\paragraph{Public randomness.} The players of $\prot_{r-1}$ use \underline{public randomness} to sample 
\begin{itemize}
	\item $(\Astar,\Bstar,\Cstar)$ defines the vertices of the inner graph $\Gstar$ and allows players to map their input $G_{r-1}$ to this induced subgraph; 
	\item $(J^x_1,\ldots,J^x_{j_r})$ are {full} $(r-1)$-round instances sampled from $\GG_{r-1}$ for each vertex $x$ of $G^*$ (similar in spirit to~\Cref{sec:re-background} but now for each individual vertex);
	\item $(K^{x,y,t}_1,\ldots,K^{x,y,t}_{k_r})$ are {partial} $(r-1)$-round instances that for each pairs of vertices $x,y \in \Gstar$ are sampled from $(\GG_{r-1} \mid type(x,y)=t)$\footnote{\label{footnote} Only $x$ and $y$ know $type(x,y)$ and so to sample these
		sets using public randomness, we instead need to sample the sets for all possible types from $[r] \cup \{0\}$ and then let $x,y$ pick the right type for themselves.}  so that together with the type of $(x,y)$ they can form a full $(r-1)$-round instance (similar to~\Cref{sec:attempt2}); 
	\item $(L^{x,t}_1,\ldots,L^{x,t}_{\ell_r})$ are sets of neighbors for each vertex $x \in \Gstar$ and a fixed type $t \in [r] \cup \set{0}$;  (similar to~\Cref{sec:attempt3}); 
	\item $M_{\textnormal{public}}$ contains publicly sampled messages to subsets of $(L^{x,t}_1,\ldots,L^{x,t}_{\ell_r})$ from each $x \in \Gstar$ and type $t \in [r] \cup \set{0}$  (similar to~\Cref{sec:attempt3}). 
\end{itemize}
\paragraph{Pair randomness.} Each pair of vertices $(x,y) \in \Gstar$ samples the messages $M^{(1)}_{x,y}$ conditioned on $type(x,y)$ and all publicly sampled variables (similar to~\Cref{sec:attempt2} and~\Cref{sec:attempt3}). 

\paragraph{Private randomness.} By the above steps, each vertex $x \in \Gstar$ has sampled many of its channels in $G$. It then simply samples any remaining channels so that it has $d_r$ channels for each type in each layer, conditioned
on all the publicly sampled variables and all variables sampled with pair randomness that are known to $x$ (i.e., are related to channels incident on $x$). 

The parameters of the protocol are chosen as follows (notice that the dependencies are acyclic): 
\[
\ell_r = \poly(n_{r-1}), \quad k_r = \poly(\ell_r), \quad j_r = \poly(k_r), \quad d_r = \poly(j_r), \quad n_r = \poly(d_r).  
\]
These variables then, using a combination of ideas outlined in the previous subsections, allow us to
\begin{enumerate}
	\item $L$-variables and $M_{\textnormal{public}}$: break the correlation between messages $M^{(1)}_{x,y}$ sampled by pair randomness for all pairs $(x,y) \in \Gstar$; 
	\item $K$-variables: break the correlation between messages $M^{(1)}_{x,y}$ sampled by pair randomness and input graph $\Gstar_{-(x,y)}$ for all pairs $(x,y) \in \Gstar$; 
	\item $J$-variables: break the correlation between $M_{\textnormal{public}}$ sampled by public randomness and $\Gstar$. 
\end{enumerate}

\paragraph{Simulation step.} We are not yet done because we also need to handle messages and inputs of players in $G$ that are outside $\Gstar$. While for some other problems, this can be quite challenging (see, e.g.~\cite{AssadiKZ22} and their
``partial-input embedding'' technique for handling this), for us this step is quite straightforward\footnote{This also means our techniques and~\cite{AssadiKZ22}, besides their starting point in~\cite{AlonNRW15}, are almost entirely disjoint: all our efforts in this paper is in the handling of ``special inner'' vertices whereas in~\cite{AssadiKZ22}, this part is done exactly as in~\cite{AlonNRW15} and instead their main focus is on the
	remaining vertices.}. Any vertex $x$ in $G \setminus \Gstar$ only has a single channel by construction and this channel is to a vertex $y \in \Gstar$; thus, in our protocol, 
either the channel $(x,y)$ has been sampled publicly or vertex $y$ has sampled this channel privately. In either case, $y$ knows the entire input of $x$ and since $x$ only communicates with $y$, vertex $y$ can completely simulate the work of $x$ on its own by sampling
its message also. 

This means that players of $\prot_{r-1}$ on $G_{r-1}$ have all the information needed to simulate the protocol $\prot_r$ on the graph $G$ they have collectively created, use the sampled messages $M^{(1)}$ for the first round of $\prot_r$ without themselves
communicating at all, and then communicating messages of $\prot_r$ in its remaining $(r-1)$ rounds as part of their own $(r-1)$ rounds of communication. This way, and using our analysis for the sampling steps, we can prove that $\prot_{r-1}$ satisfies
\[
\suc(\prot_{r-1},\GG_{r-1}) \geq \suc(\prot_r,\GG_r) - 1/\poly(n_{r-1}),
\]
using only $(r-1)$ rounds and a similar bandwidth as $\prot_r$. 

Continuing like this allows us to obtain a $0$-round protocol for $\GG_0$ with success probability much better than $7/8$ (which is easy to show is the optimal bound on the distribution $\GG_0$ for $0$-round protocols), a contradiction. 
Finally, in terms of parameters, given that the size of instances grows polynomially from $\GG_{r-1}$ to $\GG_r$, we have 
\[
n_r \geq 2^{2^{\Omega(r)}}
\]
which means an $r$-round lower bound on $\GG_r$ translates to an
\[
r = \Theta(\log\log{n_r})
\]
lower bound as desired by~\Cref{res:main}.

\clearpage


\section{A Hard Distribution and its Properties}\label{sec:hard-dist}
We start our formal technical proofs from this section. Throughout the proofs, we use a large number of 
 random variables, and~\Cref{sec:list-rv} lists them to help the reader in keeping track of them. 
We also provide a schematic organization of our proofs in~\Cref{sec:schematic-org}. 

Throughout, we use $\cG_r(n_r)$ to denote the hard distribution for $r$-rounds over tripartite graphs with $n_r$ vertices in each layer. In this section, we describe our hard distribution and some of its properties, 
and prove the following theorem (minus its technical details) that formalizes~\Cref{res:main}. 
\begin{theorem}\label{thm:hard-dist}
	There exists a family of distributions $\{\cG_r(n_r)\}_{r \geq 0}$ over tripartite graphs with $n_r$ vertices in each layer, such that for $n_r > r^{4 \cdot 34^r}$, any deterministic protocol $\prot_r$ with,
	\[
		\round(\prot_r) = r \qquad \textnormal{ and } \qquad \suc(\prot_r, \cG_r(n_r)) \geq 15/16,
	\]
	must have, \[
		\bwidth(\prot_r) \geq (n_r)^{(1/2) \cdot (1/ 34^r)} \cdot (480)^{-2}.
	\]
\end{theorem}

Let us see how the hardness of distribution $\cG_r(n_r)$ proves \Cref{res:main}. 
\begin{proof}[Proof of \Cref{res:main}]
By the easy direction of Yao's minimax principle (i.e, an averaging argument), it is sufficient to argue that for deterministic protocols using $o(\log \log n)$ rounds, $O(\log n)$ length messages are not sufficient to solve triangle detection with probability of success at least $15/16$ when input graphs are drawn from distribution $\cG_r(n_r)$. 

In \Cref{thm:hard-dist}, when $r = o(\log \log n_r)$, we know that the initial condition of $n_r > r^{4 \cdot 34^r}$ is satisfied. The bandwidth required, however, is,
\[
	 (n_r)^{(1/2) \cdot (1/ 34^r)} \cdot (480)^{-2} \gg \poly \log{(n_r)}, 
\]
which is more than the $O(\log n)$ bandwidth allotted in the CONGEST model.
\end{proof}

We begin by describing the hard distribution for $0$-round protocols, and we recursively construct hard distributions for $r$-rounds using the distribution for $(r-1)$-rounds. We also talk about the joint distribution of the input and first round messages to set us up for the analysis. 

\subsection{Base Case: Hard Distribution for $0$-rounds}
In this subsection, we will describe the hard distribution for $0$-rounds, and prove some of its necessary properties. 
By a $0$-round protocol, we mean that there is no communication and the players must output whether a triangle exists just based on their own input. This will serve as the base case for our inductive proof of \Cref{thm:hard-dist}. 
\begin{Distribution}\label{dist:0}
	\textbf{Distribution $\cG_0(n_0)$ over graphs $G$ for $0$-round protocols:} 
	
	(See \Cref{fig:0-round} for an illustration.)
	\begin{enumerate}[label=$(\arabic*)$]
		\item Start with a tripartite graph with $n_0$ vertices in each layer, and with vertex set $V = A \sqcup B \sqcup C$, and $A = \{a_1, a_2, \ldots, a_{n_0}\}, B = \{b_1, \ldots, b_{n_0}\}, C = \{c_1, \ldots, c_{n_0}\}$. 
		\item Choose three random vertices $\astar \in A$, $\bstar \in B$ and $\cstar \in C$. 
		\item For all three pairs of $(\astar,\bstar), (\bstar,\cstar), (\cstar,\astar)$, add an edge between each pair with probability $1/2$ uniformly at random and independently. 
	\end{enumerate}
\end{Distribution}

\begin{figure}[t!]
	\centering
	\includegraphics[scale=0.25]{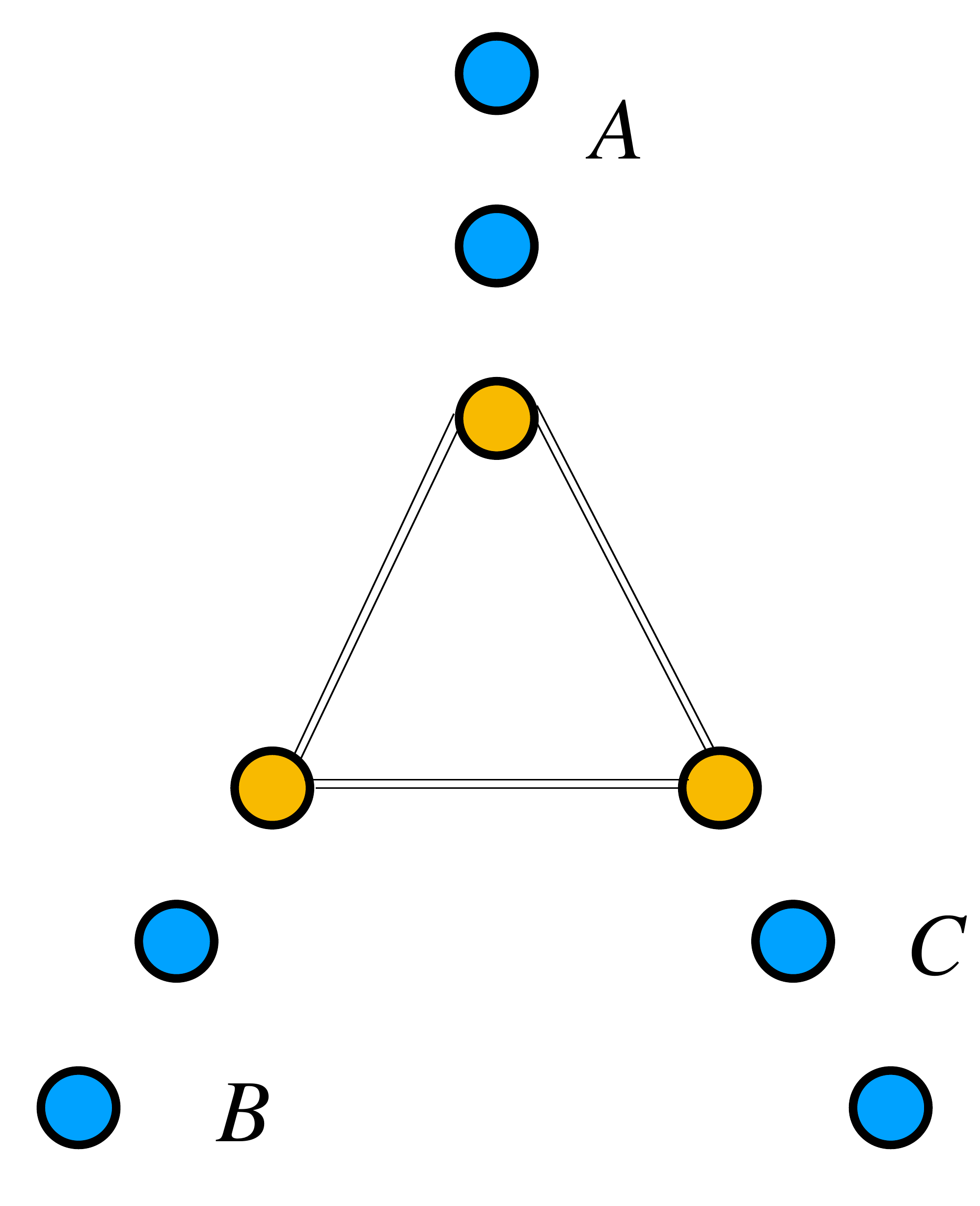}
	\caption{An illustration of an instance from the hard distribution for $0$-rounds with $n_0 = 3$. The middle (yellow) vertices are $\astar, \bstar, \cstar$. The double line indicates that the edge is present with probability $1/2$. }\label{fig:0-round}
\end{figure}

First, we prove a simple property about $0$-round instances sampled from~\Cref{dist:0}. 
\begin{observation}\label{obs:tri-exist-0-round}
	In any graph $G \sim \cG_0(n_0)$, a triangle exists between vertices $\{\astar, \bstar, \cstar\}$ with probability $1/8$, and there is no triangle otherwise. 
\end{observation}

\begin{proof}
	The edges are only between $\astar,\bstar,$ and $\cstar$ and exist independently with probability $1/2$. 
\end{proof} 

It is straightforward that protocols which do not communicate at all cannot have a high probability of success for 0-round instances; we formalize this statement next for completeness. 

\begin{claim}\label{clm:0-round-prot}
	Any deterministic protocol $\prot$ with $0$-rounds detects whether a triangle exists in $G \sim \cG_0(n)$ with success probability at most $7/8$.
\end{claim}
\begin{proof}
	In the following, we  additionally condition on any fixed choice of $(\astar,\bstar,\cstar)$ and assume this choice is known to all players. This can only strengthen the lower bound. 
	
	First consider any vertex $w \notin \set{\astar,\bstar,\cstar}$. If $w$ outputs \emph{Yes}, namely, that there is a triangle in $G$, then, 
	by~\Cref{obs:tri-exist-0-round}, the answer is only correct with probability $1/8$ as the choice of edges between $\astar,\bstar,\cstar$ is independent. So, we assume all vertices other than $(\astar,\bstar,\cstar)$ output \emph{No}. 
	
	Now, consider any vertex $\xstar \in \set{\astar,\bstar,\cstar}$. When degree of $\xstar$ is anything other than two, it can simply output \emph{No} and will be correct in this case. There are now two cases: 
	\begin{itemize}
		\item \emph{All} vertices in $\set{\astar,\bstar,\cstar}$ still output \emph{No} even when their degree is two: in that case, the protocol always outputs \emph{No} (by our assumption earlier) and thus is wrong with probability $1/8$ by~\Cref{obs:tri-exist-0-round}. 
		\item At least one vertex $\xstar \in \set{\astar,\bstar,\cstar}$ outputs \emph{Yes} when its degree is two: in that case, with probability half the edge between $\ystar \neq \zstar \in \set{\astar,\bstar,\cstar}$ may not appear, thus making
		the protocol wrong with probability $1/2$ in this case. The probability that $\xstar$ has degree two here is $1/4$ and conditioned on this, its output will be wrong with probability $1/8$. Thus, in this case also the protocol 
		is wrong with probability $1/8$ again. 
	\end{itemize}
	Overall the success probability of the protocol is at most $7/8$. 
\end{proof}

\subsection{Hard Distribution for $r$-rounds}

In this subsection, we describe the hard distribution for $r$-round protocols for $r \geq 1$. 
\vspace{-0.15cm}
\paragraph{Notation.}
 All vertices in the graph know which layer among $\{A, B, C\}$ they belong to. Each vertex in $X \in \{A, B, C\}$ has a unique identity $x_i \in X$ for $i \in [n_r]$, which distinguishes it from other vertices in $X$.

For $i \in [n_r]$, the input of each vertex $x_i \in X$ is given as two ordered vectors of length $n_r$ each labeled as $\Nbhood{x}{Y}{i}$ and $\Nbhood{x}{Z}{i}$, denoting the types of all vertex pairs $(x, w)$  for $w \in Y$ and $w \in Z$, respectively. 
The $j^{\textnormal{th}}$ entry in the vector corresponding to layer $Y$ (denoted by $\Nbhood{x}{Y}{i}[y_j]$) contains the type of the vertex pair $x_i, y_j \in [r+1] \cup \{0\}$, for $j \in [n_r]$ (and similarly for $\Nbhood{x}{Z}{i}$ and $Z$). 

Sometimes, we also use $\Nbtogether{x}{i}$ to denote the vectors of $x_i$ together. We use $\Nbtogether{x}{i}[S]$ for any $S \subseteq Y \cup Z$ to denote the types of all the vertex pairs of the form $(x_i, w)$ for $w \in S$. 

To avoid confusion, we sometimes write $X(G)$ to denote the layer $X \in \set{A,B,C}$ of the graph $G$ (to specify the graph). See also our note about the notation of $X$, etc. in~\Cref{sec:prelim}.

\begin{Distribution}\label{box:cGr-definition}
	\textbf{Distribution $\cG_r(n_r)$ over graphs $G$ for $r$-round protocols with $r \geq 1$:}
	
	(See \Cref{fig:1-round} for an illustration.)
	\begin{enumerate}[label=$(\arabic*)$]
		\item \label{item:Gr-Gr-1-sample}Sample a graph $G_{r-1} \sim \cG_{r-1}(n_{r-1})$ for some parameter $n_{r-1}$ to be fixed later in \Cref{eq:params-r-1}.
		\item \label{item:Gr-id-sample} For each $X \in \{A, B, C\}$ in $G$, sample $n_{r-1}$ distinct indices $\Xstar = \{ \xstar_1, \xstar_2, \ldots, \xstar_{n_{r-1}}\}$ from $X$ uniformly. The vertex $x_i$ in layer $X(G_{r-1})$ of $G_{r-1}$ 
		takes on the identity $\xstar_i$ in $G$.  
		\item Set $type(\xstar_i, \ystar_j)$ in $G$ to be $type(x_i,y_j)$ for any $x_i,y_j$ in different layers of $G_{r-1}$ for $i, j \in [n_{r-1}]$. 
		\item \label{item:Gr-neighbor-sample} Sample $2n_{r-1} \cdot (r+1)$ subsets $S^{y \to X}_t, S^{z \to X}_t \subseteq X \setminus \Xstar$, for each $y, z \in G_{r-1}$, and type $ t \in [r] \cup \{0\}$, all disjoint from each other. The sets are of size $d_r$ each, and are sampled uniformly at random. Parameter $d_r$ will be fixed later in \Cref{eq:params-r-1}.
		\item \label{item:Gr-degree-make} For $i \in[n_{r-1}]$ and $x_i \in X=X(G_{r-1})$, each other layer $Y\neq X$, and type $t \in [r] \cup \{0\}$, add channels of type $t$ between $\xstar_i$ to just enough sampled vertices from $S_t^{x_i \to Y}$ so that the total number of neighboring channels between $\xstar_i$ to any other layer $Y$ of type $t$ is exactly $d_r$.\footnote{$\xstar_i$ may have some type-$t$ channels from $G_{r-1}$, so we add more channels to increase its type-$t$ channels to $d_r$.}
	\item Set the type of all pairs of vertices, the types of which have not been fixed yet to be $r+1$. 
	\end{enumerate}
\end{Distribution}

\begin{figure}[t!]
	\centering
	\begin{subfigure}{0.45\textwidth}
		\centering 	\includegraphics[scale=0.25]{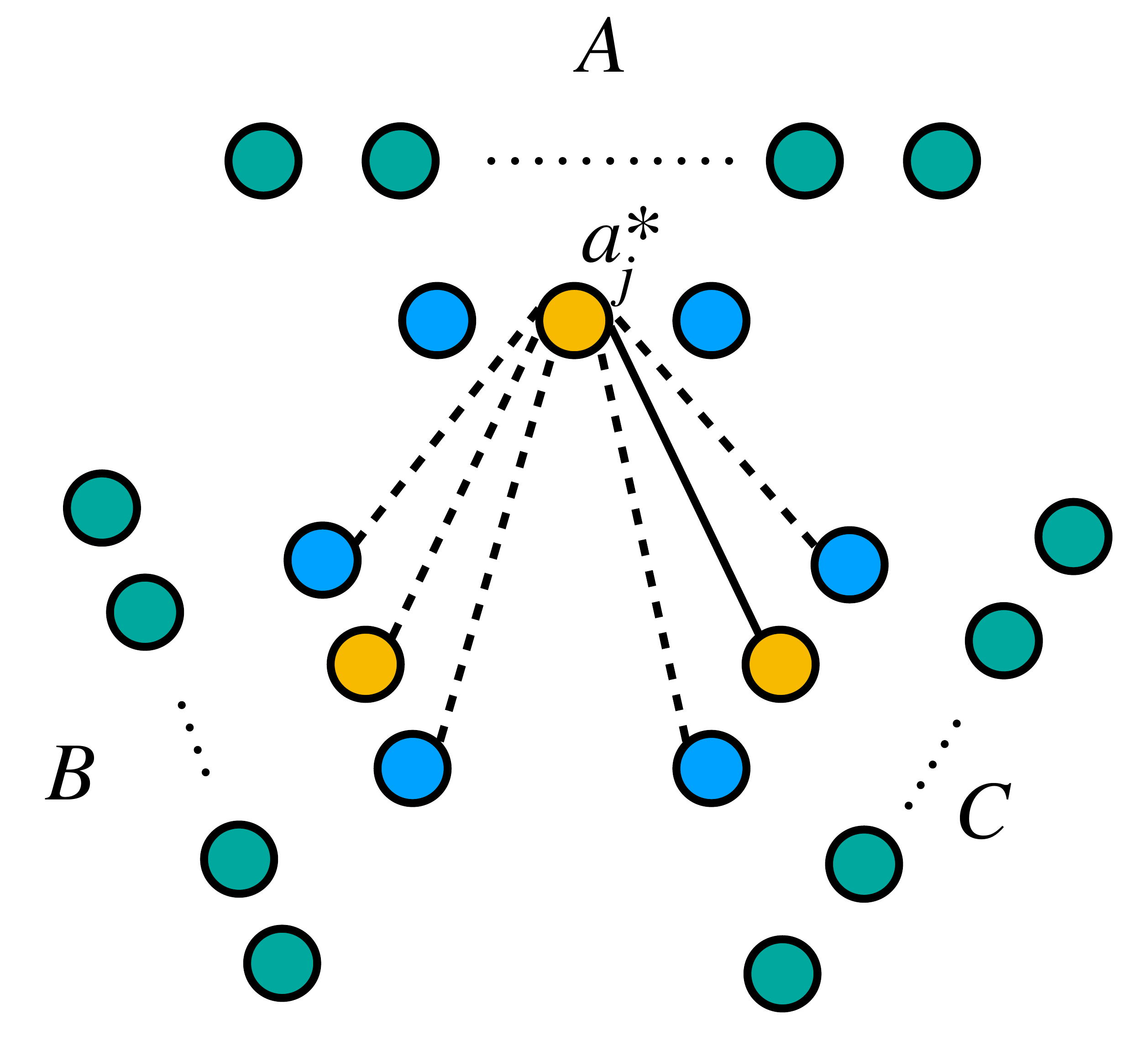}
		\caption{Realization of all types of channels between one inner vertex $\astar_j$ to other layers. The other inner vertices (blue and yellow) also have simlar channels to each other if they are in different layers.}
	\end{subfigure}
	\hfill
	\begin{subfigure}{0.45\textwidth}
		\centering 	\includegraphics[scale=0.25]{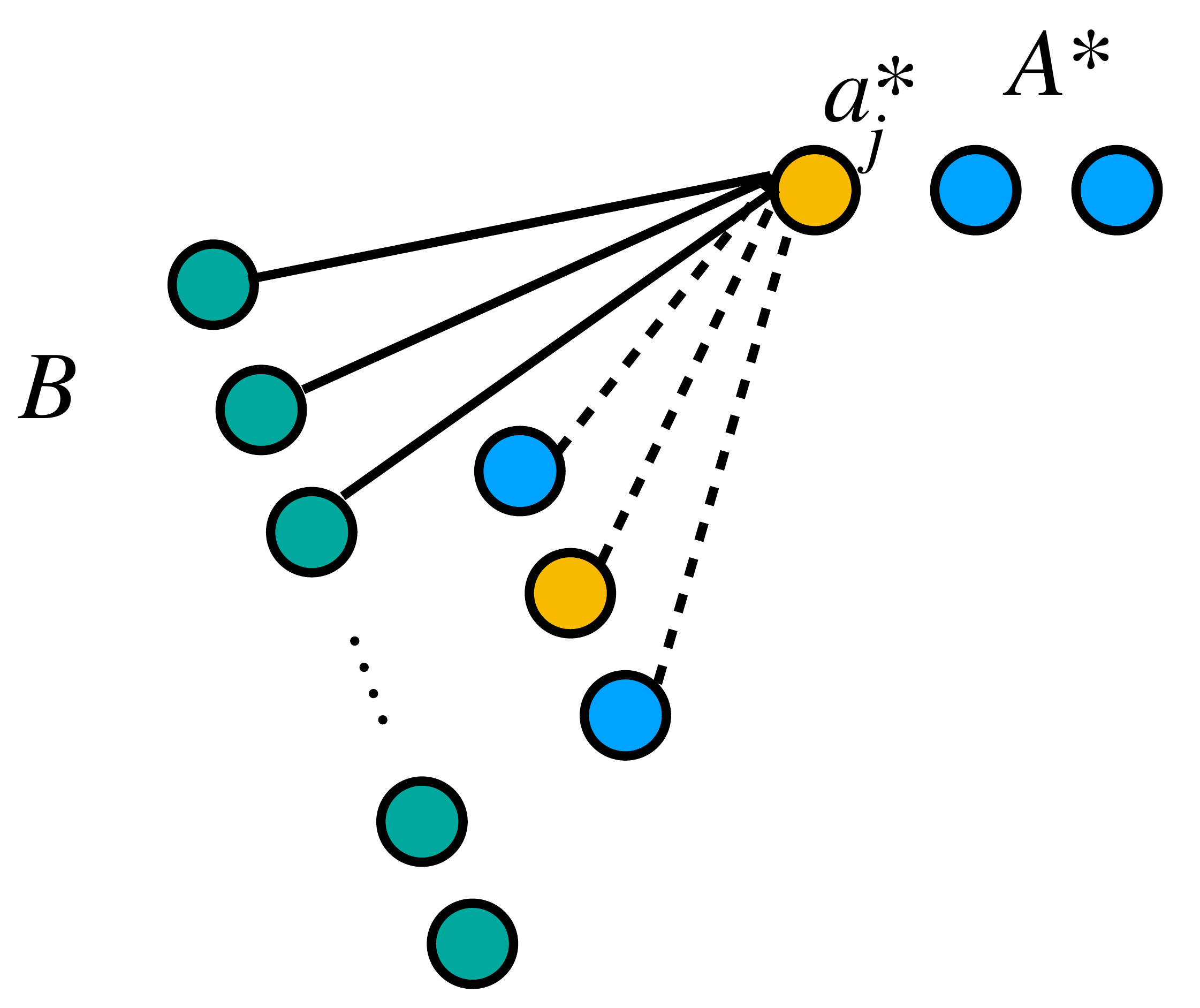}
		\caption{Addition of channels of type $0$ for one inner vertex $\astar_j$ to outer vertices to ensure correct channel-degree for all types (some parts of the instance are omitted).}
	\end{subfigure}

\vspace{0.25cm}

\begin{subfigure}{0.6\textwidth}
	\centering 
	\includegraphics[scale=0.25]{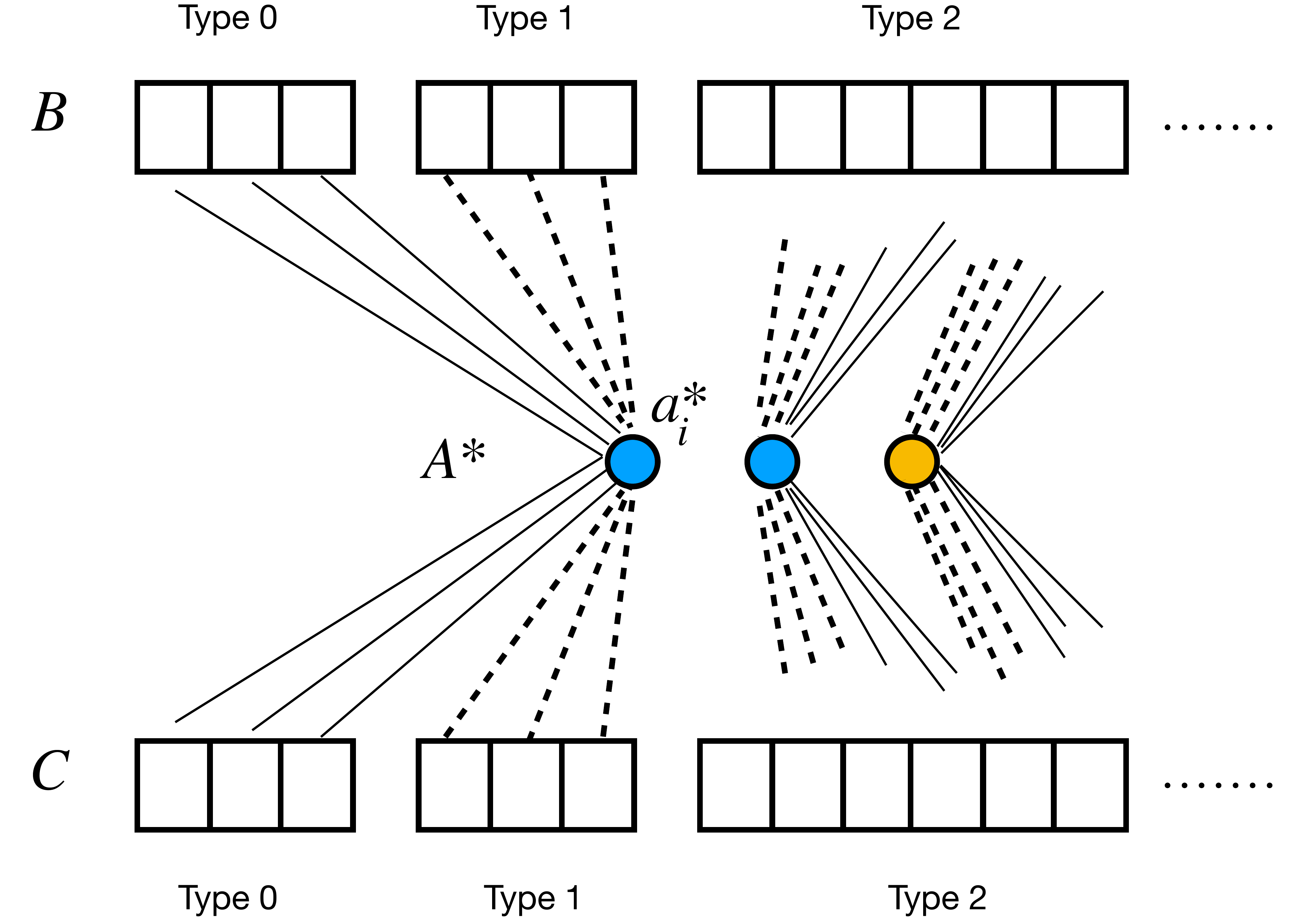}
	\caption{An illustration of the input of one inner vertex. The other inner vertices are connected only to outer vertices $u$ such that type of $(a^*_i, u)$ is $2$. Vertex $a_i$ does not know which vertices in $B, C$ are inner vertices. In the figure, only $\Astar$ and the relevant vertices of $B, C$ are shown. }
\end{subfigure}
	\caption{An illustration of some parts of the 1-round instance. Inner vertices from $0$-round instance are blue and yellow (in accordance with \Cref{fig:0-round}), and outer vertices are green. We used $n_0 = d_1 = 3$ for illustration, but this is certainly not true in the actual instance. Dashed lines indicate a channel of type $1$, straight lines indicate a channel of type $0$, and other pairs of vertices with no lines are of type $2$.  The input of blue and yellow vertices of the inner graph have exactly $3$ channels of type $0$ and type $1$ in total.}\label{fig:1-round} \vspace{-4mm}
\end{figure}
Let us fix the parameters in the hard distribution as follows, \vspace{-0.2cm}
\begin{equation}\label{eq:params-r-1}
	n_{r-1} = (n_r)^{1/34}  \qquad d_r = (n_{r-1})^{13}.
\end{equation}

We sample $n_{r-1}$ vertices from $X$ for each $x_i \in G_{r-1}$ in step \ref{item:Gr-Gr-1-sample}. We also sample $d_r$ distinct vertices for each $y_i, z_j$ and type $t$ with $i, j \in [n_{r-1}]$, and $t \in [r] \cup \{0\}$ in step \ref{item:Gr-neighbor-sample}. In total, 
\begin{equation*}
n_{r-1} + (r+1) \cdot d_{r} \cdot 2n_{r-1}
\end{equation*}
distinct vertices are sampled from $X$ uniformly at random. 
We have to check that $n_r$ is large enough to make this sampling process  feasible:
\begin{align*}
	n_{r-1} + (r+1) \cdot d_{r} \cdot 2n_{r-1} = n_{r-1}(2d_r(r+1) + 1) \leq n_{r-1}^{16} < n_r.
\end{align*}
Thus, we have more than enough room in $n_r$ for step \ref{item:Gr-neighbor-sample}.

We call the vertices in $\Xstar, \Ystar, \Zstar$ to be \textbf{inner vertices} (sampled in step \ref{item:Gr-id-sample}), and all other vertices  as \textbf{outer vertices}. We refer to the graph $G_{r-1}$ from step \ref{item:Gr-Gr-1-sample} as the \textbf{inner graph}. 
We prove some simple properties of $\cG_r(n_r)$ now. 

\begin{observation}\label{obs:G_r-degree}
	The number of neighboring channels of type $t \in [r] \cup \{0\}$ of any inner vertex in $G$ is exactly $d_r$. The channel-degree of all outer vertices is at most one. 
\end{observation}
\begin{proof}
	In step \ref{item:Gr-degree-make}, we ensure the number of channels of each type $t \in [r] \cup \{0\}$ for each inner vertex is exactly $d_r$ explicitly, by adding enough channels to sets sampled in step \ref{item:Gr-neighbor-sample}. Any outer vertex gets sampled at most once in the sets of step \ref{item:Gr-neighbor-sample}, as all the sets are disjoint. Therefore all outer vertices have channel-degree at most one.
\end{proof}

We only ever refer to the input of inner vertices to talk about graph $G \sim \cG_r(n_r)$, due to the following observation. 

\begin{observation}\label{obs:input-fix-inner}
	In any $G \sim \cG_r(n_r)$, fixing the identity $\xstar_i$ and the two $n_r$ size neighborhood vectors, $\Nbhood{x}{Y}{}, \Nbhood{x}{Z}{}$ for all inner vertices $x_i \in G_{r-1}$ fixes the graph $G$. 
\end{observation}
\begin{proof}
	The only randomness involved in the distribution $\cG_r(n_r)$ is in the graph $G_{r-1}$, the identities chosen for each inner vertex, and the at most $(r+1) \cdot d_r$ many neighboring channels sampled for each inner vertex. Thus, the 
	variables in the observation statement fix the graph $G$. 
\end{proof}

We can say something further about the structure of $\Nbhood{x}{Y}{}$ for every inner vertex $x$ and layer $Y$ with $x \notin Y$.
\begin{observation}\label{obs:input-vec}
For every inner vertex $x \in G_{r-1}$, the marginal distribution of $\Nbhood{x}{Y}{}$ for every other layer $Y$ is  obtained by:
\begin{itemize}
	\item Setting the types of all vertices in $\Ystar$ and $\Zstar$ based on the type in $G_{r-1}$. 
	\item Sampling the rest of $\Nbtogether{x}{}$ uniformly random, conditioned on the types of $\Ystar \cup \Zstar$ so that there are exactly $d_r$ vertices of type $t \in [r] \cup \{0\}$, and $n_r - d_r(r+1)$ vertices of type $r+1$.
\end{itemize}
\end{observation}
\begin{proof}
	We know that the number of channels of type $t \in [r] \cup \{0\}$ for each vertex $x_i \in G_{r-1}$ is exactly $d_r$ from \Cref{obs:G_r-degree}. Therefore the total number of channels with types in $[r] \cup \{0\}$ is $(r+1) \cdot d_r$. The remaining $n_r - d_r \cdot (r+1)$ channels are of type $r+1$. 
	
	The sets $S_t^{x \to Y}$ and $S_t^{x \to Z}$ are also chosen uniformly at random, so that they are disjoint from $\Ystar$ and $\Zstar$. Hence, the final distribution of $\Nbhood{x}{Y}{}$ is such that uniformly random subsets are chosen for each type. The total number of vertices in $\Ystar$ is exactly $n_{r-1}$, and even if all of them had the same type, we know that there is still enough room for $d_r$ vertices of each type $t \in [r] \cup \{0\}$ as $d_r > n_{r-1}$ from \Cref{eq:params-r-1}.
\end{proof}

The identities of inner vertices are chosen at random, so the distribution of the neighborhoods are symmetric over all $n_r$ vertices of $X$. More formally, we have the following observation. 
\begin{claim}\label{clm:symmetry-over-layer}
	For $r \geq 0$, in $\cG_r(n_r)$ for all $i \in [n_{r-1}]$ and $x_i \in X(G_{r-1})$, the marginal distribution of the two neighborhood vectors given to the inner vertex $x_i$ in $G_{r-1}$ is independent of $x$ and $i$.
\end{claim}
\begin{proof}
	It is sufficient to show that the marginal distribution of the input given to all vertices in distribution $\cG_{r-1}(n_{r-1})$ are the same, and independent of $x$ and $i$ for all $r \geq 1$. 
	We prove the claim by induction on $r$. 
	For the base case when $r = 1$, the inner graph is sampled from $\cG_0(n_0)$. We know that the vertex $\xstar$ is chosen uniformly at random from $[n_0]$. Therefore, each vertex in layer $X$ has equal probability of being chosen as $\xstar$ in $\cG_0(n_0)$. 
	
	For any $r > 1$, we know that all the identities of inner vertices $\xstar_1, \xstar_2, \ldots, \xstar_{n_{r-2}}$ in $\cG_{r-1}(n_{r-1})$ are chosen uniformly at random from $[n_{r-1}]$. Therefore, any vertex from $X(G_{r-1})$ is equally likely to be chosen as the identity of a vertex from $G_{r-2}$ (the inner graph of $G_{r-1}$). The marginal distribution of the neighborhood in $G_{r-2}$ is identical by induction hypothesis. 
	
	For the outer vertices in $X$, again, the sets sampled in step \ref{item:Gr-neighbor-sample} are uniform over $[n_{r-1}]$, and vertices from $[n_{r-1}]$ are equally likely to be chosen to be a part of these sets in distribution $\cG_{r-1}(n_{r-1})$. Our construction is fully symmetric over the three layers and the vertices in each layer. Therefore, the marginal distribution of the input given to each vertex in graph $G_{r-1}$ is the same and independent of $x$ and $i$. 
\end{proof}

Lastly, we need to argue about the existence of triangles in $G \sim \cG_r(n_r)$.

\begin{observation}\label{obs:r-triangle-exist}
	In any graph $G \sim \cG_r(n_r)$, a triangle exists if and only if a triangle is present in $G_{r-1}$ sampled in step \ref{item:Gr-Gr-1-sample} while sampling $G$.
\end{observation}
\begin{proof}
	Firstly, the channel-degree of all outer vertices is at most one by \Cref{obs:G_r-degree}. Thus, all of them have at most one edge also. Hence, only the inner vertices can be a part of a triangle. We do not add or remove any edges between inner vertices, we 
	only add channels which do not form edges. Hence, triangle existence in $G_{r-1}$ is preserved in graph $G$. 
\end{proof}


\subsection{Another Way of Sampling from Hard Distribution}
In this subsection, we talk more about the hard distribution $\cG_r(n_r)$, and give an alternate way of sampling it. This makes the analysis much easier, as it gives us another perspective of looking at the hard distribution. This alternate way is not without loss, however, but we show that these distributions are quite close in total variation distance. 

 We only ever talk about the input and identities of inner vertices, as this fixes distribution $\cG_r(n_r)$ by \Cref{obs:input-fix-inner}. 

\paragraph{Notation.}
We use $\idlayer{X}$ to denote the variables $\xstar_1, \xstar_2, \ldots, \xstar_{n_{r-1}}$ (elements of $\Xstar$) for $X\in {A, B, C}$. We use $\ids$ to denote the variables $(\idlayer{A}, \idlayer{B}, \idlayer{C})$ collectively.
We use $\id{x_i}$ to denote $\xstar_i$ for $i \in [n_{r-1}]$.
We use $\Ninner{x}{i}$ to denote the input of inner vertex $x_i \in G_{r-1}$. This is two vectors of length $n_{r-1}$ each, containing types to layers $Y$ and $Z$ in $G_{r-1}$. 

We use $\cDmarg$ to denote the marginal distribution of $\Ninner{x}{i}$ for each inner vertex $x_i$. Formally, $\cDmarg{}$ is the distribution of $\Nbtogether{x}{i}$ in an instance $G \sim \cG_{r-1}(n_{r-1})$ for $i \in [n_{r-1}]$. Note that this is identical for each $i \in [n_{r-1}]$ by \Cref{clm:symmetry-over-layer}.

We need to define three more parameters to proceed. 
\begin{equation}\label{eq:params-alphabetagamma}
	\alpha_r = (n_{r-1})^{11} \qquad \beta_r = (n_{r-1})^{5}\qquad \gamma_r = (n_{r-1})^6.
\end{equation}
We know that $d_r$ is much larger than $n_{r-1}$ from \Cref{eq:params-r-1}. However, we want to argue that, for any inner vertex $x_i \in G_{r-1}$, given only its input $\Nbtogether{x}{i}$ in graph $G$, there are multiple choices for what the other inner vertices may be; formally, we need the following for each inner vertex $x_i \in G_{r-1}$:
\begin{itemize}[leftmargin=5pt]
	\item An $\alpha_r$-size collection of sets of size $2n_{r-1}$ each, all of which are distributed independently and exactly according $\cDmarg$ with channels to $x_i$. 
	\item For each type $t \in [r] \cup \{0\}$ and any other inner vertex $y \notin X(G_{r-1})$, a $\beta_r$-size collection of sets of size $2n_{r-1}-1$ each, all of which are distributed  independently of each other according to the $(r-1)$-round input on $2n_{r-1}-1$ vertices, conditioned on the channel to the one remaining vertex being of type $t$.
	\item For each type $t \in [r] \cup \{0\}$ and any other inner vertex $y \notin X(G_{r-1})$, a set of size $\gamma_r$ with channels of  type $t$ to $x_i$. 
\end{itemize}
We show it is possible to sample from $\cG_r(n_r)$ while satisfying these conditions with high probability.

\begin{Distribution}\label{box:cGtilr-definition}
			\textbf{Distribution $\cGtil_r(n_r)$ for $r$-round protocols with $r \geq 1$:}
		\begin{enumerate}[label=$(\arabic*)$]
			\item Sample $G_{r-1}$ and sets $\Xstar, \Ystar, \Zstar$ as before from \Cref{box:cGr-definition}. 
			
			\item \label{item:disjt-sampling-1} Fix any inner vertex $x \in G_{r-1}$, and from the other sets $Y \setminus \Ystar$ and $Z \setminus \Zstar$, sample the following subsets, disjoint from each other for all $x$: (\Cref{obs:num-ZZ*} shows that this is feasible)
			\begin{enumerate}[label=$(\alph*)$]
				\item \label{item:J-appears} A collection $\cJup{x} = (\Jupdown{x}{1},  \ldots, \Jupdown{x}{\alpha_r})$, where each $\Jupdown{x}{i}  $ for $i \in [\alpha_r]$ is a set containing $n_{r-1}$ elements of $Y \setminus \Ystar$ and $n_{r-1}$ elements of $Z \setminus \Zstar$.
				\item \label{item:K-appears} For each type $t \in [r] \cup \{0\}$, value $i \in [n_{r-1}]$, each other layer $Y, Z$ a collection 
				\[
				\cKupdown{x \to Y}{ t, i} = (\Kupdown{x \to Y}{ t, i, 1}, \ldots, \Kupdown{x \to Y}{t, i, \beta_r})
				\]
				 where each $\Kupdown{x \to Y}{t, i, j}$ for $j \in [\beta_r]$ is a set containing $(n_{r-1}-1)$ elements from $Y \setminus \Ystar$ and $n_{r-1}$ elements from $Z \setminus \Zstar$. Similarly, $\cKupdown{x \to Z}{t,i}$ is a collection of $\beta_r$-sets, each with $n_{r-1}$ elements from $Y \setminus \Ystar$ and $n_{r-1}-1$ elements from $Z \setminus \Zstar$. 
				\item \label{item:L-appears} For each type $t$ and value $i \in [n_{r-1}]$ and each other layer $Y$, a set $\Lupdown{x \to Y}{t, i}$ of $\gamma_r$ elements from $Y \setminus \Ystar$.
				\end{enumerate}
			\item \label{item:disjt-sampling-2} For each inner vertex $x \in G_{r-1}$,  do the following: 
			\begin{enumerate}
				\item For each $i \in [\alpha_r]$, sample $(\Nbtogether{x}{}[\Jupdown{x }{i}])$ from $\cDmarg{}$ independently of other $i$.
				\item For each $t \in [r] \cup \{0\}$, each other layer $Y$,  $i \in [n_{r-1}]$ and $j \in [\beta_r]$, sample $\Nbtogether{x}{}[\Kupdown{x \to Y}{t, i, j}]$ so that $\Nbtogether{x}{}[\Kupdown{x \to Y}{t, i, j} \cup \{\ystar_i\}]$ is sampled from $\cDmarg$, conditioned on  $type(x, \ystar_i)=t$. 
				\item For each $t \in [r] \cup \{0\}, i \in [n_{r-1}]$, and other layer $Y$, set $\Nbhood{x}{Y}{}[y]$ for each $y \in \Lupdown{x \to Y}{t, i}$ to be type $t$. 
				\item Sample the rest of the input of $\Nbtogether{x}{} \sim \cG_r(n_r)$ conditioned on the random variables sampled in earlier steps independent of the input of other inner vertices (see \Cref{clm:sample-input-separate}.)
			\end{enumerate}
			\item For all pairs of vertices whose types are not fixed yet, set them to be of type $r+1$.
		\end{enumerate}
\end{Distribution}

\begin{figure}[t!]
	\centering
	\includegraphics[scale=0.25]{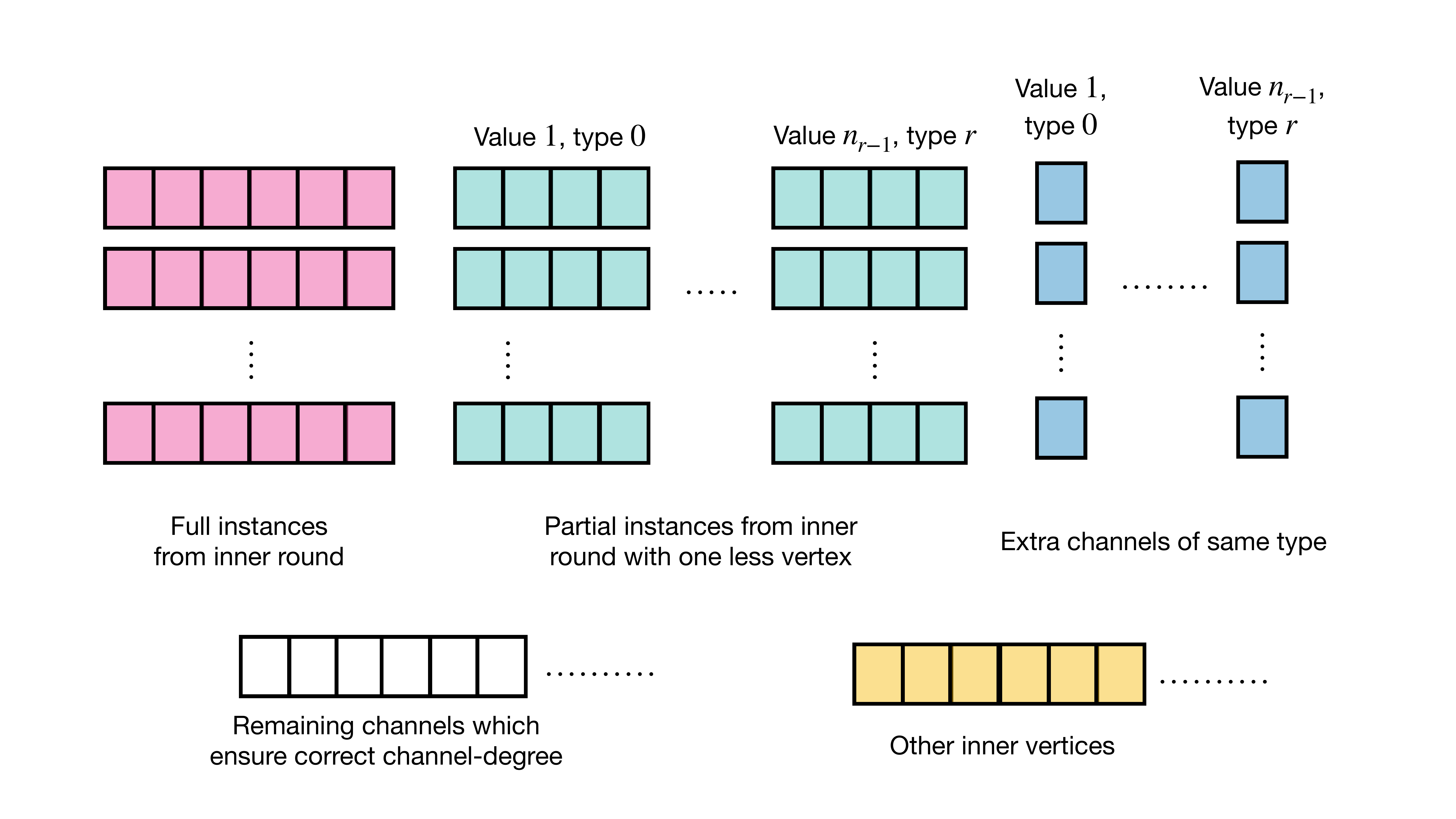}
	\caption{An illustration of what the input of one inner vertex $x$ looks like in $\cGtil_r$. Pink outer vertices correspond to collections $\cJup{}$, green outer vertices correspond to collections $\cKupdown{}{}$ (only for one other layer $Y$ are shown, this is repeated once more), blue outer vertices correspond to sets $\Lupdown{}{}$ (again, only for $Y$ are shown, this is repeated for $Z$), yellow vertices correspond to vertices in $\Ystar \cup \Zstar$, and the unshaded vertices are the other outer vertices with various types to ensure correct number of channels of each type $t \in [r] \cup \{0\}$.}\label{fig:Gtilr}
\end{figure}

The parameters in \Cref{eq:params-r-1} and \Cref{eq:params-alphabetagamma} are chosen so that the sets in step \ref{item:disjt-sampling-1} can be sampled:

\begin{observation}\label{obs:num-ZZ*-onevert}
	In step \ref{item:disjt-sampling-1} of \Cref{box:cGtilr-definition}, for each inner vertex $x$, and layer $Y$ with $x \notin Y$, the total number of identities sampled is at most $2(n_{r-1})^{12}$. 
\end{observation}
\begin{proof}
		For each inner vertex $x \in X$, to another layer $Y$, we sample:
		\vspace{-5pt}
	\begin{itemize}
		\item $\alpha_r$-many sets of size $n_{r-1}$ each. 
		\item A collection of $\beta_r$-many sets of size at most $n_{r-1}$ each for type $t \in [r] \cup \{0\}$, value $i \in [n_{r-1}]$ and each other layer which does not contain $x$. 
			\item A collection of size $\gamma_r$ for $(r+1)$ types and $n_{r-1}$ values. 
	\end{itemize}
	\vspace{-5pt}
		In total, we sample, at most 
	\begin{align*}
		\alpha_r \cdot n_{r-1} + 2\beta_r \cdot (n_{r-1})^2 \cdot (r+1) + \gamma_r \cdot (r+1) \cdot n_{r-1} 
		\leq (n_{r-1})^{12} + 2 (n_{r-1})^8 + (n_{r-1})^8 \tag{as $r+1 < n_{r-1}$ for large $n_r$ and by \Cref{eq:params-alphabetagamma} and for large enough $n_{r-1}$} 
		\leq 2 (n_{r-1})^{12}
	\end{align*}
	vertices for each inner vertex not in $Y$. 
\end{proof}

\begin{observation}\label{obs:num-ZZ*}
	In Step \ref{item:disjt-sampling-1} of \Cref{box:cGtilr-definition}, for each layer $Z$, the total number of vertices sampled from $Z \setminus \Zstar$ is at most $4(n_{r-1})^{13} < n_r - n_{r-1}$.
\end{observation}
\begin{proof}
	By \Cref{obs:num-ZZ*-onevert}, for each inner vertex not in $Z$, we sample at most, $2(n_{r-1})^{12}$ 
	vertices from $Z$. There are exactly $2n_{r-1}$ inner vertices which are not in $Z$. In total, from set $Z \setminus \Zstar$, we sample at most $4 (n_{r-1})^{13}$ vertices, which is less than $n_{r}-n_{r-1}$ by construction.  
\end{proof}

We show that the parameters are set so that step \ref{item:disjt-sampling-2} can be performed. 
\begin{claim}\label{clm:sample-input-separate}
	In the definition of $\cGtil_r(n_r)$ from \Cref{box:cGtilr-definition}, conditioned on $G_{r-1}$ and values of $\Xstar, \Ystar$ and $\Zstar$, in step \ref{item:disjt-sampling-2}, each $\Nbtogether{x}{}$  can be sampled marginally the same way as in distribution $\cG_r(n_r)$ from \Cref{box:cGr-definition}.
\end{claim}
\begin{proof}
We will prove that after the sampling process in step \ref{item:disjt-sampling-1} is done, it is feasible to sample the input $\Nbtogether{x}{}$ of any inner vertex $x$ from $\cG_r(n_r)$ marginally. 
	We refer to \Cref{obs:input-vec}. We know that the types of all vertices in $\Ystar \cup \Zstar$ are set correctly. 
	
We argue that it is still possible to sample the input so that uniformly random $d_r$-size subsets are chosen for each type $t \in [r] \cup \{0\}$ for each layer $Y, Z$. 
	
	By \Cref{obs:num-ZZ*-onevert}, we know that the total number of entries in $\Nbhood{x}{Y}{}$ for each other layer $Y$ whose types are fixed by the first three steps of \ref{item:disjt-sampling-2} are at most $2(n_{r-1})^{12}$,
	which is still less than $d_r$ by \Cref{eq:params-r-1}. Therefore, we have enough room to marginally sample the rest of $\Nbtogether{x}{}$ so that the conditions in \Cref{obs:input-vec} are satisfied. 
\end{proof}

Finally, we argue that these distributions are close overall. 

\begin{claim}\label{clm:cG-cGtil-close}
Distribution $\cGtil_r(n_r)$ which is obtained by marginally sampling each $\Nbhood{x}{Y}{}$ as in \Cref{box:cGtilr-definition} is close in total variation distance to $\cG_r(n_r)$. 
\[
	\tvd{\cG_r(n_r)}{\cGtil_r(n_r)} \leq \frac1{n_{r-1}}.
\]
\end{claim}
\begin{proof}
In distribution $\cG_r(n_r)$, after the types of vertex pairs $(x, y)$ are set according to $G_{r-1}$ for inner vertices $x, y$, we sample \emph{disjoint} subsets of the remaining $X \setminus \Xstar$ to add as neighbors to inner vertices in $\Ystar$ and $\Zstar$. This ensures that all outer vertices have channel-degree at most one, as we have stated in \Cref{obs:G_r-degree}. 

In distribution $\cGtil_r(n_r)$, all the subsets of outer vertices sampled in step \ref{item:disjt-sampling-1} are disjoint from each other by construction. Even though channels of some type $t \in [r] \cup \{0\}$ are added to all these outer vertices, exactly one channel is added, and they have channel-degree at most one. 

However, the rest of the input for each inner vertex $x$ is sampled independenly of the input of other inner vertices in $\cGtil_r(n_r)$. 
This process corresponds to sampling at most $d_r$ vertices for each type $t \in [r] \cup \{0\}$ for each inner vertex, so that the condition stated in \Cref{obs:G_r-degree} is satisfied. Namely, the condition that the number of channels of each type $t \in [r] \cup \{0\}$ for each inner vertex is exactly $d_r$. Outer vertices sampled for this purpose \emph{may} have collisions, which will cause outer vertices to have channel-degree more than one.

We will show that even though the rest of the outer vertices in step \ref{item:disjt-sampling-2} for each inner vertex are sampled independently of each other, it is highly unlikely that some vertex is sampled twice. Let $\Euniq$ be the event that all the outer vertices sampled to fill up the $d_r$ slots for each type, and for each inner vertex are unique for all layers $X, Y, Z$. We will show that probability of $\Euniq$ happening is large. We have, 
\begin{align*}
\Pr[\Euniq] \leq (\Pr[(2n_{r-1} \cdot (r+1) \cdot d_r  + n_{r-1})\textnormal{ vertices sampled u.a.r. from $[n_r]$ do not overlap}])^3,
\end{align*}
where we bound the probability of a super set of $\Euniq$. In $\Euniq$, fewer than $d_r \cdot (r+1)$ vertices are sampled from any layer $Y$ for any inner vertex $x \notin Y$. We will show that even if $d_r \cdot (r+1)$ vertices are sampled for each inner vertex, in addition to the set $\Ystar$ uniformly at random and independently of each other, no vertex is sampled more than once except with a very small probability. 

We choose $2n_{r-1} \cdot (r+1) \cdot d_r  + n_{r-1} < n_{r-1}^{16} := \theta$ vertices uniformly at random and independently from $n_r $ vertices in layer $Y$. The probability that none of them are sampled more than once is,
\begin{align*}
\binom{n_r}{\theta} \cdot \theta! \cdot \frac1{(n_r)^{\theta}} &= \prod_{i=0}^{\theta-1} (1-i/n_r) \geq \exp\paren{-2 \cdot \sum_{i=0}^{\theta} i/n_r} = \exp(-(\theta-1) \cdot \theta/n_r).
\end{align*}

We do this for all three layers $A, B$ and $C$. The probability that no element is sampled twice in any layer is at least, 
\begin{align*}
	\Pr[\Euniq] \geq	\exp(-3 \cdot (\theta-1) \cdot \theta/n_r) \geq \exp(-3/n_{r-1}^2) \geq 1-1/n_{r-1},
\end{align*}
where we have used the value of $n_r = n_{r-1}^{34} = \theta^2 \cdot n_{r-1}^2$ from \Cref{eq:params-r-1} and that $n_{r-1} \geq 3$. 

Conditioned on event $\Euniq$, we argue that the distribution of $\cG_r(n_r)$ and $\cGtil_r(n_r)$ are identical. We already know that the distributions of the input of each inner vertex is marginally the same by \Cref{clm:sample-input-separate}. When $\Euniq$ happens, all the outer vertices have channel-degree at most one, and the neighbors of each inner vertex are disjoint barring $\Xstar, \Ystar, \Zstar$. This satisfies all the properties of instances sampled from $\cG_r(n_r)$. 

Let $\rIuniq$ be the indicator random variable for event $\Euniq$. 
\begin{align*}
	\tvd{\cG_r(n_r)}{\cGtil_r(n_r)} &\leq \Exp_{I \sim \rIuniq} \tvd{(\cG_r(n_r) \mid \rIuniq = I)}{(\cGtil_r(n_r) \mid \rIuniq = I)} \tag{by ``overconditioning'' as in \Cref{fact:tvd-over-conditioning}}\\
	&\leq \Pr[\rIuniq = 1] \cdot \tvd{(\cG_r(n_r) \mid \Euniq)}{(\cGtil_r(n_r) \mid \Euniq)} + \Pr[\rIuniq = 0] \\
	&=  0 + \Pr[\rIuniq = 0]  \tag{conditioned on $\Euniq$, the distributions are the same} \\
	&\leq \frac1{n_{r-1}},
\end{align*}
where, in the last inequality we have used our upper bound on the probability of $ \neg \Euniq$.
\end{proof}

We need some more notation about random variables in $\cGtil_r(n_r)$ for the later sections. 
\paragraph{Notation.}
We use $\cJall$ to denote the joint random variable containing all $\cJup{x}$ for all inner vertices $x$. Similarly, we use $\cKall, \Lall$ to denote the joint random variable $\cKupdown{x \to Y}{t, i}, \Lupdown{x \to Y}{t, i}$ respectively for all inner vertices $x$, layer $Y$ with $x \notin Y$, type $t \in [r] \cup \{0\}$, and value $i \in [n_{r-1}]$. 
We call the random varibles $\cJall, \cKall, \Lall$ as \textbf{auxiliaries}, and together we use $\aux$ to denote them.

\begin{observation}\label{obs:input-indep}
		We have the following independence properties for the inputs of vertices in graph $G \sim \cGtil_r(n_r)$:
	\begin{enumerate}[label=$(\roman*)$]
		\item \label{item:ind-prop-1} Conditioned on $(\rG_{r-1},\raux, \rids)$, the inputs of each of the inner vertices in $G$ are independent of each other. 
		\item \label{item:ind-prop-2} The input $\Nbhood{x}{Y}{}, \Nbhood{x}{Z}{}$ for any inner vertex $x$ in the graph $G$ is independent of all other random variables in $G_{r-1}$ conditioned on $\Ninner{x}{}, \rids$ and $\raux$. 
		\item \label{item:ind-prop-3} The random variable $\rG_{r-1}$ is independent of $\rids$ and $\raux$. 
	\end{enumerate}
\end{observation}
\begin{proof}
	For part \ref{item:ind-prop-1}, we know that for each inner vertex $x$, once the new identities $\ids$ and auxiliaries $\aux$ are fixed, and graph $G_{r-1}$ is fixed, its input is sampled independently of the other inner vertices in step \ref{item:disjt-sampling-2}.
	
	Part \ref{item:ind-prop-2} follows by a similar argument: once $\Ninner{x}{}, \ids$ and $\aux$ are fixed, the entire input of $x$ is fixed by randomness completely independent of $G_{r-1}$. 
	
	For part \ref{item:ind-prop-3}, it is easy to see that $\rids, \rcJall, \rcKall $ and $\rLall$ are sampled from $A, B, C$ in graph $G$, irrespective of the inner graph $\rG_{r-1}$. 
\end{proof}

\subsection{Distribution of Input and Messages}\label{subsec:dist-input-message}
In this subsection, we talk about the joint distribution of input $\cG_r(n_r)$ and the messages sent across all channels in the first round of communication. 
Let $\prot_r$ be a protocol for solving triangle detection on an input $G $ sampled from $\cG_r(n_r)$.

\paragraph{Notation.}
Let $\cDreal$ denote the distribution of the inputs and messages in the first round of $\prot_r$. For any inner vertex $x_i$, let $\msg{x}{Y}{i}$ denote the $n_r$-length vector of the messages sent by $x_i$ to its channels in $Y$. When there is no channel between $x_i$ to some vertex $y_j$, we use $\bot$ to denote the null message that $x_i$ sends.
We use $\msgrest{x}{i}$ to denote the messages $x_i$ receives from outer vertices. We use $\msgall$ to denote all the messages sent and received by all inner vertices in the first round. Note that this is also all the messages sent over the first round as the type of all pairs of outer vertices is set as $r+1$, and they can send no messages to each other in any round.

We define a distribution $\cDrealtil$ for the input graph and first round messages below. We will show that this distribution is generated when input $G$ is sampled from $\cGtil_r(n_r)$, and protocol $\prot_r$ is run over this graph instead. We will also show that distribution $\cDrealtil$ is close to $\cDreal$. 

\begin{Distribution} \label{box:cDrealtil}
\textbf{Distribution $\cDrealtil$:}
\begin{align*}
	&\rG_{r-1} \times  (\rids, \rcJall, \rcKall, \rLall) \tag{the inner graph, identities, auxiliaries}\\
	&\times (\bigtimes_{\substack{x \in \{a, b, c\}, \\  i \in [n_{r-1}]}} \rNbhood{x}{Y}{i}, \rNbhood{x}{Z}{i} \mid \rNinner{x}{i}, \raux, \rids)  \tag{the inputs of all inner vertices}\\
	&\times (\bigtimes_{\substack{x \in \{a, b, c\}, \\  i \in [n_{r-1}]}} \rmsg{x}{Y}{i}, \rmsg{x}{Z}{i} \mid \rNbhood{x}{Y}{i}, \rNbhood{x}{Z}{i}, \raux, \rids)  \tag{the messages sent by inner vertices}\\
	&\times (\bigtimes_{\substack{x \in \{a, b, c\}, \\  i \in [n_{r-1}]}} (\rmsgrest{x}{i} \mid \rNbhood{x}{Y}{i}, \rNbhood{x}{Z}{i}, \raux, \rids)) \tag{the messages sent by outer vertices to inner vertices}
\end{align*}
\end{Distribution}

\begin{observation}\label{obs:cDrealtil-input-same}
The distribution of inputs in $\cDrealtil$ is sampled from $\cGtil_r(n_r)$.
\end{observation}
\begin{proof}
	In distribution $\cGtil_r(n_r)$, we know that all the $\rids, \rcJall$, $\rcKall $ and $\rLall$ are sampled conditioned on each other. And, they are independent of $\rG_{r-1}$ from \Cref{obs:input-indep}-\ref{item:ind-prop-3}. Hence, we have that the distributions of $\rG_{r-1}$, $\rids$ and $\raux$ are exactly as in $\cGtil_r(n_r)$ in $\cDrealtil$.
	
	The input of the inner vertices are independent of each other conditioned on $\rG_{r-1}$, $\rids$ and $\raux$ from \Cref{obs:input-indep}-\ref{item:ind-prop-1}. Therefore, they can be sampled separately as in $\cDrealtil$.
	
	Lastly, we know that for any inner vertex $x$, its input is independent of $\rG_{r-1}$ when conditioned on $\rNinner{x}{}$, $\rids$ and $\raux$ by \Cref{obs:input-indep}-\ref{item:ind-prop-2}, so the other random variables in $\rG_{r-1}$ can be ignored also when sampling these inputs. 
\end{proof}

\begin{observation}\label{obs:cDrealtil-messages}
	For any graph $G$ sampled from $\cG_r(n_r)$, the distribution of all the first round messages is the same in $\cDreal$ and $\cDrealtil$ conditioned on the input graph being $G$.
\end{observation}
\begin{proof}
	The protocol $\prot$ is deterministic, hence the messages sent by any vertex are only a function of its input. For all the inner vertices, in $\cDrealtil$, the messages are sampled conditioned on their entire input, and are distributed the same as in $\cDreal$. 
	
	For outer vertices, when input graph $G$ lies in support of $\cG_r(n_r)$, we know that the channel-degree of all these vertices is at most one, by \Cref{obs:G_r-degree}. For any outer vertex $u$ connected to inner vertex $x$, the entire input given to $u$ is known to $x$. Thus, $x$ can sample the messages that $u$ would send to $x$ only with its input. 
\end{proof}

\begin{claim}\label{clm:Dreal-Drealtil-close}
	Distributions $\cDreal$ and $\cDrealtil$ are close to each other:
	\[
		\tvd{\cDreal}{\cDrealtil} \leq 1/n_{r-1}.
	\]
\end{claim}
\begin{proof}
	By \Cref{obs:cDrealtil-input-same} and \Cref{fact:tvd-chain-rule}, 
	\begin{align*}
		\tvd{\cDreal}{\cDrealtil} &\leq \tvd{\cG_r(n_r)}{\cGtil_r(n_r)} + \Exp_{G \sim \cG_r(n_r)} \tvd{\cDreal(\rmsgall \mid G)}{\cDrealtil(\rmsgall \mid G)}  \\
		&\leq \frac1{n_{r-1}} +  \Exp_{G \sim \cG_r(n_r)} \tvd{\cDreal(\rmsgall \mid G)}{\cDrealtil(\rmsgall \mid G)}   = \frac1{n_{r-1}} ,
	\end{align*}
where the last inequality is by~\Cref{clm:cG-cGtil-close} and the equality follows from \Cref{obs:cDrealtil-messages}.
\end{proof}

\subsection{Proof of \Cref{thm:hard-dist} Barring Round Elimination}\label{sec:proof-hard-dist}

In the next section, we construct a round elimination protocol, with the following parameters. 

\begin{lemma}\label{lem:round-elim}
	For any $s,r \geq 1$ and $\delta \in (0,1)$, 
	given a deterministic protocol $\prot_r$ for instances $G_r \sim \cG_r(n_r)$, with the following parameters:
	\begin{equation*}
		\round(\prot_r) = r \qquad \bwidth(\prot_r) = s \qquad \suc(\prot_r, \cG_r(n_r))\geq \delta, 
	\end{equation*}
	we can construct a deterministic protocol $\prot_{r-1}$ which takes instances  $G_{r-1} \sim \cG_{r-1}(n_{r-1})$ and has the following parameters:
	\begin{equation*}
		\round(\prot_{r-1}) = r-1 \qquad \bwidth(\prot_{r-1}) \leq s \qquad \suc(\prot_{r-1}, \cG_{r-1}(n_{r-1}))\geq \delta - \frac1{n_{r-1}} - 15 \sqrt{\frac{s}{n_{r-1}}}.
	\end{equation*}
\end{lemma}

The proof of \Cref{lem:round-elim} is the focus of \Cref{sec:round-elim-descr}. We now prove \Cref{thm:hard-dist} using \Cref{lem:round-elim}.

\begin{proof}[Proof of \Cref{thm:hard-dist}]
	Assume towards a contradiction that there exists a protocol $\prot_r$ with $r$-rounds, 
	\begin{equation*}
		\suc(\prot_r, \cG_r(n_r)) \geq 15/16, \qquad \textnormal{ and }  \qquad \bwidth(\prot_r) = s < n_r^{(1/2) \cdot (1/34^r)} \cdot \frac1{(480)^2}.
	\end{equation*}
	We prove the theorem by repeatedly applying \Cref{lem:round-elim} for $r$ times on protocol $\prot_r$, to get a protocol $\prot_0$ for $\cG_0(n_0)$. 
	The value of $n_0$, using \Cref{eq:params-r-1} is $n_0 = {n_1}^{1/34} = \cdots = {n_r}^{1/34^r}$. 	From the large value of $n_r$ in the statement of the theorem, we have, 
	\begin{equation}\label{eq:inter-thm-1}
		n_0 = n_r^{1/34^r} > r^4. 
	\end{equation}
	The probability of success of $\prot_0$ is at least, 
	\begin{align*}
		&\suc(\prot_r, \cG_r(n_r)) - \sum_{\ell = 1}^r \frac1{n_{r-\ell}} - 15 \cdot \sqrt{s} \cdot  \sum_{\ell = 1}^r \frac1{(n_{r-\ell})^{1/2}} \\
		&\geq \frac{15}{16} - \sum_{\ell = 1}^r \frac1{n_{r-\ell}} - 15 \cdot \sqrt{s} \cdot  \sum_{\ell = 1}^r \frac1{(n_{r-\ell})^{1/2}} \tag{by value of $\suc(\prot_r, \cG_r(n_r))$ from the theorem statement} \\
		&\geq  \frac{15}{16} - r \cdot (\frac1{n_0}) - 15 \cdot \sqrt{s} \cdot  r \cdot (\frac1{n_0^{1/2}}) \tag{as $n_0, n_1, \ldots, n_{r-1}$ are increasing} \\
		&\geq  \frac{15}{16} - \frac{1}{32} - 15 \cdot \sqrt{s} \cdot  r \cdot (\frac1{n_0^{1/2}}) \tag{as $r  < 32 n_0$ from \Cref{eq:inter-thm-1}} \\
		&\geq \frac{29}{32} - 15 \cdot \frac1{480} \cdot r \cdot (1/n_0^{-0.25}) \tag{by value of $s$, and $n_0 = n_r^{1/34^r}$} \\
		&> \frac{29}{32} - \frac1{32} \tag{as $r < {n_0}^{1/4}$ from \Cref{eq:inter-thm-1}} > \frac78.
	\end{align*}
	This contradicts \Cref{clm:0-round-prot}, thereby proving the theorem.
\end{proof}

\clearpage

\section{Round Elimination: Description}
\label{sec:round-elim-descr}

In this section, we give the desired protocol $\prot_{r-1}$ in \Cref{lem:round-elim} for solving instances $G_{r-1} \sim \cG_{r-1}(n_{r-1})$, and provides parts of the analysis of the protocol.
This protocol constructs an instance $G_r \sim \cG_r(n_r)$ and embeds its input $G_{r-1}$ into it. It also samples the first round messages that protocol $\prot_r$ sends using the different sources of randomness in our model (see \Cref{subsec:prelim-model}).

We want to sample the input and messages from $\cDrealtil$. As we have established in \Cref{clm:Dreal-Drealtil-close}, this distribution is close to the correct distribution $\cDreal$. However, sampling from $\cDrealtil$ is not feasible either. 

We sample from a distribution we label as $\cDfake$, and show that this is close in total variation distance to $\cDrealtil$ and thus $\cDreal$. 
The description of $\prot_{r-1}$ is given in steps where the random variables in $\cDrealtil$ and $\cDfake$ are sampled progressively.

\subsection{Key Steps of the Protocol}
In this subsection, we describe some prominent steps of the protocol $\prot_{r-1}$.

\subsubsection{Public Random Variables}
Given a graph $G_{r-1} \sim\cG_{r-1}(n_{r-1})$, it is easy to sample new identites for each of the vertices using public randomness. It is also easy to sample the auxiliary random variables. This is the first step of the protocol. 
To break the correlation between messages that each inner vertex sends, we also sample some messages publicly. 

\begin{Protocol}\label{box:r-1-step1-ids}
	\textbf{Step $\bm{(1)}$ of protocol $\prot_{r-1}$ for $G_{r-1} \sim \cG_{r-1}(n_{r-1})$ given $\prot_{r}$ for $G_{r} \sim \cG_r(n_r)$}:
	\begin{enumerate}[label=$(\arabic*)$]
		\item Sample the random variables $\rids$ and also, $\rcJall, \rcKall$,  and $\rLall$ which jointly form $\raux$ using public randomness. 
		\item For $i \in [n_{r-1}]$, the vertex with identity $x_i$ in $G_{r-1}$ takes on the identity of $\xstar_i \in \Xstar$ in random variable $\rids$. It changes the identities of all its neighbors in $G_{r-1}$ correspondingly.
		\item For each inner vertex $x$ and layer $Y$ with $x \notin Y$, set $\Nbhood{x}{Y}{}[\Lupdown{x \to Y}{t, i}]$ to all be type $t$ for $t \in [r] \cup \{0\}$ and $i \in [n_{r-1}]$. 
		\item For each inner vertex $x$, and layer $Y$ with $x \notin Y$, sample the message that $x$ sends to vertex $y_j \in \Lupdown{x \to Y}{t, i} \cap \{<\ystar_i\}$, for all types $t \in [r] \cup \{0\}$, and $i \in [n_{r-1}]$, with public randomness.
	\end{enumerate}
	
	We use $\Npub{x}$ and $\Mpub{x}$ to denote the input and messages sent by each inner vertex $x$ that are sampled publicly in this protocol. 
\end{Protocol}

\begin{observation}\label{obs:public-msgs-number}
	The total number of messages sampled in $\Mpub{x}$ for each inner vertex $x$ is at most $2 \cdot \gamma_r \cdot (r+1) \cdot n_{r-1}$.
\end{observation}
\begin{proof}
The size of set $\Lupdown{x \to Y}{t, i}$ is exactly $\gamma_r $ for all $i \in [n_{r-1}]$ and type $t \in [r] \cup \{0\}$, and other layer $Y$ with $x \notin Y$. At most $\gamma_r $ many messages are sampled for each of the $r+1$ types and $n_{r-1}$ values of $i$, and the two choices of $Y$.
\end{proof}

We do not have the complete description of $\cDfake$ yet, this will be clear as we formalize more of protocol $\prot_{r-1}$. But we can prove that what we have sampled so far is the same way as in $\cDrealtil$. 

\begin{observation}\label{obs:real-fake-ids-same}
	The distribution of $\rG_{r-1}, \rids$ and $\raux$ are the same in $\cDrealtil$ and $\cDfake$. 
\end{observation}
\begin{proof}
	Random variable $\rG_{r-1}$ is the same in $\cDrealtil$ and $\cDfake$, as it is just sampled from $\cG_{r-1}(n_{r-1})$.
	
	In step (1) of \Cref{box:r-1-step1-ids}, the random variables $\rids$ and $\raux$ are jointly sampled independently of $\rG_{r-1}$. However, even in $\cDrealtil$, $\rG_{r-1}$ is independent of $ \rids, \raux$ by \Cref{obs:input-indep}-\ref{item:ind-prop-3}. 
\end{proof}

\subsubsection{Sampling Messages Between Inner Vertices}

We now use pair randomness to sample messages between inner vertices. 

In the description of $\cG_r(n_r)$ from \Cref{box:cGr-definition}, the type of each pair $(x, y)$ in $G_{r-1}$ is unchanged. The types of all these vertices is upper bounded by $r$ in graph $G_{r-1} \sim \cG_{r-1}(n_{r-1})$. Hence, all pairs $(x, y) \in G_{r-1}$ are present in $\cC_r$ of instance $G$ from $\cG_r(n_r)$, and they can send messages to each other in the first round of $\prot_r$. We describe how to sample these messages now. 



\begin{Protocol}\label{box:r-1-step2-inner-msgs} 
	\textbf{Step $\bm{(2)}$ of protocol $\prot_{r-1}$ for $G_{r-1} \sim \cG_{r-1}(n_{r-1})$ given $\prot_{r}$ for $G_{r} \sim \cG_r(n_r)$}:
	\vspace{2mm}
	
For each pair $x,y \in G_{r-1}$, with $y \notin X$,  sample the message that $x$ sends to $y$ in the first round conditioned on $\ids$, $\typeval{x, y}, \aux, \Npub{x}$ and $\Mpub{x}$ using pair randomness between $x$ and $y$.

\bigskip

We use $\Minner{x}$ to denote all the messages that inner vertex $x$ sends to some other inner vertex $y \in G_{r-1}$ and $y \notin X$. We use $\Minner{x}(y)$ to denote the message that $x$ sends to the inner vertex $y \in G_{r-1}$, and 
$\Minner{x}(-y)$ to denote all the messages that $x$ sends to other inner vertices that are not $y$. These messages are all sampled by this protocol.  
\end{Protocol}

\subsubsection{Sampling Rest of the Input and Messages}

We can now sample the rest of the random variables in distribution $\cDrealtil$.

\begin{Protocol}\label{box:r-1-step3-input-rest}
			\textbf{Step $\bm{(3)}$ of protocol $\prot_{r-1}$ for $G_{r-1} \sim \cG_{r-1}(n_{r-1})$ given $\prot_{r}$ for $G_{r} \sim \cG_r(n_r)$}:
	\vspace{2mm}

Each inner vertex $x_i \in G_{r-1}$ does the following:
		\begin{enumerate}[label=(\alph*)]
			\item Sample the rest of its input from distribution $\cDrealtil$ and the remaining messages that $x_i$ sends to its neighbors with private randomness conditioned on $\ids$, $\aux$, $\Npub{x_i}$, $\Mpub{x_i}$, $\Minner{x_i}$, and its inner input $\Ninner{x}{i}$. 
	
			\item Sample all the messages that $x_i$ receives from outer vertices, conditioned on the entire input of $x_i$, $\ids$ and $\aux$.
		\end{enumerate}
		
		We use $\Nrest{x_i}$ to denote all the random variables that $x_i$ samples in step (3) of  \Cref{box:r-1-step3-input-rest} using private randomness. This is comprised of the rest of the input of $x_i$, the messages it sends to outer neighbors, and all the messages it receives from its outer neighbors. 
\end{Protocol}

We can prove that the random variables $\Nrest{x_i}$ for each inner vertex $x_i$  are sampled from the correct distribution $\cDrealtil$.

\begin{observation}\label{obs:rest-input-correct}
In $\cDfake$, conditioned on any choice of $\rids$, $\raux$, $\rG_{r-1}$, $\rNpub{x}, \rMpub{x}$ and $\rMinner{x}$ for each inner vertex $x$, $\rNrest{y}$ for all inner vertices $y$ are jointly sampled from $\cDrealtil$. 
\end{observation}
\begin{proof}
	In $\cDrealtil$, for each inner vertex $x$, its input, all the messages it sends, and the messages it receives from outer vertices are sampled independently of the other inner vertices and independently of $\rG_{r-1}$ when conditioned on $\rNinner{x}{}, \rids$ and $\raux$ by \Cref{obs:input-indep}-\ref{item:ind-prop-1} and \ref{item:ind-prop-2}.
	
	Therefore, in step (3) of \Cref{box:r-1-step3-input-rest}, when we sample the remaining input of each inner vertex $x$ it is sufficient to condition on $\rNinner{x}{}$, $\rMinner{x}$, $\rMpub{x}$, $\rNpub{x}$, $\rids$ and $\raux$. This is true of the remaining messages sent by $x$ also. 
	
	From the definition of $\cDrealtil$ from \Cref{box:cDrealtil}, we know that all the messages that $x$ receives from its outer neighbors are sampled only conditioned on the entire input of $x$, the value of $\aux$ and $\ids$. This process is the same as in  step (3) in \Cref{box:r-1-step3-input-rest}.
\end{proof}

We have sampled all the random variables associated with $\cDrealtil$. However, the distribution of input and first round messages we get is different from $\cDrealtil$, even though parts of them may be the same from \Cref{obs:real-fake-ids-same} and \Cref{obs:rest-input-correct}.

\subsection{Full Protocol $\prot_{r-1}$ and its Distribution}\label{sec:fake-protocol}
We finally put together the multiple parts of protocol $\prot_{r-1}$, and see the full description of $\cDfake$.

\begin{Protocol}\label{box:r-1-final-all-steps}
\textbf{Protocol $\prot_{r-1}$ for $G_{r-1} \sim \cG_{r-1}(n_{r-1})$ given $\prot_{r}$ for $G_{r} \sim \cG_r(n_r)$}:

\begin{enumerate}[label=$(\arabic*)$]
\item The random variable $\rG_{r-1}$ is given as input. Random variables $\rids, \raux,$ and $ \rNpub{x}, \rMpub{x}$ for all inner vertices $x$ are sampled according to \Cref{box:r-1-step1-ids}.  Identities are changed appropriately as given in step (1) from \Cref{box:r-1-step1-ids}. 
\item All messages between pairs of inner vertices $x, y \in G_{r-1}$ which have a channel in $G_{r-1}$ are sampled as given in \Cref{box:r-1-step2-inner-msgs}. 
\item The remaining input for each inner vertex $x_i$, the other messages it sends and all the messages it receives from outer vertices are sampled in \Cref{box:r-1-step3-input-rest}.
\item All the vertices continue to run $\prot_r$ starting from the second round, assuming first round messages are as sampled, and output the same answer as $\prot_r$. 
 See \Cref{clm:run-round2} that specifies how to implement this step. 
 \end{enumerate}
\end{Protocol}

First, let us argue that the vertices can run the protocol $\prot_r$ starting from the second round. 

\begin{claim}\label{clm:run-round2}
	Given any input graph $G$ which lies in the support of $\cG_r(n_r)$, all inner vertices can run protocol $\prot_r$ correctly starting from second round based on the information they have access to. 
\end{claim}
\begin{proof}
	First, we argue that all the inner vertices have access to all the messages they send and receive in the first round. For any message sent by inner vertex $x$ to another inner vertex $y$, we use pair randomness between $x, y$ to sample it in step (2) from \Cref{box:r-1-step2-inner-msgs}. Thus, both $x, y$ have access to the message. All messages sent by any inner vertex $x$ to any other inner vertex is known to $x$. All messages received by $x$ sent by some other inner vertex $y$ is known to $x$ as well. 
	
	The remaining messages that $x$ sends to outer vertices, and all the messages it receives from outer vertices are sampled with private randomness in step (3) from \Cref{box:r-1-step3-input-rest}, and are all known to the inner vertex $x$. 
	
	When the input graph $G$ lies in the support of $\cG_r(n_r)$, the channel-degree of all outer vertices is at most one from \Cref{obs:G_r-degree}. Therefore, any inner vertex $x$ with outer neighbor $u$ can simulate the messages $u$ sends to $x$, as it has access to the entire input of $u$. 
	
	As for the inner vertices, they have access to all their inputs. We have argued they also have access to all the first round messages sent and received by them at the end of step (3) of $\prot_{r-1}$. They can continue to run protocol $\prot_r$ from the second round by simulating the behavior of all outer vertices connected to them. This proceeds exactly as in $\prot_r$ till the last round, and the inner vertices can find the output of $\prot_r$ (again, since they can fully simulate outer vertices). 
\end{proof}

Next, let us talk about the distribution $\cDfake$, which is the joint distribution of the input and first round messages that we generate when running protocol $\prot_{r-1}$. 
\begin{Distribution}\label{box:cDfake}
\textbf{Distribution $\cDfake$:}
\begin{align*}
	&\rG_{r-1} \times (\rids, \rcJall, \rcKall, \rLall) \tag{the inner graph, identities and auxiliaries}\\
	&\times (\bigtimes_{x \in G_{r-1}} \rNpub{x}, \rMpub{x} \mid \rids, \raux) \tag{the public parts of input and messages from step (1) in \Cref{box:r-1-step1-ids}} \\
	&\times (\bigtimes_{x \in G_{r-1}} \bigtimes_{\substack{y  \in G_{r-1} \\ y \notin X}} (\rMinner{x}(y) \mid  \rtypeval{x, y}, \rids, \raux, \rNpub{x}, \rMpub{x})   \tag{messages sent by inner vertices to each other in step (2) from \Cref{box:r-1-step2-inner-msgs}} \\
	&\times (\bigtimes_{x \in G_{r-1}}  \rNrest{x} \mid \rMinner{x}, \rNinner{x}{}, \rids, \raux, \rNpub{x}, \rMpub{x}). \tag{rest of the input and messages for inner vertices in step (3) from \Cref{box:r-1-step3-input-rest}}
\end{align*}
\end{Distribution}

\begin{claim}\label{clm:cDfake-prot-r-1}
	The distribution of the input and first round messages in $\prot_{r-1}$ is exactly $\cDfake$ as defined in \Cref{box:cDfake}.
\end{claim}
\begin{proof}
	We know that $\rids, \rcJall$, $\rcKall$ and $\rLall$ are sampled jointly with public randomness in $\prot_{r-1}$, independent of $\rG_{r-1}$. Therefore, their distributions are exactly as given in \Cref{box:cDfake}.
	
	The random variables $\rNpub{x}, \rMpub{x}$ are sampled with public randomness as well, but we know by \Cref{obs:input-indep}-\ref{item:ind-prop-1},\ref{item:ind-prop-2} that they are independent of random variables associated with other inner vertices. Hence, they are sampled only conditioned on $\rids$ and $\raux$. (There is no access to $\rNinner{x}{}$ with public randomness.)
	
	For each pair of inner vertices $x, y$, we sample $\rMinner{x}(y)$ in step (2) of \Cref{box:r-1-step2-inner-msgs} conditioned on exactly the random variables stated in the claim: $\rids$, $\raux$ from public randomness, $\rtypeval{x, y}$ that both inner vertices $x, y$ have access to, and $\rNpub{x}, \rMpub{x}$, again from public randomness. 
	
	Finally, each inner vertex $x$ samples the rest of its input, the remaining messages it sends and the messages it receives from outer vertices conditioned on $\rids, \raux$ from public randomness, $\rMinner{x}$ that it sampled using the pair randomness to other inner vertices, $\rNpub{x}, \rMpub{x}$ again from public randomness and $\rNinner{x}{}$ it has from graph $G_{r-1}$ in step (3) from \Cref{box:r-1-step3-input-rest}. This is as stated in the description of $\cDfake$ in \Cref{box:cDfake}.
\end{proof}

We can see that $\cDfake$ and $\cDrealtil$ are not the same distributions. But, we can prove that they are close to each other in total variation distance. This is the focus of the next section. 

\begin{lemma}[``$\cDfake$ and $\cDreal$ are not that different'']\label{lem:round-elim-distance}
\[
	\tvd{\cDreal}{\cDfake} \leq \frac1{(n_{r-1})} + 15 \sqrt{\frac{s}{n_{r-1}}}.
\]
\end{lemma}
Proof of \Cref{lem:round-elim-distance} can be found in \Cref{subsec:round-elim-distance} and constitutes the main technical part of our argument. 
This concludes the description of the round elimination protocol. 

\subsection{Proof of \Cref{lem:round-elim}}\label{subsec:proof-round-elim}
We can prove \Cref{lem:round-elim}, restated below, using \Cref{lem:round-elim-distance}.

\begin{lemma*}[Restatement of \Cref{lem:round-elim}]
	For any $s,r \geq 1$ and $\delta \in (0,1)$, 
	given a deterministic protocol $\prot_r$ for instances $G_r \sim \cG_r(n_r)$, with the following parameters:
	\begin{equation*}
		\round(\prot_r) = r \qquad \bwidth(\prot_r) = s \qquad \suc(\prot_r, \cG_r(n_r))\geq \delta, 
	\end{equation*}
	we can construct a deterministic protocol $\prot_{r-1}$ which takes instances  $G_{r-1} \sim \cG_{r-1}(n_{r-1})$ with the following parameters:
	\begin{equation*}
		\round(\prot_{r-1}) = r-1 \qquad \bwidth(\prot_{r-1}) \leq s \qquad \suc(\prot_{r-1}, \cG_{r-1}(n_{r-1}))\geq \delta - \frac1{n_{r-1}} - 15 \sqrt{\frac{s}{n_{r-1}}}.
	\end{equation*}
\end{lemma*}
\begin{proof}[Proof of \Cref{lem:round-elim}]
	We will prove that protocol $\prot_{r-1}$ from \Cref{box:r-1-final-all-steps} obeys the conditions in the statement of the lemma. 
	
	For the number of rounds in $\prot_{r-1}$, this is exactly $r-1$ as $\prot_r$ has $r$ rounds, and the first round messages are sampled without any communication between the vertices. 
	
	The total length of any message sent in the $r-1$ rounds of $\prot_{r-1}$ is at most $s$, as all the messages sent by $\prot_r$ are at most $s$ bits. 
	
	We know by \Cref{obs:r-triangle-exist} that for the graphs in the support of $\cG_r(n_r)$, if protocol $\prot_r$ outputs the correct answer, $\prot_{r-1}$ is correct in detecting whether a triangle exists as well. The output of $\prot_{r}$ and $\prot_{r-1}$ are the same when inputs are sampled from $\cDreal$, by \Cref{clm:run-round2}. Therefore, it is sufficient to lower bound the probability that $\prot_r$ outputs the correct answer when the inputs and first round messages are sampled from $\cDfake$ as opposed to $\cDreal$.
	
	We use \Cref{fact:tvd-small} to lower bound the probability of success of $\prot_{r-1}$ where $\event$ is the event that the answer on input $G_{r-1}$ is correct: 
	\begin{align*}
		\Pr_{\cG_{r-1}}[\prot_{r-1} \textnormal{ is correct}] &= \Pr_{\cDfake}[\prot_{r} \textnormal{ is correct}] \tag{by the discussion above} \\
		&\geq \Pr_{\cDfake}[\prot_{r} \textnormal{ is correct}] - \tvd{\cDreal}{\cDfake} \\
		&\geq \delta - \paren{\frac1{n_{r-1}} + 15 \sqrt{\frac{s}{n_{r-1}}}},
	\end{align*}
	where in the second inequality is by the assumption on the correctness probability of $\prot_r$ and~\Cref{lem:round-elim-distance}. 
	
	Finally, the protocol $\prot_{r-1}$ is randomized, but it can be made deterministic by fixing the random string by the easy direction of Yao's minimax principle, without changing its probability of success, rounds, and bandwidth. 
	The different sources of randomness do not cause an issue, as the randomness in all of these sources can be fixed to be a particular string. 
\end{proof}

\clearpage


\section{Round Elimination: Analysis}\label{sec:round-elim-analysis}
In this subsection, we analyze distributions $\cDfake$ (see \Cref{box:cDfake}) and $\cDreal$ (see \Cref{box:cDrealtil}), and prove the upper bound on their total variation distance in \Cref{lem:round-elim-distance}. 
We already know that distribution $\cDreal$ is close to $\cDrealtil$ by \Cref{clm:Dreal-Drealtil-close}. Hence, we focus on proving the distance between $\cDrealtil$ and $\cDfake$ is small.
When we talk about messages in this section, we only refer to messages sent over the first round. 

 Our analysis is split into three parts:

\begin{description}
	\item[Part 1.] The messages sent by each inner vertex $x$ are correlated with each other through the input of $x$. However, in step (2) of $\prot_{r-1}$ from \Cref{box:r-1-step2-inner-msgs}, they are sampled independently, only based on some parts of the input of $x$. We show that the correlation between these messages are low and thus can be sampled this way without too much loss.
	
	\item[Part 2.] We publicly sample some messages sent by each inner vertex $x$ for both layers $Y, Z$ in step (1) of $\prot_{r-1}$ from \Cref{box:r-1-step1-ids}, independent of the input of $x$. We remove the correlation between these messages and the input of $x$ from $G_{r-1}$ to allow for this sampling. 
	
	\item[Part 3.] Messages sent by $x$ to other inner vertices are sampled in step (2) of $\prot_{r-1}$ from \Cref{box:r-1-step2-inner-msgs}. These are sampled conditioned on some part of the input of $x$. We will show that the correlation between these messages and the rest of the input of $x$ from $G_{r-1}$ is low.
\end{description}

These are the only differences between $\cDrealtil$ and $\cDfake$. We have argued that the other random variables are sampled correctly in \Cref{obs:real-fake-ids-same} and \Cref{obs:rest-input-correct}.

\subsection{Setting up the Analysis}\label{subsec:setup-analysis}

We move from $\cDrealtil$ to $\cDfake$ by way of multiple distributions (one for each part) which lie between them and use a hybrid argument. In this subsection, we describe these hybrid distributions. 

First, let us cast distribution $\cDrealtil $ in terms of the random variables in $\cDfake$.

\begin{claim}\label{clm:cDrealtil-new-way}
	Distribution $\cDrealtil$ can be written as:
	\begin{align*}
		&\rG_{r-1} \times (\rids, \rcJall, \rcKall, \rLall) \tag{the inner graph, identities and auxiliaries}\\
		&\times (\bigtimes_{x \in G_{r-1}} \rNpub{x}, \rMpub{x} \mid \rids, \raux, \rNinner{x}{}) \tag{the public parts of input and messages from step (1) in \Cref{box:r-1-step1-ids}} \\
		&\times (\bigtimes_{x \in G_{r-1}} (\rMinner{x} \mid  \rids, \raux, \rNpub{x}, \rMpub{x}, \rNinner{x}{})   \tag{messages sent by inner vertices to each other in step (2) from \Cref{box:r-1-step2-inner-msgs}} \\
		&\times (\bigtimes_{x \in G_{r-1}}  \rNrest{x} \mid \rMinner{x}, \rNinner{x}{}, \rids, \raux, \rNpub{x}, \rMpub{x}). \tag{rest of the input and messages for inner vertices in step (3) from \Cref{box:r-1-step3-input-rest}}
	\end{align*}
\end{claim}

\begin{proof}
It is evident that random variables $\rG_{r-1}, \rids$ and $\raux$ are sampled the same way in the statement and in $\cDrealtil$. 

From \Cref{obs:input-indep}-\ref{item:ind-prop-1}, we know that the inputs of all the inner vertices are independent of each other conditioned on $\rG_{r-1}, \rids$ and $\raux$. Moreover, we know from \Cref{obs:input-indep}-\ref{item:ind-prop-2} that the input of any inner vertex $x$ only depends on the random variables $\rNinner{x}{}, \rids$ and $\raux$  in $\rG_{r-1}$. Therefore, the random variables $\rNpub{x}, \rMpub{x}$ for each inner vertex in $x$ are sampled only conditioned on $\rids, \raux$ and $\rNinner{x}{}$. 

As the protocol is determinstic, we know that the messages sent by any inner vertex $x$ depend only on its input. Therefore, $\rMinner{x}$ is also sampled from the right distribution in the statement of the claim. This argument applies to the rest of the input of $x$ and the messages sent by $x$ to outer vertices also. 

Lastly, in distribution $\cDrealtil$, all the messages received by inner vertex $x$ from outer vertices are sampled only based on the input of $x$, $\rids$ and $\raux$, as is the case in our statement. 
\end{proof}

To bound the distance between $\cDrealtil$ and $\cDfake$, we use weak chain rule of total variation distance from \Cref{fact:tvd-chain-rule} repeatedly. We define an ordering on the vertices and channels for this purpose. 

\paragraph{Ordering vertices and channels.}
We use a lexicographic ordering on the inner vertices of the form $a_1, a_2, \ldots, a_{n_{r-1}},$ then $b_1, \ldots, b_{n_{r-1}}$ and lastly $c_{1}, \ldots, c_{n_{r-1}}$. For each inner vertex $x$, to order the channels that $x$ has to inner vertices in the other two layers, we again use the lexicographic ordering we have over all the vertices. 
(The specific ordering is unimportant and we have picked the simplest one.)

We will now describe the hybrid distributions that we use as intermediate steps between the two distributions $\cDrealtil$ and $\cDfake$. 

\subsubsection*{Sampling Inner Messages Separately}

The first hybrid distribution which lies between $\cDrealtil$ and $\cDfake$ is $\cH_1$ that we define here.
One key difference between $\cDrealtil$ and $\cDfake$ is in how the messages to other inner vertices sent by any inner vertex $x$ are sampled. In $\cDrealtil$, they are sampled together, whereas this correlation is absent in $\cDfake$. We show that this correlation is small and can be broken. 

\begin{Distribution}\label{box:hybrid1-inner-msgs}
	\textbf{Distribution $\cH_1$:}
	\begin{align*}
		&\rG_{r-1} \times (\rids, \rcJall, \rcKall, \rLall) \tag{the inner graph, identities and auxiliaries}\\
			&\times (\bigtimes_{x \in G_{r-1}} \rNpub{x}, \rMpub{x} \mid \rids, \raux, \rNinner{x}{}) \tag{the public parts of input and messages from step (1) in \Cref{box:r-1-step1-ids}} \\
		&\times \eqcolor{(\bigtimes_{x \in G_{r-1}} \bigtimes_{\substack{y \in G_{r-1} \\y \notin X}} (\rMinner{x}(y) \mid  \rids, \raux, \rNpub{x}, \rMpub{x}, \rNinner{x}{})}   \tag{messages sent by inner vertices to each other in step (2) from \Cref{box:r-1-step2-inner-msgs}} \\
		&\times (\bigtimes_{x \in G_{r-1}}  \rNrest{x} \mid \rMinner{x}, \rNinner{x}{}, \rids, \raux, \rNpub{x}, \rMpub{x}). \tag{rest of the input and messages for inner vertices in step (3) from \Cref{box:r-1-step3-input-rest}}
	\end{align*}
\end{Distribution}

We prove distributions $\cDrealtil$ and $\cH_1$ are close to each other in \Cref{subsec:msgs-low-correlation}.

\begin{lemma}[``Messages sent by inner vertices can be sampled separately'']\label{lem:hybrid-1-inner-msgs}
	\[
	\tvd{\cDrealtil}{\cH_1} \leq 6 \cdot (n_{r-1})^{5/2} \cdot \sqrt{\frac{s}{\gamma_r + 1}}.
	\]
\end{lemma}

\subsubsection*{Sampling Public Messages without Inner Inputs}

The only difference between $\cH_1$ and $\cDfake$, is that in some parts of $\cH_1$ there is additional conditioning on some parts of the input $\rNinner{x}{}$ of the inner vertices. We show that these conditionings can be slowly broken. To this end, we give our second hybrid distribution that lies between $\cH_1$ and $\cDfake$, which samples public messages according to $\cDfake$.

\begin{Distribution}\label{box:hybrid2-public-msgs}
	\textbf{Distribution $\cH_2$:}
\begin{align*}
	&\rG_{r-1} \times (\rids, \rcJall, \rcKall, \rLall) \tag{the inner graph, identities and auxiliaries}\\
	&\times \eqcolor{(\bigtimes_{x \in G_{r-1}} \rNpub{x}, \rMpub{x} \mid \rids, \raux)} \tag{the public parts of input and messages from step (1) in \Cref{box:r-1-step1-ids}} \\
	&\times (\bigtimes_{x \in G_{r-1}}\bigtimes_{\substack{y \in G_{r-1} \\y \notin X}} (\rMinner{x}(y) \mid  \rids, \raux, \rNpub{x}, \rMpub{x}, \rNinner{x}{})   \tag{messages sent by inner vertices to each other in step (2) from \Cref{box:r-1-step2-inner-msgs}} \\
	&\times (\bigtimes_{x \in G_{r-1}}  \rNrest{x} \mid \rMinner{x}, \rNinner{x}{}, \rids, \raux, \rNpub{x}, \rMpub{x}). \tag{rest of the input and messages for inner vertices in step (3) from \Cref{box:r-1-step3-input-rest}}
\end{align*}
\end{Distribution}

We prove the following lemma in \Cref{subsec:public-msgs}.
\begin{lemma}[``Low correlation between first round messages and inputs of inner vertices'']\label{lem:cH1-cH2-public-msgs}
	\[
		\tvd{\cH_1}{\cH_2} \leq 3n_{r-1}  \cdot \sqrt{\frac{s \cdot \gamma_r  \cdot (r+1) \cdot n_{r-1}}{\alpha_r + 1}}.
	\]
\end{lemma}

\subsubsection*{Sampling Inner Messages with only Some Inner Inputs}

We come to the last part of our analysis, which shows that distribution $\cH_2$ is close to $\cDfake$. At this stage, the only difference between them is in how the messages between inner vertices are sampled. In $\cH_2$, when sampling the message that $x_i \in G_{r-1}$ sends to $y_j$, we condition on the entire input of the inner vertex $x_i$, whereas in $\cDfake$, we only condition on the input that both inner vertices $x_i$ and $y_j$ are aware of. 
We show that these distributions are close together (see~\Cref{box:cDfake} for the distribution $\cDfake$). 

\begin{lemma}[``Low correlation between messages of inner vertices and parts of their inputs'']\label{lem:cH2-cDfake-inner-msgs}
	\[
		\tvd{\cH_2}{\cDfake} \leq 6(n_{r-1})^2 \cdot  \sqrt{\frac{s}{2(\beta_r + 1)}}.
	\]
\end{lemma}

The proof of \Cref{lem:cH2-cDfake-inner-msgs} is presented in \Cref{subsec:inner-msgs}. 

\subsection{Proof of \Cref{lem:round-elim-distance}}\label{subsec:round-elim-distance}

Proof of \Cref{lem:round-elim-distance} now follows from  \Cref{lem:hybrid-1-inner-msgs,lem:cH1-cH2-public-msgs,lem:cH2-cDfake-inner-msgs} with some minimal calculations using the values of $\alpha_r$, $\beta_r$ and $\gamma_r$ from \Cref{eq:params-alphabetagamma}.

\begin{lemma*}[Restatement of \Cref{lem:round-elim-distance}]
	\[
	\tvd{\cDreal}{\cDfake} \leq \frac1{n_{r-1}} + 15 \sqrt{\frac{s}{n_{r-1}}}.
	\]
\end{lemma*}

\begin{proof}[Proof of \Cref{lem:round-elim-distance}]
	By~\Cref{clm:Dreal-Drealtil-close}, we have
	\[
		\tvd{\cDreal}{\cDrealtil} \leq \frac1{n_{r-1}}.
	\]
	Thus, we can focus on bounding the distance between  $\cDrealtil$ and $\cDfake$ and use triangle inequality.  
	\begin{align*}
		\tvd{\cDrealtil}{\cDfake} &\leq \tvd{\cDrealtil}{\cH_1} + \tvd{\cH_1}{\cH_2} + \tvd{\cH_2}{\cDfake} \tag{by triangle inequality} \\
		&\leq 6 \cdot \sqrt{\frac{s \cdot (n_{r-1})^5}{ (n_{r-1})^6}} +  \tvd{\cH_1}{\cH_2} + \tvd{\cH_2}{\cDfake} \tag{by \Cref{lem:hybrid-1-inner-msgs} and \Cref{eq:params-alphabetagamma}} \\
		&\leq 6 \sqrt{\frac{s}{n_{r-1}}} + 3n_{r-1}  \cdot \sqrt{\frac{s \cdot \gamma_r  \cdot (r+1) \cdot n_{r-1}}{\alpha_r + 1}} + \tvd{\cH_2}{\cDfake} \tag{by \Cref{lem:cH1-cH2-public-msgs}} \\
		&\leq 6 \sqrt{\frac{s}{n_{r-1}}} + 3  \cdot \sqrt{\frac{s \cdot(n_{r-1})^6  \cdot (r+1) \cdot (n_{r-1})^3}{2 \cdot (n_{r-1})^{11}}}  + \tvd{\cH_2}{\cDfake}  \tag{by \Cref{eq:params-alphabetagamma}} \\
		&\leq  6 \sqrt{\frac{s}{n_{r-1}}} + 3  \cdot \sqrt{\frac{ s \cdot(n_{r-1})^6 \cdot (n_{r-1}) \cdot (n_{r-1})^3}{ (n_{r-1})^{11}}}  + \tvd{\cH_2}{\cDfake} \tag{as $ r + 1 \leq n_{r-1}$ for large $n_r$} \\
		& \leq 9 \sqrt{\frac{s}{n_{r-1}}} +   6(n_{r-1})^2 \cdot  \sqrt{\frac{s}{2(\beta_r + 1)}} \tag{by \Cref{lem:cH2-cDfake-inner-msgs}} \\
		&\leq  9 \sqrt{\frac{s}{n_{r-1}}} + 6 \cdot  \sqrt{\frac{s \cdot (n_{r-1})^4}{2((n_{r-1})^5+ 1)}} \tag{by \Cref{eq:params-alphabetagamma}} \\
		&\leq 15 \sqrt{\frac{s}{n_{r-1}}},
	\end{align*}
	concluding the proof. 
\end{proof}
The rest of this section contains the heart of our whole argument, namely, proving different steps of the round elimination protocol incur low losses. 

\subsection{Messages from One Vertex have Low Correlation}\label{subsec:msgs-low-correlation}
In this subsection, we prove \Cref{lem:hybrid-1-inner-msgs}.
Let us recall the two distributions $\cDrealtil$ and $\cH_1$ from \Cref{box:cDrealtil} and \Cref{box:hybrid1-inner-msgs} respectively. 

\vspace{3mm}
\hspace{-1.5em}
\begin{minipage}{0.49\textwidth}
	\small
		\textbf{Distribution $\cDrealtil$:}
		\begin{align*}
		&\rG_{r-1} \times (\rids, \rcJall, \rcKall, \rLall) \\
		&\times (\bigtimes_{x \in G_{r-1}} \rNpub{x}, \rMpub{x} \mid \rids, \raux, \rNinner{x}{})\\
		&~\eqcolor{\times (\bigtimes_{x \in G_{r-1}} (\rMinner{x} \mid  \rids, \raux, \rNpub{x}, \rMpub{x}, \rNinner{x}{})}  \\
		&\times (\bigtimes_{x \in G_{r-1}}  \rNrest{x} \mid \rMinner{x}, \rNinner{x}{}, \rids, \raux, \rNpub{x}, \rMpub{x}). 
	\end{align*}
\end{minipage}
\vrule width 1pt 
\hspace{0.5em} 
\begin{minipage}{0.49\textwidth}
	\small
		\textbf{Distribution $\cH_1$:}
	\begin{align*}
		&\rG_{r-1} \times (\rids, \rcJall, \rcKall, \rLall)\\
		&\times (\bigtimes_{x \in G_{r-1}} \rNpub{x}, \rMpub{x} \mid \rids, \raux, \rNinner{x}{})  \\
		&~\eqcolor{\times (\bigtimes_{x \in G_{r-1}}\bigtimes_{\substack{y \in G_{r-1} \\y \notin X}}  (\rMinner{x}(y) \mid  \rids, \raux, \rNpub{x}, \rMpub{x}, \rNinner{x}{})}   \\
		&\times (\bigtimes_{x \in G_{r-1}}  \rNrest{x} \mid \rMinner{x}, \rNinner{x}{}, \rids, \raux, \rNpub{x}, \rMpub{x}). 
	\end{align*}
\end{minipage}
\vspace{3mm}

The difference is in the third line, in how the messages sent by $x$ to other inner vertices in $G_{r-1}$ are sampled. In $\cDrealtil$, they are jointly sampled conditioned on $\rids, \raux$, $\rNpub{x}, \rMpub{x}$ and $\rNinner{x}{}$. In $\cH_1$, they are sampled independently of each other, conditioned on the same random variables $\rids, \raux$, $\rNpub{x}, \rMpub{x}$ and $\rNinner{x}{}$. 

\paragraph{Notation.}
We use $\Ninner{x}{i}(y)$ to denote the $type(x_i, y)$ in $G_{r-1}$ for inner vertex $y$ in the input of $x_i$. 

\begin{claim}\label{clm:msgs-mi-low}
For every inner vertex $x \in G_{r-1}$, and every other inner vertex $y_i \in Y$ for $i \in [n_{r-1}]$, 
\[
	\mi{\rMinner{x}(y_i)}{\rMinner{x}(-y_i)\mid  \rNinner{x}{}, \rids, \raux, \rNpub{x}, \rMpub{x}} \leq  \frac1{\gamma_r+1} \cdot 2n_{r-1} \cdot s.
\]
\end{claim}
\begin{proof}
	By our notation, we have $\rid{y_i} = \ystar_i$. 
	First, we have, 
	\begin{align*}
			&\mi{\rMinner{x}(y_i)}{\rMinner{x}(-y_i)\mid  \rNinner{x}{}, \rids, \raux, \rNpub{x}, \rMpub{x}} \\
			&\mi{\rMinner{x}(y_i)}{\rMinner{x}(-y_i)\mid  \rNinner{x}{}, \rids, \raux, \rNpub{x}, \rMpub{x}, \eqcolor{\rtypeval{x, y_i}}} \tag{as $\rtypeval{x, y_i}$ is fixed by $\rNinner{x}{}$} \\
			&= 	\mi{\eqcolor{\rmsg{x}{Y}{}[\rid{y_i}]}}{\rMinner{x}(-y_i)\mid \rNinner{x}{}, \rids, \raux, \rNpub{x}, \rMpub{x}, \rtypeval{x, y_i}} \tag{as $\rid{y_i}$ is the vertex in $G$ to which $y_i$ is mapped} \\
			&= \eqcolor{\Exp_{t \sim \rtypeval{x, y_i}}}\mi{\rmsg{x}{Y}{}[\rid{y_i}]}{\rMinner{x}(-y_i)\mid \rNinner{x}{}, \rids, \raux, \rNpub{x}, \rMpub{x}, \eqcolor{\rtypeval{x, y_i}= t}} \tag{by the definition of conditional mutual information}.
	\end{align*}
	We use $\rWrest$ to be the joint random variable $\rNinner{x}{}, \rid{w}$ for each $w\neq y_i$ and $w \neq x$, $\rcJall, \rcKall$, all random variables in $\rLall$ barring $\rLupdown{x \to Y}{t, i}$, all random variables in $\rNpub{x}$ barring $\rNbhood{x}{Y}{}[\rLupdown{x \to Y}{t, i}]$, all random variables in $\rMpub{x}$ barring $\rmsg{x}{Y}{}[y_j]$ for $y_j \in \rLupdown{x \to Y}{t, i} \cap \{<\ystar_i\}$. 
	These random variables are bundled together because they are not relevant to the rest of the proof.

	We prove the statement for each type $t$ separately. Let $\Et$ denote the event that $\rtypeval{x, y_i} = t$. 
	We omit the superscript $ \rightarrow Y$ and replace $x \rightarrow Y$ by $\vecx$ to avoid the clutter. 
	\begin{align*}
		&	\mi{\rmsg{x}{Y}{}[\rid{y_i}]}{\rMinner{x}(-y_i)\mid \rNinner{x}{}, \rids, \raux, \rNpub{x}, \rMpub{x}, \rtypeval{x, y_i}= t} \\
		&=	\mi{\rmsg{x}{Y}{}[\rid{y_i}]}{\rMinner{x}(-y_i)\mid \eqcolor{\rid{y_i}, \rid{x}, \rLupdown{x \to Y}{t, i}, \rNbhood{x}{Y}{}[\rLupdown{x \to Y}{t, i}], \rmsg{x}{Y}{}[\rLupdown{x \to Y}{t, i} \cap \{< \rid{y_i}\}], \rWrest, \Et}} \\
		&=	\mi{\rmsgup{\vecx}[\rid{y_i}]}{\rMinner{x}(-y_i)\mid \rid{y_i}, \rid{x}, \rLupdown{\vecx}{t, i}, \rNbhoodup{\vecx}[\rLupdown{\vecx}{t, i}] ,\rmsgup{\vecx}[\rLupdown{\vecx}{t, i} \cap\{ <\rid{y_i}\}], \rWrest, \Et}  \tag{changing superscript $x \to Y$ to $\vecx$ for readability}\\
		&=	\mi{\rmsgup{\vecx}[\rid{y_i}]}{\rMinner{x}(-y_i)\mid \rid{y_i}, \rid{x}, \rLupdown{\vecx}{t, i}, \eqcolor{\rNbhoodup{\vecx}[\rLupdown{\vecx}{t, i} \cup \{\rid{y_i}\}]},\rmsgup{\vecx}[\rLupdown{\vecx}{t, i} \cap \{<\rid{y_i}\}], \rWrest, \Et}  \tag{as $\rNbhoodup{\vecx}[\rid{y_i}]$ is fixed to be $t$, conditioned on $\Et$ and thus conditioning on $\rNbhoodup{\vecx}[\rid{y_i}]$ is w.l.o.g.}\\
		 &= 	\mi{\rmsgup{\vecx}[\rid{y_i}]}{\rMinner{x}(-y_i)\mid \rid{y_i}, \rid{x}, \eqcolor{\rLupdown{\vecx}{t, i} \cup \{\rid{y_i}\}}, \rNbhoodup{\vecx}[\rLupdown{\vecx}{t, i} \cup \{\rid{y_i}\}], \\
		 	&\hspace{85mm}\rmsgup{\vecx}[\rLupdown{\vecx}{t, i} \cap\{ < \rid{y_i}\}],\rWrest, \Et} \tag{as $\rLupdown{\vecx}{t, i}, \rid{y_i}$ are fixed by $\rLupdown{\vecx}{t, i} \cup \{\rid{y_i}\}, \rid{y_i}$ and vice-versa} \\
		 &= \eqcolor{\Exp_{\rLupdown{\vecx}{t, i} \cup \{\rid{y_i}\} = P}}	\mi{\rmsgup{\vecx}[\rid{y_i}]}{\rMinner{x}(-y_i)\mid \rid{y_i}, \rid{x}, \eqcolor{\rNbhoodup{\vecx}[P], \rmsgup{\vecx}[P\cap \{<\rid{y_i}\}]}, \\
		 	&\hspace{85mm}\rWrest, \eqcolor{\rLupdown{\vecx}{t, i} \cup \{\rid{y_i}\} = P}, \Et}. \tag{by the definition of conditional mutual information}
		\end{align*}
		Again, we prove the statement for every set $P \subset [n_r]$ of size $\gamma_r + 1$. Let $\Econd$ be the event that $\rLupdown{\vecx}{t, i} \cup \{\rid{y_i}\} = P$ and $\Et$ happen. 
		
		We argue that conditioned $\Econd$, the value of $\rid{y_i}$ is uniform over the set $P$. This is because, by the definition of $\rLupdown{\vecx}{t, i}$, for all $y_j \in \rLupdown{\vecx}{t, i}$, the value of $\Nbhood{x}{Y}{}[y_j]  = t$, similar to $\rNbhood{x}{Y}{}[\rid{y_i}]$. 
		The random variable $\rid{y_i}$ is uniform over $[n_r]$, and $\rLupdown{\vecx}{t, i}$ is a uniformly random set of size $\gamma_r$ that does not contain $\rid{y_i}$. If $P$ is the set $\rLupdown{\vecx}{t, i} \cup \{\rid{y_i}\}$, the value of $\rid{y_i}$ can be any value in this set with equal probability. As such, continuing the above equations for any fixed $P$ and $\Econd$, we will get, 
		\begin{align*}
				&\mi{\rmsgup{\vecx}[\rid{y_i}]}{\rMinner{x}(-y_i)\mid \rid{y_i}, \rid{x}, \rmsgup{\vecx}[P\cap \{<\rid{y_i}\}], \rWrest, \Econd} \\
				&= \eqcolor{\Exp_{y_j \in P}} \mi{\rmsgup{\vecx}[y_j]}{\rMinner{x}(-y_i)\mid \rid{x}, \rmsgup{\vecx}[P\cap \{<y_j\}], \rWrest,  \eqcolor{\rid{y_i} = y_j}, \Econd} \\
				&= \eqcolor{\frac1{\gamma_r+1} \cdot \sum_{y_j \in P}} \mi{\rmsgup{\vecx}[y_j]}{\rMinner{x}(-y_i)\mid \rid{x}, \rmsgup{\vecx}[P\cap \{<y_j\}], \rWrest,  \rid{y_i} = y_j, \Econd}.  \tag{as $\rid{y_i}$ is uniform over $P$}
			\end{align*}
	Now, we argue that the event $\rid{y_i}  = y_j$ is independent of the joint distribution of all the other random variables in the mutual information term, conditioned on $\Econd$. 
	Let us list these random variables, and argue in steps (in each step, we condition on everything in the previous steps also). 
	\begin{itemize}
		\item $\rid{w}$ for $w \neq y_i$: these are disjoint from $P$ and independent of the identity of $y_i$ inside $P$. 
		\item $\rNbhood{x}{Y}{}[P]$: this is deterministically fixed to all be $t$ conditioned on $\Econd$.
		\item $\rNbhood{x}{Y}{}$: the input of $x$ to layer $Y$ is already fixed to be type $t$ for all elements of $P$. The rest of the input is independent of \emph{which} of these choices are the actual identity of $y_i$. 
		\item $\rNbhood{x}{Z}{}$: the input of $x$ to layer $Z$ is independent of which of the identities in $P$ belong to $y_i$. 
		\item $\rmsg{x}{Y}{}, \rmsg{x}{Z}{}$: we know the inputs of $x$ are independent of the event $\rid{y_i} = y_j$. As protocol $\prot_r$ is deterministic, the messages sent by $x$ are independent of this event also. 
		\item $\rcJall, \rcKall$: these are sets which are disjoint from $P$, and are also independent of $\rid{y_i} = y_j$.
		\item $\rNinner{x}{}$: this input to $x$ comes from $G_{r-1}$, which is independent of the identity of $y_i$. 
		\item $\rNpub{x}$ barring $\rNbhood{x}{Y}{}[P]$: this is comprised of inputs chosen for $x$ on sets disjoint from $P$ and are independent of which identity in $P$ is given to $y_i$. 
		\item $\rLall$ barring $\rLupdown{x \to Y}{t, i}$: these are all sets disjoint from $P$, independent of what happens inside $P$. 
		\item $\rWrest$: we have argued that all the random variables in $\rWrest $ are independent of $P$.
	\end{itemize}
	As such, the joint distribution of these random variables is independent of the event $\rid{y_i} = y_j$ conditioned on $\Econd$. 
	Hence, we can continue as,
		\begin{align*}
			& \frac1{\gamma_r+1} \cdot \sum_{y_j \in P}  \mi{\rmsgup{\vecx}[y_j]}{\rMinner{x}(-y_i)\mid \rid{x}, \rmsgup{\vecx}[P\cap \{<y_j\}], \rWrest,  \rid{y_i} = y_j, \Econd} \\
			& = \frac1{\gamma_r+1} \cdot \sum_{y_j \in P}  \mi{\rmsgup{\vecx}[y_j]}{\rMinner{x}(-y_i)\mid \rid{x}, \rmsgnos[P\cap \{<y_j\}], \rWrest, \Econd} \tag{as $\rid{y_i} = y_j$ is independent of all random variables together conditioned on $ \Econd$} \\
			&=  \frac1{\gamma_r+1} \cdot  \mi{\rmsgup{\vecx}[P]}{\rMinner{x}(-y_i)\mid \rid{x}, \rWrest, \Econd} \tag{by the chain rule of mutual information in \itfacts{chain-rule}} \\
			&\leq \frac1{\gamma_r + 1} \cdot \en{\rMinner{x}(-y_i) \mid \Econd} \tag {by \itfacts{info-entropy}} \\
			&\leq \frac1{\gamma_r + 1} \cdot 2n_{r-1} \cdot s. \tag{as $\rMinner{x}(-y_i)$ has $2n_{r-1}-1$ messages of length at most $s$ sent by $x$}
	\end{align*}
	Putting things together, we have, 
	\begin{align*}
	&	\mi{\rMinner{x}(y_i)}{\rMinner{x}(-y_i)\mid  \rNinner{x}{}, \rids, \raux, \rNpub{x}, \rMpub{x}} \\
	&\hspace{30pt}\leq \Exp_{\substack{\rtypeval{x, y_i} \\ = t}} ~ \Exp_{\substack{\rLupdown{x \to Y}{t, i} \cup \{\rid{y_i}\} \\ = P}} \bracket{\frac1{\gamma_r + 1} \cdot 2n_{r-1} \cdot s} \tag{by all the bounds above} \\
	&\hspace{30pt}= \frac1{\gamma_r + 1} \cdot 2n_{r-1} \cdot s.
	\end{align*}
This completes the proof.
\end{proof}

We can now prove \Cref{lem:hybrid-1-inner-msgs} using the weak chain rule of total variation distance. First, let us recall this chain rule from \Cref{fact:tvd-chain-rule}. 
For any two distributions $\mu, \nu$ on $k$ random variables $\rw^1, \rw^2, \ldots, \rw^k$, we have, 
\[
	\tvd{\mu}{\nu} \leq \sum_{i \in [k]} \Exp_{\rw^{<i} \sim \mu} \tvd{\mu(\rw^i \mid \rw^{<i})}{\nu(\rw^{i} \mid \rw^{<i} )}. 
\]
To bound the distance between $\cDrealtil$ from \Cref{box:cDrealtil} and $\cH_1$ from \Cref{box:hybrid1-inner-msgs}, we use the lexicographic ordering on all the vertices and inner channels.

We need to define some more random variables. These random variables are local to this subsection. They may be used for different purposes later. (See \Cref{sec:list-rv} for a list of global random variables.)
\begin{itemize}
	\item Variable $\rwstart$: This is the joint random variable $\rG_{r-1}, \rids$, $\rcJall, \rcKall, \rLall$ along with $\rNpub{x}, \rMpub{x}$ for each inner vertex $x$. 
	\item Variables $\rw^{x \to y}$ for inner vertices $x, y$ in different layers : this is the random variable $\rMinner{x}(y)$ for inner vertices $x, y$.  
	\item Variables $\rw^{x}$ for each inner vertex $x$: this is the joint random variable of $\rw^{x \to y}$ for all inner vertices $y$ which are not in the same layer as $x$. We use the lexicographic ordering defined earlier to order these random variables based on $y$.  
	\item Variables $\rwend$: Joint random variable $\rNrest{x}$ for each inner vertex $x$. 
\end{itemize}

We need some simple observations.
\begin{observation}\label{obs:inter-hybrid-list}
About random variables associated with $\cDrealtil$ and $\cH_1$:
\begin{enumerate}[label=$(\roman*)$]
\item \label{item:hybrid1-1}	The random variable $\rwstart$ is distributed the same way in $\cDrealtil$ and $\cH_1$.
\item \label{item:hybrid1-2} In both $\cDrealtil$ and $\cH_1$, $\rw^{w} \perp \rw^{w'} \mid \rwstart$ for any distinct pair of inner vertices $w, w' \in A \cup B \cup C$. This is true regardless of whether $w, w'$ are in the same layer or in different layers. 
\item \label{item:hybrid1-3} 	Conditioned on any choice of random variable $\rwstart$ and $\rw^{x}$ for all inner vertices $x$, distribution of $\rwend$ is the same in $\cDrealtil$ and $\cH_1$. 
\item \label{item:hybrid1-4} In distribution $\cDrealtil$ and $\cH_1$, for any inner vertex $x$, $\rw^{x} \perp \rwstart \mid \rNinner{x}{}, \rids, \raux, \rNpub{x}, \rMpub{x}$. 
\item \label{item:hybrid1-5} In distribution $\cH_1$, for any inner vertex $x$, $\rw^{x\to w} \perp \rw^{x \to w'} \mid \rNinner{x}{}, \rids, \raux, \rNpub{x}, \rMpub{x}$, for any distinct pair of inner vertices $w, w' \in Y \cup Z$ of $G_{r-1}$. Again, $w$ and $w'$ may be in the same layer or different layers.
\end{enumerate}
\end{observation}
\begin{proof}
Part \ref{item:hybrid1-1} is apparent from the definition of distributions $\cDrealtil$ and $\cH_1$, as all random variables in $\rwstart$ are sampled the same way in $\cDrealtil$ and $\cH_1$.

For part \ref{item:hybrid1-2}, we know that when $\rw^{w}$ is sampled in $\cDrealtil$, it is sampled conditioned on random variables $\rids, \raux,  \rNinner{w}{}$, $\rNpub{w}$, and $\rMpub{w}$ and $\rids$. All these random variables are fixed by $\rwstart$. Thus, $\rw^{w}$ and $\rw^{w'}$ are independent for two distinct inner vertices $w, w'$ conditioned on $\rwstart$ in distribution $\cDrealtil$. 

Similarly in $\cH_1$, random variable $\rw^{w}$ is sampled only conditioned on $\rids, \raux,  \rNinner{w}{}$, $\rNpub{w}$, and $\rMpub{w}$, all of which are fixed by $\rwstart$.

Part \ref{item:hybrid1-3} is clear, again from the definition of $\cDrealtil$ and $\cH_1$ as $\rNrest{x}$ is sampled the same way in the two distributions for all inner vertices $x$ for any choice of the other random variables. 

Part \ref{item:hybrid1-4} is evident from how $\rw^{x}$ are sampled. They are conditioned only on $\rids, \raux,  \rNinner{x}{}$, $\rNpub{x}$, and $\rMpub{x}$ and are independent of the other random variables in $\rwstart$.

For part \ref{item:hybrid1-5}, in distribution $\cH_1$, we know that when $\rw^{x \to w} = \rMinner{x}(w)$ is sampled, it is conditioned only on $\rids, \raux,  \rNinner{x}{}$, $\rNpub{x}$, and $\rMpub{x}$. It is independent of the messages sent by $x$ to other inner vertices, conditioned on these random variables. 
\end{proof}

First, we show that the distribution of $\rw^{x}$ in $\cDrealtil$ and $\cH_1$ are close to each other for all inner vertices $x$, conditioned on $\rwstart$. 

\begin{claim}\label{clm:inter-hybrid-1}
		For any inner vertex $x$, we have, 
		\[
			\Exp_{\rwstart \sim \cDrealtil} \tvd{\cDrealtil(\rw^{x} \mid \rwstart )}{\cH_1(\rw^{x} \mid \rwstart )} \leq 2 \cdot (n_{r-1})^{3/2} \cdot \sqrt{\frac{s}{\gamma_r + 1}}.
		\]
\end{claim}
\begin{proof}
Let $u_1, u_2, \ldots, u_{2n_{r-1}}$ be all the inner vertices not in the same layer as $x$, following the ordering we defined.
We use $\rw^j$ to denote the random variable $\rw^{x\to u_j}$ for $j \in [2n_{r-1}]$. We use $\rw^0$ to denote $\rwstart$. We use $\rw^{<j}$ to denote the random variables $\rw^0, \rw^1, \ldots \rw^{j-1}$.
We have, 
\begin{align*}
&\Exp_{\rwstart \sim \cDrealtil} \tvd{\cDrealtil(\rw^{x} \mid \rwstart )}{\cH_1(\rw^{x} \mid \rwstart )} \\
&\leq  \sum_{j \in [2n_{r-1}]} \Exp_{\rw^{<j}\sim \cDrealtil} \tvd{\cDrealtil(\rw^j \mid \rw^{<j})}{\cH_1(\rw^j \mid \rw^{<j} )} \tag{by the weak chain rule of total variation distance \Cref{fact:tvd-chain-rule}}\\
&\leq \frac1{\sqrt{2}} \cdot \sum_{j \in [2n_{r-1}]}  \Exp_{\rw^{<j}\sim \cDrealtil} \sqrt{\kl{\cDrealtil(\rw^j \mid \rw^{<j} )}{\cH_1(\rw^j \mid \rw^{<j} )}} \tag{by Pinsker's inequality \Cref{fact:pinskers}} \\
&\leq \frac1{\sqrt{2}} \cdot \sum_{j \in [2n_{r-1}]}  \sqrt{ \Exp_{\rw^{<j}\sim \cDrealtil} \kl{\cDrealtil(\rw^j \mid \rw^{<j})}{\cH_1(\rw^j \mid \rw^{<j} )}}. \tag{by Jensen's inequality and concavity of square root} 
\end{align*}

We bound the KL-divergence term separately for each $j \in [2n_{r-1}]$. 
Let us use $\rMinner{x}(< u_{j})$ to denote the random variables $\rMinner{x}(u_1), \rMinner{x}(u_2), \ldots, \rMinner{x}(u_{j-1})$.  We use $\rWcond$ to denote the random variables $\raux, \rids,$ $ \rNinner{x}{}, \rNpub{x}$, and $\rMpub{x}$ as they always appear in the conditioning.

We know by \Cref{obs:inter-hybrid-list}-\ref{item:hybrid1-4} that the distribution of $\rw^j \mid \rw^{<j}$ in $\cDrealtil$ is,
\[
	\cDrealtil(\rw^j \mid \rw^{<j}) = \rMinner{x}(u_j) \mid \rMinner{x}(< u_j),  \rWcond.
\]
We know by \Cref{obs:inter-hybrid-list}-\ref{item:hybrid1-5} and \ref{item:hybrid1-4} that the distribution of $\rw^j \mid \rw^{<j}$ in $\cH_1$ is, 
\[
	\cH_1(\rw^j \mid \rw^{<j}) =  \rMinner{x}(u_j) \mid \rWcond.
\]
Combining the above two equations gives us, 
\begin{align*}
&\Exp_{\rw^{<j}\sim \cDrealtil}	\kl{\cDrealtil(\rw^j \mid \rw^{<j})}{\cH_1(\rw^j \mid \rw^{<j} )} \\
&\hspace{30pt}= \Exp_{\rMinner{x}(<u_j), \rWcond} \kl{\rMinner{x}(u_j) \mid\rMinner{x}(<u_j),\rWcond }{\rMinner{x}(u_j) \mid \rWcond} \\
&\hspace{30pt}\leq \mi{\rMinner{x\to u_j}}{\rMinner{x}(<u_j) \mid  \rWcond} \tag{by \Cref{fact:kl-info}} \\
&\hspace{30pt}\leq  \mi{\rMinner{x\to u_j}}{\rMinner{x}(-u_j)  \mid \rWcond} \tag{by \itfacts{data-processing}, as $\rMinner{x}(< u_j)$ is fixed by $\rMinner{x}(- u_j)$} \\
&\hspace{30pt}=   \mi{\rMinner{x\to u_j}}{\rMinner{x}(-u_j)  \mid\rids, \raux, \rNinner{x}{}, \rNpub{x}, \rMpub{x}} \tag{by definition of $\rWcond$} \\
&\hspace{30pt}\leq s \cdot 2n_{r-1} \cdot 1/(\gamma_r + 1) . \tag{by \Cref{clm:msgs-mi-low}}
\end{align*}
Putting things together, we have, 
\begin{align*}
	&	\Exp_{\rwstart \sim \cDrealtil} \tvd{\cDrealtil(\rw^{x} \mid \rwstart )}{\cH_1(\rw^{x} \mid \rwstart )} \\
	&\hspace{30pt}\leq  \frac1{\sqrt{2}} \cdot \sum_{j \in [2n_{r-1}]}  \sqrt{ \Exp_{\rw^{<j}\sim \cDrealtil} \kl{\cDrealtil(\rw^j \mid \rw^{<j})}{\cH_1(\rw^j \mid \rw^{<j} )}} \tag{from our earlier bound on the total variation distance in the proof}\\
	&\hspace{30pt}\leq \frac1{\sqrt{2}} \cdot 2n_{r-1} \cdot \sqrt{s \cdot 2n_{r-1} \cdot 1/(\gamma_r + 1)} = 2 \cdot (n_{r-1})^{3/2} \cdot \sqrt{\frac{s}{\gamma_r + 1}}. \qedhere
\end{align*}
\end{proof}

We prove \Cref{lem:hybrid-1-inner-msgs}, with another application of \Cref{fact:tvd-chain-rule}.
\begin{proof}[Proof of \Cref{lem:hybrid-1-inner-msgs}]
	\newcommand{\rwall}{\ensuremath{\rv{w}^{\textnormal{all}}}}
		Let $u_1, u_2, \ldots, u_{3n_{r-1}}$ be the inner vertices with the lexicographic ordering. We use $\rw^{u_{<\ell}}$ to denote the joint random variable $\rw^{u_1}, \rw^{u_2}, \ldots, \rw^{u_{\ell-1}}$ for $\ell \in [3n_{r-1}]$. We use $\rwall$ to denote all $\rw^{u_{\ell}}$ for $\ell \in [3n_{r-1}]$. Using \Cref{fact:tvd-chain-rule}, we get, 
	\begin{align*}
		&\tvd{\cDrealtil}{\cH_1} \\
		&\hspace{30pt}\leq \tvd{\cDrealtil(\rwstart)}{\cH_1(\rwstart)} + \sum_{\ell \in [3n_{r-1}]} \Exp_{\rw^{u_{< \ell}} \sim \cDrealtil}\tvd{\cDrealtil(\rw^{u_{\ell}} \mid \rw^{u_{<\ell}})}{\cH_1(\rw^{u_{\ell}} \mid \rw^{u_{< \ell}})} \\
		& \hspace{60pt} + 	\Exp_{\rwstart, \rwall \sim \cDrealtil}	\tvd{\cDrealtil(\rwend \mid \rwstart,\rwall)}{\cH_1(\rwend \mid  \rwstart, \rwall)} \\
		&\hspace{30pt}= \eqcolor{0} + \sum_{\ell \in [3n_{r-1}]} \Exp_{\rw^{u_{< \ell}} \sim \cDrealtil}\tvd{\cDrealtil(\rw^{u_{\ell}} \mid \rw^{u_{<\ell}})}{\cH_1(\rw^{u_{\ell}} \mid \rw^{u_{< \ell}})}  \\
		& \hspace{60pt} + 	\Exp_{\rwstart, \rwall \sim \cDrealtil}	\tvd{\cDrealtil(\rwend \mid \rwstart,\rwall)}{\cH_1(\rwend \mid  \rwstart, \rwall)}  \tag{by \Cref{obs:inter-hybrid-list}-\ref{item:hybrid1-1}, we know $\cDrealtil(\rwstart)$ and $\cH_1(\rwstart)$ are the same distribution}  \\
		&\hspace{30pt}= \sum_{\ell \in [3n_{r-1}]} \Exp_{\rw^{u_{< \ell}} \sim \cDrealtil}\tvd{\cDrealtil(\rw^{u_{\ell}} \mid \rw^{u_{<\ell}})}{\cH_1(\rw^{u_{\ell}} \mid \rw^{u_{< \ell}})} + \eqcolor{0} \tag{by \Cref{obs:inter-hybrid-list}-\ref{item:hybrid1-3}, distribution of $\rwend$ is the same} \\
		&\hspace{30pt}= \sum_{\ell \in [3n_{r-1}]} \Exp_{\rwstart \sim {\cDrealtil}} \tvd{\cDrealtil(\rw^{x} \mid \eqcolor{\rwstart} )}{\cH_1(\rw^{x} \mid \eqcolor{\rwstart} )} \tag{by \Cref{obs:inter-hybrid-list}-\ref{item:hybrid1-2}} \\
		&\hspace{30pt}\leq (3n_{r-1})  \cdot 2 \cdot (n_{r-1})^{3/2} \cdot \sqrt{\frac{s}{\gamma_r + 1}} \tag{by \Cref{clm:inter-hybrid-1}} \\
		&\hspace{30pt} = 6 \cdot (n_{r-1})^{5/2} \cdot \sqrt{\frac{s}{\gamma_r + 1}}. \qedhere
	\end{align*}
\end{proof}

\subsection{Inner Inputs and Public Messages have Low Correlation}\label{subsec:public-msgs}

In this subsection, we prove \Cref{lem:cH1-cH2-public-msgs} that bounds the distance between distributions $\cH_1$ and $\cH_2$, from \Cref{box:hybrid1-inner-msgs} and \Cref{box:hybrid2-public-msgs}, respectively. 
We recall the definition of these distributions: 

\vspace{3mm}
\hspace{-1.5em}
\begin{minipage}{0.48\textwidth}
	\small
\textbf{Distribution $\cH_1$:}
\begin{align*}
	&\rG_{r-1} \times (\rids, \rcJall, \rcKall, \rLall)\\
	&~\eqcolor{\times (\bigtimes_{x \in G_{r-1}} \rNpub{x}, \rMpub{x} \mid \rids, \raux, \rNinner{x}{})}  \\
	&\times (\bigtimes_{x \in G_{r-1}}\bigtimes_{\substack{y \in G_{r-1} \\y \notin X}}  (\rMinner{x}(y) \mid  \rids, \raux, \rNpub{x}, \rMpub{x}, \rNinner{x}{})   \\
	&\times (\bigtimes_{x \in G_{r-1}}  \rNrest{x} \mid \rMinner{x}, \rNinner{x}{}, \rids, \raux, \rNpub{x}, \rMpub{x}). 
\end{align*}
\end{minipage}
\vrule width 1pt 
\hspace{0.5em} 
\begin{minipage}{0.48\textwidth}
	\small
	\textbf{Distribution $\cH_2$:}
\begin{align*}
	&\rG_{r-1} \times (\rids, \rcJall, \rcKall, \rLall) \\
	&~\eqcolor{\times (\bigtimes_{x \in G_{r-1}} \rNpub{x}, \rMpub{x} \mid \rids, \raux)} \\
	&\times (\bigtimes_{x \in G_{r-1}}\bigtimes_{\substack{y \in G_{r-1} \\y \notin X}} (\rMinner{x}(y) \mid  \rids, \raux, \rNpub{x}, \rMpub{x}, \rNinner{x}{})    \\
	&\times (\bigtimes_{x \in G_{r-1}}  \rNrest{x} \mid \rMinner{x}, \rNinner{x}{}, \rids, \raux, \rNpub{x}, \rMpub{x}). 
\end{align*}
\end{minipage}
\vspace{3mm}

 The only difference is that the conditioning on $\rNinner{x}{}$ when sampling some messages publicly in the second line in distribution $\cH_1$  is removed in $\cH_2$. We show that the correlation between messages sampled publicly and $\rNinner{x}{}$ is small in this subsection. 

\begin{observation}\label{obs:re-inter1}
For every inner vertex $x$, 
\[
\rNpub{x} \perp \rNinner{x}{} \mid \rids, \raux.
\]
\end{observation}
\begin{proof}
We know that $\rNpub{x}$ is comprised of $\Nbhood{x}{Y}{}[\rLupdown{x \to Y}{t, i}]$ for every other layer $Y$ with $x \notin Y$, type $t \in [r] \cup \{0\}$ and value $i \in [n_{r-1}]$. All these values are deterministically fixed to be $t$ for each $t \in [r] \cup \{0\}$ by definition of $\rLupdown{x \to Y}{t, i}$. Thus, $\rNpub{x}$ is independent of $\rNinner{x}{}$ conditioned on $\rids, \raux$.
\end{proof}

\begin{claim}\label{clm:mi-low-public-msgs}
	For every inner vertex $x$, we have, \[
		\mi{\rMpub{x}, \rNpub{x}}{\rNinner{x}{} \mid \rids, \raux} \leq \frac1{(\alpha_r + 1)} \cdot s \cdot 2\gamma_r  \cdot (r+1) \cdot n_{r-1}.
	\]
\end{claim}
\begin{proof}
	First, we have, 
	\begin{align*}
		&	\mi{\rMpub{x}, \rNpub{x}}{\rNinner{x}{} \mid \rids, \raux} \\
		&\hspace{30pt}= \mi{\rNpub{x}}{\rNinner{x}{} \mid \rids, \raux} + \mi{\rMpub{x}}{\rNinner{x}{} \mid \rids, \raux, \rNpub{x}} \tag{by the chain rule of mutual information \itfacts{chain-rule}} \\
		&\hspace{30pt} = 0 + \mi{\rMpub{x}}{\rNinner{x}{} \mid \rids, \raux, \rNpub{x}}  \tag{by \Cref{obs:re-inter1} and by \itfacts{info-zero}} \\
		&\hspace{30pt}= 	\mi{\rMpub{x}}{\rNinner{x}{} \mid \eqcolor{\ridlayer{X}, \ridlayer{Y}, \ridlayer{Z}, \rcJall, \rcKall, \rLall}, \rNpub{x}}. \tag{expanding $\raux$ and $\rids$} \\
		&\hspace{30pt}= 	\mi{\rMpub{x}}{\eqcolor{\rNbtogether{x}{}[\ridlayer{Y} \cup \ridlayer{Z}]} \mid \ridlayer{X}, \ridlayer{Y}, \ridlayer{Z}, \rcJall, \rcKall, \rLall, \rNpub{x}}. \tag{as $\rNinner{x}{}$ is fixed by $\rNbtogether{x}{}[\ridlayer{Y} \cup \ridlayer{Z}]$ when the identities are unique}
	\end{align*}
We use $\rWrest $ to denote the random variables $ \ridlayer{X}$, $\rcKall, \rLall, \rNpub{x}$ and all random variables in $\rcJall$ barring $\rcJup{x}$. These random variables are bundled together as they are not relevant to the rest of the argument. 
\begin{align*}
	&	\mi{\rMpub{x}}{\rNbtogether{x}{}[\ridlayer{Y} \cup \ridlayer{Z}] \mid \ridlayer{X}, \ridlayer{Y}, \ridlayer{Z}, \rcJall, \rcKall, \rLall, \rNpub{x}} \\
	&\hspace{30pt}= 	\mi{\rMpub{x}}{\rNbtogether{x}{}[\ridlayer{Y} \cup \ridlayer{Z}] \mid \ridlayer{Y}, \ridlayer{Z}, \rcJup{x}, \eqcolor{\rWrest}} \\
	&\hspace{30pt}= 	\mi{\rMpub{x}}{\rNbtogether{x}{}[\ridlayer{Y} \cup \ridlayer{Z}] \mid \eqcolor{\ridlayer{Y} \cup \ridlayer{Z}}, \rcJup{x}, \rWrest},
	\intertext{
where for the last equality, we have used that $\ridlayer{Y} \cup \ridlayer{Z}$ is fixed by $\ridlayer{Y}, \ridlayer{Z} $ and vice-versa. This is because $\ridlayer{Y} = \{\ridlayer{Y} \cup \ridlayer{Z}\} \cap Y$ and $\ridlayer{Z} = \{\ridlayer{Y} \cup \ridlayer{Z}\} \cap Z$ as they are disjoint. We continue, }
		&\hspace{30pt}= \mi{\rMpub{x}}{\rNbtogether{x}{}[\ridlayer{Y} \cup \ridlayer{Z}] \mid \ridlayer{Y} \cup \ridlayer{Z}, \eqcolor{\rcJup{x} \cup \{\ridlayer{Y} \cup \ridlayer{Z}\}}, \rWrest}, 
\end{align*}
where, the last step is true because $\rcJup{x} \cup \{\ridlayer{Y} \cup \ridlayer{Z}\}, \ridlayer{Y} \cup \ridlayer{Z}$ is fixed by $\rcJup{x}, \ridlayer{Y} \cup \ridlayer{Z}$ and vice-versa (the elements in $\rcJup{x}$ are disjoint from all the elements in $\ridlayer{Y} \cup \ridlayer{Z}$ by definition). 

We know that $\rcJup{x} \cup \{\ridlayer{Y} \cup \ridlayer{Z}\}$ is a collection of size $\alpha_r + 1$, where each set has $2n_{r-1}$ elements ($n_{r-1}$ ones from each of  $Y$ and $Z$). We condition on this collection being sets $\cJ = (J_1, J_2, \ldots, J_{\alpha_r + 1})$ in the next step and have, 
\begin{align*}
&\mi{\rMpub{x}}{\rNbtogether{x}{}[\ridlayer{Y} \cup \ridlayer{Z}] \mid \ridlayer{Y} \cup \ridlayer{Z}, \rcJup{x} \cup \{\ridlayer{Y} \cup \ridlayer{Z}\}, \rWrest} \\
&= \eqcolor{\Exp_{\rcJup{x} \cup \{\ridlayer{Y} \cup \ridlayer{Z}\} = \cJ}}\mi{\rMpub{x}}{\rNbtogether{x}{}[\ridlayer{Y} \cup \ridlayer{Z}] \mid \ridlayer{Y} \cup \ridlayer{Z}, \eqcolor{\rcJup{x} \cup \{\ridlayer{Y} \cup \ridlayer{Z}\} = \cJ}, \rWrest}. \tag{by the definition of conditional mutual information}
\end{align*}
We prove the statement separately for each collection $\cJ$. 

Now, we argue that conditioned on $\rcJup{x} \cup \{\ridlayer{Y} \cup \ridlayer{Z}\} = \cJ$, the value of $\ridlayer{Y} \cup \ridlayer{Z}$ is uniform over all the sets in collection $\cJ$. Random variable $\rcJup{x}$ is made of $\alpha_r$ disjoint collections, each with $n_{r-1}$ elements from each of $Y$ and $Z$. Moreover, both $\rcJup{x}$ and $\ridlayer{Y} \cup \ridlayer{Z}$ are chosen uniformly at random such that the sets in $\rcJup{x}$ and $\ridlayer{Y} \cup \ridlayer{Z}$ are disjoint. Hence, given that  $\rcJup{x} \cup \{\ridlayer{Y} \cup \ridlayer{Z}\} = \cJ$, $\ridlayer{Y} \cup \ridlayer{Z}$ can be any set in this collection chosen uniformly. We continue as, 
\begin{align*}
&\mi{\rMpub{x}}{\rNbtogether{x}{}[\ridlayer{Y} \cup \ridlayer{Z}] \mid \ridlayer{Y} \cup \ridlayer{Z}, \rcJup{x} \cup \{\ridlayer{Y} \cup \ridlayer{Z}\} = \cJ, \rWrest} \\
&= \frac1{\alpha_r + 1} \cdot \sum_{i \in [\alpha_r + 1]} \mi{\rMpub{x}}{\rNbtogether{x}{}[J_i] \mid \rWrest,  \ridlayer{Y} \cup \ridlayer{Z} = J_i,  \rcJup{x} \cup \{\ridlayer{Y} \cup \ridlayer{Z}\} = \cJ}.
\end{align*}

We will prove that the joint distribution of all the random variables in the mutual information term is independent of the event $\{\ridlayer{Y} \cup \ridlayer{Z}\} = J_i$ conditioned on the event $\rcJup{x} \cup \{\ridlayer{Y} \cup \ridlayer{Z}\} = \cJ$. We will go over the terms one by one, and for each term, we also condition on the preceding terms.  
\begin{itemize}
	\item $\rNbtogether{x}{}[J_i]$ for $i \in [\alpha_r + 1]$: all these subsets are sampled from $\cDmarg$ independently of each other. The value of set $\ridlayer{Y} \cup \ridlayer{Z} $ inside collection $\cJ$ has no correlation with the joint distribution of these inputs.
	\item $\rNbhood{x}{Y}{}, \rNbhood{x}{Z}{}$: this is the input of vertex $x$. These variables can only be correlated to the event $\ridlayer{Y} \cup \ridlayer{Z}  = J_i$ through random variables $\rNbtogether{x}{}[J_i]$ for $i \in [\alpha_r + 1]$, which we have just argued is independent of the event. 
	\item $\rWrest$: this has random variable $\rNpub{x}$, which are channel types of $x$ to vertices disjoint from collection $\cJ$ (and independent of what happens inside $\cJ$), and the other random variables $\ridlayer{X}, \rcKall$, $\rLall$, all $\rcJup{w}$ for $w \neq x$.  These other variables are disjoint from $\rcJup{x} \cup \{\ridlayer{Y} \cup \ridlayer{Z}\}$, and are independent of the value of set $\ridlayer{Y} \cup \ridlayer{Z}$ inside collection $\cJ$.
	\item $\rMpub{x}$: these are the messages that $x$ sends to specific indices in $\rLupdown{x \to Y}{t, j}$ for each other layer $Y$, type $t \in [r] \cup \{0\}$ and $j \in [n_{r-1}]$. We have argued that the input to vertex $x$ is independent of the identity of $\ridlayer{Y} \cup \ridlayer{Z}$ inside collection $\cJ$. As the protocol $\prot_r$ is deterministic, $\rMpub{x}$ is fixed by $\rNbhood{x}{Y}{}, \rNbhood{x}{Z}{}, \rid{x}$, and is independent of the event $\ridlayer{Y} \cup \ridlayer{Z} = J_i$. 
\end{itemize}
Thus, the joint distribution of all the random variables is independent of the value of $\{\ridlayer{Y} \cup \ridlayer{Z}\}$ inside collection $\cJ$. The mutual information term then becomes, 
\begin{align*}
	& \frac1{\alpha_r + 1} \cdot \sum_{i \in [\alpha_r + 1]} \mi{\rMpub{x}}{\rNbtogether{x}{}[J_i] \mid \rWrest,  \eqcolor{\ridlayer{Y} \cup \ridlayer{Z} = J_i},  \rcJup{x} \cup \{\ridlayer{Y} \cup \ridlayer{Z}\} = \cJ} \\
	 &\hspace{30pt}= 	 \frac1{\alpha_r + 1} \cdot \sum_{i \in [\alpha_r + 1]} \mi{\rMpub{x}}{\rNbtogether{x}{}[J_i] \mid \rWrest,   \rcJup{x} \cup \{\ridlayer{Y} \cup \ridlayer{Z}\} = \cJ} \\
	 &\hspace{30pt}\leq  \frac1{\alpha_r + 1} \cdot \sum_{i \in [\alpha_r + 1]} \mi{\rMpub{x}}{\rNbtogether{x}{}[J_i] \mid \rWrest,  \eqcolor{\rNbtogether{x}{}[J_1], \ldots, \rNbtogether{x}{}[J_{i-1}]}, \rcJup{x} \cup \{\ridlayer{Y} \cup \ridlayer{Z}\} = \cJ},
\end{align*}
where for the last step we have used that $\rNbtogether{x}{}[J_i] \perp \rNbtogether{x}{}[J_k] \mid \rWrest$, $\rcJup{x} \cup \{\ridlayer{Y} \cup \ridlayer{Z}\} = \cJ$ for any $i, k \in [\alpha_r+1]$ with $i \neq k$, and thus we can apply \Cref{prop:info-increase}. The independence of $\rNbtogether{x}{}[J_i]$ and $\rNbtogether{x}{}[J_k]$ follows because, conditioned on $\rcJup{x} \cup \{\ridlayer{Y} \cup \ridlayer{Z}\} = \cJ$, both these variables are sampled from $\cDmarg$ independently of each other. The value of $\rWrest$ does not affect this independence either, as $\rWrest$ is made of $\rcKall, \rLall, \ridlayer{X}$ , $\rcJup{w}$ for $w \neq x$, all of which are disjoint from $\cJ$ and $\rNpub{x}$ consisting of channel types to $x$ disjoint from collection $\cJ$. 
We proceed with the proof as follows.
\begin{align*}
&	 \frac1{\alpha_r + 1} \cdot \eqcolor{\sum_{i \in [\alpha_r + 1]}} \mi{\rMpub{x}}{\rNbtogether{x}{}[J_i] \mid \rWrest, \eqcolor{\rNbtogether{x}{}[J_1], \ldots, \rNbtogether{x}{}[J_{i-1}]}, \rcJup{x} \cup \{\ridlayer{Y} \cup \ridlayer{Z}\} = \cJ} \\
	 &= \frac{1}{\alpha_r + 1} \cdot \mi{\rMpub{x}}{\eqcolor{\rNbtogether{x}{}[J_1], \rNbtogether{x}{}[J_2], \ldots, \rNbtogether{x}{}[J_{\alpha_r + 1}]} \mid \rWrest, \rcJup{x} \cup \{\ridlayer{Y} \cup \ridlayer{Z}\} = \cJ} \tag{by the chain rule of mutual information in \itfacts{chain-rule}}\\
	 &\leq \frac1{\alpha_r + 1} \cdot \en{\rMpub{x} \mid \rWrest, \rcJup{x} \cup \{\ridlayer{Y} \cup \ridlayer{Z}\} = \cJ} \tag{by \itfacts{info-entropy}} \\
	 &\leq \frac1{\alpha_r + 1} \cdot s \cdot 2\gamma_r  \cdot (r+1) \cdot n_{r-1}. \tag{by~\itfacts{uniform}, \Cref{obs:public-msgs-number}, and as the length of each message is bounded by $s$}
\end{align*}
To finish the proof, we get, 
\begin{align*}
\mi{\rMpub{x}, \rNpub{x}}{\rNinner{x}{} \mid \rids, \raux} &\leq  \Exp_{\rcJup{x} \cup \{\ridlayer{Y} \cup \ridlayer{Z}\} = \cJ} \bracket{\frac1{\alpha_r + 1} \cdot s \cdot 2\gamma_r  \cdot (r+1) \cdot n_{r-1}} \\
&= \frac1{\alpha_r + 1} \cdot s \cdot 2\gamma_r   \cdot (r+1) \cdot n_{r-1}. \qedhere
\end{align*}
\end{proof}

We are ready to prove \Cref{lem:cH1-cH2-public-msgs} using the weak chain rule of total variation distance. 
We begin by defining the random variables on which we apply weak chain rule. These random variables are local to this subsection and may be used for different purposes later. (See \Cref{sec:list-rv} for a list of global random variables.)
	\begin{itemize}
		\item Variable $\rwstart$: this is the joint random variable $\rG_{r-1}, \rids$, $\rcJall, \rcKall, \rLall$. 
		\item Variables $\rw^{x}$ for each inner vertices $x$: this is the joint random variable $\rNpub{x}, \rMpub{x}$.
		\item Variables $\rwend$: joint random variable $\rMinner{x}, \rNrest{x}$ for each inner vertex $x$. 
	\end{itemize}

We use the following simple observations. They are fairly direct, and are not justified further.
\begin{observation}\label{obs:inter-H1H2-list}
About random variables $\rwstart, \rw^x$ for each inner vertex $x$ and $\rwend$, in distributions $\cH_1$ and $\cH_2$, we have, 
\begin{enumerate}[label=$(\roman*)$]
\item \label{item:wstart-same-H1H2} In $\cH_1$ and $\cH_2$, the distribution of random variable $\rwstart$ is the same. 
\item \label{item:wx-indep-H1H2} In $\cH_1$ and $\cH_2$, the random variables $\rw^w$ and $\rw^{w'}$ are independent of each other conditioned on $\rwstart$ for inner vertices $w \neq w'$ (where $w, w'$ could be in the same layer or in different layers). This is because all the random variables $\rids, \raux$ and $\rNinner{x}{}$ for each inner vertex $x$ are fixed when conditioned on $\rwstart$. 
\item \label{item:wx-indep-H1} In $\cH_1$, random variable $\rw^x$ is independent of $\rwstart$ when conditioned on $\rNinner{x}{} , \rids$ and $\raux$. 
\item \label{item:wx-indep-H2} In $\cH_2$, random variable $\rw^x$ is independent of $\rwstart$ when conditioned on $\rids$ and $\raux$.
\item \label{item:wend-indep-H1H2} Conditioned on any choice of $\rwstart$ and $\rw^x$ for each inner vertex $x$, the distribution of random variable $\rwend$ is the same in $\cH_1$ and $\cH_2$. 
\end{enumerate}
\end{observation}

We prove the following intermediate claim that states that the distribution of $\rw^x$ is close in total variation distance in $\cH_1$ and $\cH_2$.

\begin{claim}\label{clm:inter-H1H2-close}
	For every inner vertex $x$,
	\[
		\Exp_{\rwstart \sim \cH_1} \tvd{\cH_1(\rw^x \mid \rwstart)}{\cH_2(\rw^x \mid \rwstart)} \leq  \sqrt{\frac{s \cdot \gamma_r  \cdot (r+1) \cdot n_{r-1}}{\alpha_r + 1}}.
	\]
\end{claim}
\begin{proof}
	We know from~\Cref{obs:inter-H1H2-list}-\ref{item:wx-indep-H1}  that, 
	\begin{equation}\label{eq:H1-H2-1}
	\cH_1(\rw^x \mid \rwstart) = \rNpub{x}, \rMpub{x} \mid \rids, \raux, \rNinner{x}{}. 
	\end{equation}
	Similarly in distribution $\cH_2$ from~\Cref{obs:inter-H1H2-list}-\ref{item:wx-indep-H2}, we get, 
	\begin{equation}\label{eq:H1-H2-2}
	\cH_2(\rw^x \mid \rwstart)  = \rNpub{x}, \rMpub{x} \mid \rids, \raux.
	\end{equation}
	We can write the total variation distance as, 
	\begin{align*}
	&\Exp_{\rwstart \sim \cH_1} \tvd{\cH_1(\rw^x \mid \rwstart)}{\cH_2(\rw^x \mid \rwstart)} \\
		&\hspace{30pt}\leq \frac1{\sqrt{2}} \cdot \Exp_{\rwstart \sim \cH_1} \sqrt{\kl{\cH_1(\rw^x \mid \rwstart)}{\cH_2(\rw^x \mid \rwstart)}} \tag{by in Pinsker's inequality \Cref{fact:pinskers}}\\
		&\hspace{30pt}\leq \frac1{\sqrt{2}} \cdot \sqrt{ \Exp_{\rwstart \sim \cH_1} \kl{\cH_1(\rw^x \mid \rwstart)}{\cH_2(\rw^x \mid \rwstart)}} \tag{by Jensen's inequality and concavity of square root} \\
	&\hspace{30pt} =  \frac1{\sqrt{2}} \cdot \sqrt{ \Exp_{\rids, \raux, \rNinner{x}{}} \kl{(\rNpub{x}, \rMpub{x} \mid \rids, \raux, \rNinner{x}{})}{(\rNpub{x}, \rMpub{x} \mid \rids, \raux)}} \tag{by \Cref{eq:H1-H2-1} and \Cref{eq:H1-H2-2}} \\
	&\hspace{30pt}=  \frac1{\sqrt{2}} \cdot \sqrt{ \mi{\rNpub{x}, \rMpub{x}}{\rNinner{x}{} \mid \rids, \raux}} \tag{by relation between KL-Divergence and mutual information \Cref{fact:kl-info}} \\
	&\hspace{30pt}\leq \sqrt{\frac1{\alpha_r + 1} \cdot s \cdot \gamma_r  \cdot (r+1) \cdot n_{r-1}} \tag{by \Cref{clm:mi-low-public-msgs}}. \qedhere
	\end{align*}
\end{proof}
\newcommand{\rwall}{\rw^{\textnormal{all}}}
	
We conclude this subsection by proving \Cref{lem:cH1-cH2-public-msgs}.
\begin{proof}[Proof of \Cref{lem:cH1-cH2-public-msgs}]
	Firstly, by \Cref{obs:inter-H1H2-list}-\ref{item:wstart-same-H1H2}, we know that, 
	\begin{equation}\label{eq:inter-wstart-H1H2}
		\tvd{\cH_1(\rwstart)}{\cH_2(\rwstart)} = 0.
	\end{equation}
	We follow the lexicographic ordering on the vertices to order the random variables $\rw^x$ for inner vertices $x \in G_{r-1}$. 
	Let $u_1, u_2, \ldots, u_{3n_{r-1}}$ be the inner vertices with the ordering. We use $\rw^{u_{<\ell}}$ to denote the joint random variable $\rw^{u_1}, \rw^{u_2}, \ldots, \rw^{u_{\ell-1}}$ for $\ell \in [3n_{r-1}]$. We use $\rwall$ to denote all $\rw^{u_{\ell}}$ for $\ell \in [3n_{r-1}]$. 
	
	By \Cref{obs:inter-H1H2-list}-\ref{item:wx-indep-H1H2}, we have, for any $\ell \in [3n_{r-1}]$, 
	\begin{align}\label{eq:inter-wx-H1H2}
		\Exp_{\rw^{u_{< \ell}} \sim \cH_1}\tvd{\cH_1(\rw^{u_{\ell}} \mid \rw^{u_{< \ell}})}{\cH_2(\rw^{u_{\ell} \mid \rw^{u_{< \ell}}})} &= \Exp_{\rwstart \sim \cH_1} \tvd{\cH_1(\rw^{u_{\ell}} \mid \rwstart)}{\cH_2(\rw^{u_{\ell}} \mid \rwstart)} \nonumber\\
		&\leq \sqrt{\frac{s \cdot \gamma_r  \cdot (r+1) \cdot n_{r-1}}{\alpha_r + 1}},
	\end{align}
	where for the inequality we used \Cref{clm:inter-H1H2-close}. Lastly, from \Cref{obs:inter-H1H2-list}-\ref{item:wend-indep-H1H2}, we have, 
	\begin{equation}\label{eq:inter-H1H2-end}
	\Exp_{\rwstart, \rwall \sim \cH_1}	\tvd{\cH_1(\rwend \mid \rwstart,\rwall)}{\cH_2(\rwend \mid  \rwstart, \rwall)}  = 0.
	\end{equation}
	We can complete the proof easily now. 
	\begin{align*}
		&\tvd{\cH_1}{\cH_2} \\
		&\hspace{30pt}\leq \tvd{\cH_1(\rwstart)}{\cH_2(\rwstart)} + \sum_{\ell \in [3n_{r-1}]} \Exp_{\rw^{u_{< \ell}} \sim \cH_1}\tvd{\cH_1(\rw^{u_{\ell}} \mid \rw^{u_{<\ell}})}{\cH_2(\rw^{u_{\ell}} \mid \rw^{u_{< \ell}})} \\
		& \hspace{60pt} + 	\Exp_{\rwstart, \rwall \sim \cH_1}	\tvd{\cH_1(\rwend \mid \rwstart,\rwall)}{\cH_2(\rwend \mid  \rwstart, \rwall)}  \tag{by \Cref{fact:tvd-chain-rule}}\\
		&\hspace{30pt}= \sum_{\ell \in [3n_{r-1}]} \Exp_{\rw^{u_{< \ell}} \sim \cH_1}\tvd{\cH_1(\rw^{u_{\ell}} \mid \rw^{u_{<\ell}})}{\cH_2(\rw^{u_{\ell}} \mid \rw^{u_{< \ell}})} \tag{by \Cref{eq:inter-wstart-H1H2} and \Cref{eq:inter-H1H2-end} first and last terms are zero} \\
		&\hspace{30pt}\leq 3n_{r-1} \cdot  \sqrt{\frac{s \cdot \gamma_r  \cdot (r+1) \cdot n_{r-1}}{\alpha_r + 1}}. \qedhere \tag{by \Cref{eq:inter-wx-H1H2}}
	\end{align*}
\end{proof}

\subsection{Inner Messages and Inner Inputs have Low Correlation}\label{subsec:inner-msgs}

In this subsection, we prove \Cref{lem:cH2-cDfake-inner-msgs}, which bounds the total variation distance between $\cH_2$ from \Cref{box:hybrid2-public-msgs} and $\cDfake$ from \Cref{box:cDfake}.

\vspace{3mm}
\hspace{-1.5em}
\begin{minipage}{0.48\textwidth}

	\small
\textbf{Distribution $\cH_2$:}
\begin{align*}
	&\rG_{r-1} \times (\rids, \rcJall, \rcKall, \rLall) \\
	&\times (\bigtimes_{x \in G_{r-1}} \rNpub{x}, \rMpub{x} \mid \rids, \raux) \\
	&~\eqcolor{\times (\bigtimes_{x \in G_{r-1}}\bigtimes_{\substack{y \in G_{r-1} \\y \notin X}} (\rMinner{x}(y) \mid  \substack{\rids, \raux, \rNpub{x},  \\ \rMpub{x}, \rNinner{x}{}})}    \\
	&\times (\bigtimes_{x \in G_{r-1}}  \rNrest{x} \mid \rMinner{x}, \rNinner{x}{}, \rids, \raux, \rNpub{x}, \rMpub{x}). 
\end{align*}
\end{minipage}
\vrule width 1pt 
\hspace{0.5em} 
\begin{minipage}{0.48\textwidth}
	\small
\textbf{Distribution $\cDfake$:}
\begin{align*}
	&\rG_{r-1} \times (\rids, \rcJall, \rcKall, \rLall)\\
	&\times (\bigtimes_{x \in G_{r-1}} \rNpub{x}, \rMpub{x} \mid \rids, \raux)\\
	&~\eqcolor{\times (\bigtimes_{x \in G_{r-1}} \bigtimes_{\substack{y \in G_{r-1} \\y \notin X}} (\rMinner{x}(y) \mid \substack{\rids,  
		\raux, \rNpub{x}, \\ \rMpub{x},  \rtypeval{x, y}})}    \\
	&\times (\bigtimes_{x \in G_{r-1}}  \rNrest{x} \mid \rMinner{x}, \rNinner{x}{}, \rids, \raux, \rNpub{x}, \rMpub{x}). 
\end{align*}
\end{minipage}
\vspace{3mm}

The final step to get to distribution $\cDfake$ is that in sampling $\rMinner{x}(y)$ for each inner vertex $x$ and inner vertex $y \notin X$, in $\cH_2$, we condition on $\rNinner{x}{}$, whereas in $\cDfake$, there is only conditioning on $\rtypeval{x, y}$. We will show that the conditioning on the other random variables in $\rNinner{x}{}$ can be removed without much loss in the total variation distance. 

\begin{claim}\label{clm:mi-low-inner-msgs}
	For every inner vertex $x$ and $y_i \in G_{r-1}$, we have, 
	\[
		\mi{\rMinner{x}(y_i)}{\rNinner{x}{} \mid \rids, \raux, \rNpub{x}, \rMpub{x}, \rtypeval{x, y_i}} \leq \frac{s}{\beta_r + 1}.
	\]
\end{claim}
\begin{proof}
	We start by using the definition of conditional mutual information to write,
	\begin{align*}
	&	\mi{\rMinner{x}(y_i)}{\rNinner{x}{} \mid \rids, \raux, \rNpub{x}, \rMpub{x}, \rtypeval{x, y_i}} \\
	&\hspace{30pt}= \Exp_{\rtypeval{x, y_i} = t} 	\mi{\rMinner{x}(y_i)}{\rNinner{x}{} \mid \rids, \raux, \rNpub{x}, \rMpub{x}, \rtypeval{x, y_i} = t}. 
	\end{align*}
	Let $\Et$ denote the event that $\rtypeval{x, y_i} = t$. We prove the statement separately for each $t \in [r] \cup \{0\}$. 
	\begin{align*}
&	\mi{\rMinner{x}(y_i)}{\rNinner{x}{} \mid \rids, \raux, \rNpub{x}, \rMpub{x}, \Et} \\
	&\hspace{30pt}= \mi{\rMinner{x}(y_i)}{\eqcolor{\rNbtogether{x}{}[\ridlayer{Y} \cup \ridlayer{Z}]} \mid \eqcolor{\ridlayer{X}, \ridlayer{Y}, \ridlayer{Z}}, \raux, \rNpub{x}, \rMpub{x}, \Et}. \tag{expanding $\rids$, and as $\rNinner{x}{}$ is fixed by $\rNbtogether{x}{}[\ridlayer{Y} \cup \ridlayer{Z}] $}
	\end{align*}
	Let $\rid{-y_i}$ denote the random variables $\ridlayer{Z}$ and $\rid{y'}$ for $y' \neq y_i \in G_{r-1}$, 
	and $\rWrest$ denote the random variables $\ridlayer{X}$, $\rcJall$, $\rLall$, $\rMpub{x}, \rNpub{x}$ and all random variables in $\rcKall$ barring $\rcKupdown{x \to Y}{t, i}$.
	\begin{align*}
&\mi{\rMinner{x}(y_i)}{\rNbtogether{x}{}[\ridlayer{Y} \cup \ridlayer{Z}] \mid \ridlayer{X}, \ridlayer{Y}, \ridlayer{Z}, \raux, \rNpub{x}, \rMpub{x}, \Et} \\
&\hspace{30pt}= \mi{\rMinner{x}(y_i)}{\rNbtogether{x}{}[\ridlayer{Y} \cup \ridlayer{Z}] \mid \ridlayer{Y}, \ridlayer{Z}, \eqcolor{\rcKupdown{x \to Y}{t, i}, \rWrest}, \Et} \\
&\hspace{30pt}=  \mi{\rMinner{x}(y_i)}{\eqcolor{\rNbtogether{x}{}[\rid{y_i} \cup \rid{-y_i}]} \mid \eqcolor{\rid{-y_i}, \rid{y_i}}, \rcKupdown{x \to Y}{t, i}, \rWrest, \Et} \tag{by definition of $\rWrest, \rid{-y_i}$} \\
&\hspace{30pt}=  \mi{\rMinner{x}(y_i)}{\eqcolor{\rNbtogether{x}{}[\rid{-y_i}]} \mid \rid{-y_i}, \rid{y_i}, \rcKupdown{x \to Y}{t, i}, \rWrest, \Et} \tag{as $\rNbtogether{x}{}[\rid{y_i}]$ is fixed to be $t$ conditioned on $\Et$} \\
&\hspace{30pt}= \mi{\rMinner{x}(y_i)}{\rNbtogether{x}{}[\rid{-y_i}] \mid \rid{-y_i}, \rid{y_i}, \eqcolor{\rcKupdown{x \to Y}{t, i} \cup \{\rid{-y_i}\}}, \rWrest, \Et} \tag{as $\rid{-y_i}, \rcKupdown{x \to Y}{t, i}$ is fixed by $\rid{-y_i}, \rcKupdown{x \to Y}{t, i} \cup \{\rid{-y_i}\}$ and vice-versa} \\
&\hspace{30pt}= \eqcolor{\Exp_{\cK \sim \rcKupdown{x \to Y}{t, i} \cup \{\rid{-y_i}\}}}   \mi{\rMinner{x}(y_i)}{\rNbtogether{x}{}[\rid{-y_i}] \mid \rid{-y_i}, \rid{y_i}, \rWrest, \eqcolor{\rcKupdown{x \to Y}{t, i} \cup \{\rid{-y_i}\} = \cK}, \Et}. \tag{by the definition of conditional mutual information}
	\end{align*}
	We prove the statement separately for each collection $\cK = \{K_1, K_2, \ldots, K_{\beta_r + 1}\}$. We know $\cK$ is a collection of $\beta_r + 1$ sets, each of which has $n_{r-1}-1$ elements from $Y \setminus \rid{y_i}$ and $n_r$ elements from $Z$. 
	We argue that $\rid{-y_i}$ is uniform over all the sets in collection $\cK$. This is because $\rcKupdown{x \to Y}{t, i} $ is a collection of $\beta_r$ elements disjoint from $\rid{-y_i}$ chosen uniformly at random. Variable $\rid{-y_i}$ is also chosen uniformly at random from sets with $n_{r-1}-1$ elements from $Y$ and $n_{r-1}$ elements from $Z$. 
	
	Let $\Econd$ be the event that $\rcKupdown{x \to Y}{t, i} \cup \{\rid{-y_i}\} = \cK$ and $\Et$ happen. We have, 
	\begin{align*}
	 &\mi{\rMinner{x}(y_i)}{\rNbtogether{x}{}[\rid{-y_i}] \mid \rid{-y_i}, \rid{y_i}, \rWrest, \eqcolor{\Econd}} \\
	 &= \eqcolor{\frac1{\beta_r + 1} \cdot \sum_{j \in [\beta_r + 1]}}  \mi{\rMinner{x}(y_i)}{\rNbtogether{x}{}[K_j] \mid \rid{y_i}, \rWrest, \eqcolor{\rid{-y_i} = K_j},  \Econd},
	\end{align*}
	by uniformity of $\rid{-y_i}$ argued above. 
	
	Conditioned on $\Econd$, the joint distribution of all the random variables in the mutual information term are independent of the event $ \rid{-y_i} = K_j$. Let us list these random variables, and argue in steps (in each step, we condition on everything in the previous steps also). 
	\begin{itemize}
	\item $\rNbtogether{x}{}[K_{\ell}]$ for all $\ell \in [\beta_r + 1]$: these random variables are all distributed so that $\rNbtogether{x}{}[K_{\ell} \cup \{\rid{y_i}\}] $ is sampled from $\cDmarg$, conditioned on $\Econd$. The value of $\rid{-y_i}$ among the $\beta_r + 1$ sets is irrelevant to the distribution. 
	\item $\rNbhood{x}{Y}{}, \rNbhood{x}{Z}{}$: this is the input of $x$, and it depends on the event $ \rid{-y_i} = K_j$ only through $\rNbtogether{x}{}[K_{\ell}]$ for $\ell \in [\beta_r + 1]$, which we have argued is independent of the value of $\rid{-y_i}$. 
	\item $\rmsg{x}{Y}{}, \rmsg{x}{Z}{}$: as the protocol is deterministic, and the input of $x$ is independent of the value of $\rid{-y_i}$, the messages are independent of the event also. 
	\item $\ridlayer{X}, \rcJall, \rLall$ and $\rcKall$ barring $\rcKupdown{x \to Y}{t, i}$: all these sets are disjoint from sets in collection $\cK$, and therefore are independent of what happens inside $\cK$. 
	\item $\rWrest$: this random variable is comprised of $\ridlayer{X}, \rcJall, \rLall, \rMpub{x}$, $\rNpub{x}$ and $\rcKall$ barring $\rcKupdown{x \to Y}{t, i}$, all of which we have argued are independent of event $\rid{-y_i} = K_j$. 
	\end{itemize}
The joint distribution of all these random variables is independent of the event $\rid{-y_i} = K_j$. 
	We can continue the proof as,
	\begin{align*}
	&	\frac1{\beta_r + 1} \cdot \sum_{j \in [\beta_r + 1]}  \mi{\rMinner{x}(y_i)}{\rNbtogether{x}{}[K_j] \mid  \rid{y_i}, \rWrest, \eqcolor{\rid{-y_i} = K_j},  \Econd} \\
	&\hspace{30pt}= \frac1{\beta_r + 1} \cdot \sum_{j \in [\beta_r + 1]}  \mi{\rMinner{x}(y_i)}{\rNbtogether{x}{}[K_j] \mid \rid{y_i}, \rWrest,  \Econd} \tag{by the independence of the joint distribution of remaining variables and event $\rid{-y_i} = K_j$ } \\
	&\hspace{30pt}\leq  \frac1{\beta_r + 1} \cdot \sum_{j \in [\beta_r + 1]}  \mi{\rMinner{x}(y_i)}{\rNbtogether{x}{}[K_j] \mid \eqcolor{\rNbtogether{x}{}[K_1], \rNbtogether{x}{}[K_2], \ldots, \rNbtogether{x}{}[K_{j-1}]}, \rid{y_i}, \rWrest,  \Econd}, 
	\intertext{
	where we have used that $\rNbtogether{x}{}[K_{j}] \perp \rNbtogether{x}{}[K_{\ell}]$ conditioned on $\rid{y_i}$, $\rWrest$ and $\Econd$ for $\ell \neq j$, so we can apply \Cref{prop:info-increase}. The independence holds because for $\ell \in [\beta_r + 1]$, $\rNbtogether{x}{}[K_{\ell}]$ is sampled so that $\rNbtogether{x}{}[K_{\ell} \cup \{\rid{y_i}\}]$ is distributed according $\cDmarg$ independently other $\ell' \in [\beta_r+1]$ conditioned on $\rid{y_i}, \rWrest$ and the event $\Econd$. We continue
	}
	&\hspace{30pt}= \frac1{\beta_r + 1} \cdot \mi{\rMinner{x}(y_i)}{\eqcolor{\rNbtogether{x}{}[K_1], \rNbtogether{x}{}[K_2], \ldots, \rNbtogether{x}{}[K_{\beta_r + 1}]} \mid \rid{y_i}, \rWrest, \Econd} \tag{by chain rule of mutual information \itfacts{chain-rule}}\\
	&\hspace{30pt}\leq \frac1{\beta_r + 1} \cdot \en{\rMinner{x}(y_i) \mid \rid{y_i}, \rWrest, \Econd} \tag{by \itfacts{info-entropy}} \\
	&\hspace{30pt}\leq \frac1{\beta_r + 1} \cdot s. \qedhere \tag{as message length is bounded by $s$ and \itfacts{uniform}}
	\end{align*}
\end{proof}

We can prove \Cref{lem:cH2-cDfake-inner-msgs} again through weak chain rule of total variation distance from \Cref{fact:tvd-chain-rule}. 
We first define the random variables we perform chain rule over.  These are local to this subsection and may be used for different purposes later. (See \Cref{sec:list-rv} for a list of global random variables.)
\begin{itemize}
	\item Variable $\rwstart$: joint random variable $\rG_{r-1}, \rids$, $\raux$ and $\rNpub{x},\rMpub{x}$ for each inner vertex $x$ . 
	\item Variables $\rw^{x \to y}$ for each pair of inner vertices $x, y$ in different layers: random variable $\rMinner{x}(y)$, corresponding to the message $x$ sends to $y$.
	\item Variables $\rwend$: joint random variable $\rNrest{x}$ for each inner vertex $x$. 
\end{itemize}

For each inner vertex $x$, there are $2n_{r-1}$ other inner vertices in a different layer than  $X$. Hence, totally, there are $6(n_{r-1})^2$ random variables of the form $\rw^{x \to y}$. Let us define an ordering on these $6(n_{r-1})^2$ random variables.  We use the lexicographic ordering on $x$, followed by that of $y$ to order them. Let $\rw^1, \rw^2, \ldots, \rw^{6(n_{r-1})^2}$ be these random variables in order. 

We look at the specific distribution of these random variables in $\cH_2$ and $\cDfake$ more carefully. 
We state the following simple observations without proof.
\begin{observation}\label{obs:inter-H2fake-list}
	About random variables $\rwstart, \rw^{x \to y}$ for each pair of inner vertices $x, y$ in different layers and $\rwend$, in distributions $\cH_2$ and $\cDfake$, we have, 
	\begin{enumerate}[label=$(\roman*)$]
		\item \label{item:wstart-same-H2fake} In $\cH_2$ and $\cDfake$, the distribution of random variable $\rwstart$ is the same. 
		\item \label{item:wxy-indep-H2fake} In $\cH_2$ and $\cDfake$, the random variables $\rw^{i}$ and $\rw^{j}$ are independent of each other conditioned on $\rwstart$ for any $i \neq j$, $i,j \in [6(n_{r-1})^2]$. Random variable $\rw^{x \to y}$ is sampled only conditioned on $\rids, \raux$ and $\rNinner{x}{}, \rNpub{x}, \rMpub{x}$ for each inner vertex $x$, all of which are fixed when conditioned on variable $\rwstart$. 
		\item \label{item:wxy-indep-H2} In $\cH_2$, random variable $\rw^{x \to y}$ is independent of the rest of $\rwstart$ when conditioned on $\rNinner{x}{} , \rids, \raux$, $\rNpub{x}$ and $\rMpub{x}$ inside $\rwstart$. 
		\item \label{item:wxy-indep-fake} In $\cDfake$, random variable $\rw^{x \to y}$ is independent of the rest of $\rwstart$ when conditioned on $\rids, \raux, \rNpub{x}, \rtypeval{x, y}$ and $\rMpub{x}$ inside $\rwstart$.
		\item \label{item:wend-indep-H2fake} Conditioned on any choice of $\rwstart$ and $\rw^{x \to y}$ all pairs of inner vertices $x, y$ in different layers, the distribution of random variable $\rwend$ is the same in $\cH_2$ and $\cDfake$. 
	\end{enumerate}
\end{observation}

\begin{claim}\label{clm:inter-H2fake}
	For any $\ell \in [6(n_{r-1})^2]$, we have, 
	\[
	\Exp_{\rwstart, \rw^1, \ldots, \rw^{\ell-1} \sim \cH_2}	\tvd{\cH_2(\rw^{\ell} \mid \rw^1, \ldots, \rw^{\ell-1}, \rwstart)}{\cDfake(\rw^{\ell} \mid \rw^1, \ldots, \rw^{\ell-1}, \rwstart)} \leq  \sqrt{\frac{s}{2(\beta_r + 1)}}.
	\]
\end{claim}
\begin{proof}
	Let $\rw^{\ell} = \rw^{x \to y}$.
	Firstly, by \Cref{obs:inter-H2fake-list}-\ref{item:wxy-indep-H2fake}, we have, 
	\begin{align*}
	\cH_2(\rw^{\ell} \mid \rw^1, \ldots, \rw^{\ell-1}, \rwstart) &= \cH_2(\rw^{\ell} \mid  \rwstart), \quad \text{and} \\
		\cDfake(\rw^{\ell} \mid \rw^1, \ldots, \rw^{\ell-1}, \rwstart) &= \cDfake(\rw^{\ell} \mid  \rwstart).
	\end{align*}
	Hence the total variation distance term becomes, 
	\begin{align*}
		&\Exp_{\rwstart, \rw^1, \ldots, \rw^{\ell-1} \sim \cH_2}	\tvd{\cH_2(\rw^{\ell} \mid \rw^1, \ldots, \rw^{\ell-1}, \rwstart)}{\cDfake(\rw^{\ell} \mid \rw^1, \ldots, \rw^{\ell-1}, \rwstart)} \\
		&\hspace{30pt}=	\Exp_{\rwstart\sim \cH_2}	\tvd{\cH_2(\rw^{\ell} \mid  \rwstart)}{\cDfake(\rw^{\ell} \mid  \rwstart)} \tag{by the equalities above} \\
		&\hspace{30pt}\leq \frac1{\sqrt{2}} \cdot \Exp_{\rwstart \sim \cH_2} \sqrt{\kl{\cH_2(\rw^{\ell} \mid \rwstart)}{\cDfake(\rw^{\ell} \mid \rwstart)}} \tag{by Pinsker's inequality in \Cref{fact:pinskers}}\\
		&\hspace{30pt}\leq \frac1{\sqrt{2}} \cdot \sqrt{ \Exp_{\rwstart \sim \cH_2} \kl{\cH_2(\rw^{\ell} \mid \rwstart)}{\cDfake(\rw^{\ell} \mid \rwstart)}} \tag{by Jensen's inequality and concavity of square root} \\
		&\hspace{30pt}= \frac1{\sqrt{2}} \cdot \sqrt{ \Exp_{\rwstart \sim \cH_2} \kl{(\rMinner{x}(y) \mid \eqcolor{\rNinner{x}{}, \rNpub{x}, \rMpub{x}, \rids, \raux})}{\cDfake(\rw^{\ell} \mid \rwstart)}} \tag{by \Cref{obs:inter-H2fake-list}-\ref{item:wxy-indep-H2}} \\
		\hspace{30pt}&= \frac1{\sqrt{2}} \cdot \sqrt{ \Exp_{\substack{\rids, \raux, \rNinner{x}{}, \\ \rMpub{x}, \rNpub{x} }} \kl{\rMinner{x}(y) \mid \substack{\rNinner{x}{}, \rids, \raux, \\ \rNpub{x}, \rMpub{x}}}{\rMinner{x}(y) \mid \eqcolor{\substack{\rids, \raux, \rMpub{x}, \\ \rNpub{x}, \rtypeval{x,y}}}}} \tag{by \Cref{obs:inter-H2fake-list}-\ref{item:wxy-indep-fake}}\\
		&\hspace{30pt} =  \frac1{\sqrt{2}} \cdot \sqrt{\mi{\rMinner{x}(y)}{\rNinner{x}{} \mid \rids, \raux, \rNpub{x}, \rMpub{x}, \rtypeval{x, y}}}   \tag{by relation between KL-Divergence and mutual information in \Cref{fact:kl-info}} \\
		&\hspace{30pt}\leq \frac1{\sqrt{2}} \cdot \sqrt{\frac{s}{\beta_r + 1}}. \qedhere \tag{by \Cref{clm:mi-low-inner-msgs}} \qedhere
		\end{align*}
\end{proof}

Proof of \Cref{lem:cH2-cDfake-inner-msgs} is simple now. 
\begin{proof}[Proof of \Cref{lem:cH2-cDfake-inner-msgs}]
We have, 
\begin{align*}
	& \tvd{\cH_2}{\cDfake} \\
	&\hspace{30pt}= \tvd{\cH_2(\rwstart)}{\cDfake(\rwstart)} \\
	&\hspace{60pt} + \sum_{\ell \in [6(n_{r-1})^2]} \tvd{\cH_2(\rw^{\ell} \mid \rwstart, \rw^1, \ldots, \rw^{\ell-1})}{\cDfake(\rw^{\ell} \mid \rwstart, \rw^1, \ldots, \rw^{\ell-1})} \\
	&\hspace{90pt} + \tvd{\cH_2(\rwend \mid \rwstart, \rw^1, \ldots, \rw^{6(n_{r-1})^2})}{\cDfake(\rwend \mid \rwstart, \rw^1, \ldots, \rw^{6(n_{r-1})^2})} \tag{by \Cref{fact:tvd-chain-rule}}\\
	&\hspace{30pt}=  \eqcolor{0} + \sum_{\ell \in [6(n_{r-1})^2]} \tvd{\cH_2(\rw^{\ell} \mid \rwstart, \rw^1, \ldots, \rw^{\ell-1})}{\cDfake(\rw^{\ell} \mid \rwstart, \rw^1, \ldots, \rw^{\ell-1})} \\
	&\hspace{60pt} + \tvd{\cH_2(\rwend \mid \rwstart, \rw^1, \ldots, \rw^{6(n_{r-1})^2})}{\cDfake(\rwend \mid \rwstart, \rw^1, \ldots, \rw^{6(n_{r-1})^2})} \tag{by \Cref{obs:inter-H2fake-list}-\ref{item:wstart-same-H2fake}} \\
	&\hspace{30pt}=  \sum_{\ell \in [6(n_{r-1})^2]} \tvd{\cH_2(\rw^{\ell} \mid \rwstart, \rw^1, \ldots, \rw^{\ell-1})}{\cDfake(\rw^{\ell} \mid \rwstart, \rw^1, \ldots, \rw^{\ell-1})} + \eqcolor{0} \tag{by \Cref{obs:inter-H2fake-list}-\ref{item:wend-indep-H2fake}} \\
	&\hspace{30pt}\leq 6(n_{r-1})^2 \cdot  \sqrt{\frac{s}{2(\beta_r + 1)}}. \qedhere \tag{by \Cref{clm:inter-H2fake}}
\end{align*}
\end{proof}

This concludes the analysis of our round elimination protocol and the proof of~\Cref{lem:round-elim-distance}, which in turn, as shown before, 
finalizes the entire proof of~\Cref{lem:round-elim} and consequently~\Cref{thm:hard-dist}.

\bigskip
\section*{Acknowledgements} 
\addcontentsline{toc}{section}{Acknowledgements}

We thank Keren Censor-Hillel and Seri Khoury for introducing us to this problem through open problem sessions and discussions at 
Dagstuhl Seminar ``Graph Algorithms: Distributed Meets Dynamic'' and Simons Institute program on ``Sublinear Algorithms''. We also thank Rotem Oshman 
for helpful conversations on prior work in related models. We are additionally grateful to Keren Censor-Hillel, Seth Pettie and Thatchaphol Saranurak for many helpful comments and pointers to the literature, and to 
Artur Czumaj for the suggestion to add \Cref{sec:cktbounds} to this paper. 

Finally, 
we would like to express our gratitude to the organizers of the above Dagstuhl Seminar and Simons Institute program as well as ``New York Theory Day (Fall 2024)''
that initiated and facilitated of some of these discussions. 


\bibliographystyle{halpha-abbrv}
\bibliography{new}


\appendix

\part*{Appendix}

\newcommand{\ccond}{\ensuremath{c_{\textnormal{\textsc{cond}}}}}

\newcommand{\cdecomp}{\ensuremath{c_{\textnormal{\textsc{exp}}}}}
\newcommand{\crout}{\ensuremath{c_{\textnormal{\textsc{P}}}}}

\newcommand{\cquery}{\ensuremath{c_{\textnormal{\textsc{Q}}}}}

\newcommand{\vol}{\ensuremath{vol}}
\newcommand{\stat}{\ensuremath{\pi}}

\newcommand{\taumix}{\ensuremath{\tau_{mix}}}

\newcommand{\nid}{\ensuremath{new}}

\newcommand{\Texp}{\ensuremath{T_{\textnormal{exp}}}}

\newcommand{\Prou}{\ensuremath{P_{\textnormal{rou}}}}
\newcommand{\Qrou}{\ensuremath{Q_{\textnormal{rou}}}}

\section{CONGEST Lower Bounds and Circuit Complexity Barriers}\label{sec:cktbounds}

This appendix is a note on the barrier proven by \cite{EdenFFKO22} for proving lower bounds on the number of rounds required to solve triangle detection in CONGEST, and how they may break barriers in circuit complexity. 
We optimize the parameters a bit further to improve the result from $n^{\Omega(1)}$ lower bound breaking barriers in circuit complexity to $n^{\Omega(1/\log \log n)}$ lower bound breaking the same barriers in a weaker sense, but the framework of these reductions are exactly the same as \cite{EdenFFKO22}, and we claim no novelty for these ideas. 

Another purpose of this appendix is to keep track of all parameters involved explicitly which leads to the following observation. Even though, strictly speaking, proving an $n^{o(1/\log\log{n})}$ lower bound for triangle detection
does not break circuit complexity barriers \emph{currently}, it is \emph{plausible} that even proving an $\Omega(\log^c n)$ lower bound for some large constant $c > 1$ could  break circuit complexity barriers \emph{assuming} CONGEST algorithms for expander routing can be made to run in $\poly\log{(n)}$ rounds (see \Cref{rem:cktbounds})\footnote{To our knowledge, currently, no CONGEST lower bounds rule out such a possibility.}.

\subsection{Background on Expander Decomposition and Routing}

In this section, we state the prior results on expander decomposition and routing in graphs with high conductance in the CONGEST model. We start with some basic definitions and results. 

\paragraph{Definitions and terminology.}
In any graph $G = (V, E)$, we use $m$ to denote the number of edges and $n$ to denote the \emph{total} number of vertices.\footnote{In the other parts of this paper, $n$ referred to the number of vertices in one layer of a tripartite graph. Throughout this appendix, we are not working with tripartite graphs and use $n$ to denote the total number of vertices.} We use $\deg(v)$ to denote the degree of $v$. 

For any subset $S \subset V$, we use $\vol(S)$ to denote $\sum_{v \in S} \deg(v)$. 
The \textbf{conductance of a cut} $S$ and $V \setminus S$ is defined as:
\[
	\Phi(S) = \frac{\card{\set{(u,v) \mid (u,v) \in E, u \in S, v \in V \setminus S}}}{\min\{\vol(S), \vol(V \setminus S)\}}.
\] 
For the special case $S = \emptyset$, we set $\Phi(S) = 0$. The \textbf{conductance of a graph} $G$ is defined as, 
\[
	\Phi = \min_{S \subset V, 0 < \card{S} < n} \Phi(S).
\]

A \textbf{uniform lazy random walk} in a graph starting from a vertex $s \in V$ performs the following step iteratively: with probability $1/2$, stay at the current vertex; otherwise, pick a uniformly random neighbor of the current vertex and move to that neighbor. 
In any connected graph $G$, the stationary distribution of a lazy random walk is $\stat(u) = \deg(u)/2m$ for all $u \in V$. 

Let $p^s_t(v)$ denote the probability that after $t \geq 0$ steps, a lazy random walk starting from vertex $s \in V$ lands at $v$ for each $v \in V$. Then the \textbf{mixing time}, denoted by $\taumix(G)$ of graph $G$ is defined as the minimum $t$ such that for all $s, v \in V$, \[
	\card{p^s_t(v)-\stat(v)} \leq \stat(v)/n.
\]  
The following relation between mixing time and conductance is well known \cite{MarkA89}:\begin{equation}\label{eq:mixtime}
	\Theta\paren{\frac1{\Phi}} \leq \taumix(G) \leq \Theta\paren{\frac{\log n}{\Phi^2}}.
\end{equation}
It is also known that the diameter $D$ of any graph is at most the mixing time of the graph. 

\subsubsection{Results on Expander Decomposition in CONGEST}
An expander is any graph which has large conductance, usually at least $1/\poly\!\log{\!(n)}$. 
We now define the standard expander decomposition, where the vertex set of the input graph has to be partitioned into clusters that form expanders.

\begin{definition}[Expander Decomposition] \label{prob:exp-decompose}
	Given a graph $G = (V, E)$, and parameters $\eps \in (0,1)$ and $\phi \in (0,1)$, an expander decomposition of $G$ is a partition of $V$ into sets $V_1 \sqcup V_2 \sqcup\ldots \sqcup V_{\ell}$ for some integer $\ell > 0$ such that, 
	\begin{enumerate}
		\item The subgraph of $G$ induced by each partition $V_i$ for $i \in [\ell]$ is an expander with conductance at least $\phi$. 
		\item The number of edges $(u,v) \in E$ such that $ u \in V_i$ and $v \in V_{j}$ for $i \neq j \in [\ell]$, also called \emph{\textbf{inter-cluster edges}}, is at most $\eps \cdot \card{E}$. 
	\end{enumerate}
	Each partition $V_i$ for $i \in [\ell]$ is termed a \emph{\textbf{cluster}} of the decomposition. 
\end{definition} 

In the CONGEST model, we want to perform expander decomposition and require that at the end of the decomposition, each node knows the partition that it belongs to. 

\begin{problem}[Expander Decomposition in CONGEST]\label{prob:decomp-CONGEST}
	Given a graph $G = (V, E)$ in the CONGEST model, given parameters $\eps, \phi \in (0,1)$, the expander decomposition in CONGEST problem 
	is to find an expander decomposition $V_1 \sqcup \ldots \sqcup V_{\ell}$ with parameters $\phi $ and $\eps $ such that each vertex $v \in V$ knows which partition $V_i$ for $i \in [\ell]$ it belongs to. 
	
	We use $\Texp(\eps, \phi, n)$ to denote the minimum number of CONGEST rounds that such a decomposition can be implemented in for all $n$-vertex graphs. 
\end{problem}

The following result is known for expander decomposition in CONGEST. 

\begin{proposition}[\!\!{\cite[Theorem 1.2]{ChenMGS25}}]\label{prop:exp}
	There is an algorithm in the CONGEST model that solves \Cref{prob:decomp-CONGEST}; there exists some constants $\ccond, \cdecomp > 1$ such that, for any $\phi \in (0,1)$, 
	\[
		\Texp(\phi \cdot \log^{\ccond}(n) , \phi, n) \leq (\log n)^{\cdecomp } \cdot \phi^{-4}.
	\]
	The algorithm is randomized and succeeds with high probability.
\end{proposition} 

It is easy to see that once \Cref{prob:decomp-CONGEST} is solved, each vertex can also identify the number of vertices and the edges inside the cluster that it belongs to. 

\begin{claim}\label{clm:find-aux}
In any graph $G$, after \Cref{prob:decomp-CONGEST} is solved with parameters $\eps, \phi$, in an additional $O(\log n \cdot \phi^{-2})$ rounds, each vertex $v \in V$ can also know:
\begin{enumerate}
	\item For each $u \in N(v)$, which cluster $V_i$ for $i \in [\ell]$ that $u$ belongs to. 
	\item If $v \in V_i$, the size of cluster $V_i$, and the total number of edges inside cluster $V_i$. 
	\item If $v \in V_i$, each vertex also gets an identity from $[\card{V_i}]$ that is distinct from the identities of all other vertices in $V_i$, and $v$ knows the identities of all its neighbors that are also present in $V_i$.
\end{enumerate}
\end{claim}

\begin{proof}
	To communicate the cluster that each neighbor of $v$ belongs to, one additional round is sufficient, as each vertex $u \in V$ can broadcast to all its neighbors the identity of the cluster $u$ belongs to in $\log n$ bits. 
	
	Let $v \in V_i$, let $G_i$ be the subgraph of $G$ induced by $V_i$, and let $D$ be the diameter of $G_i$. The vertices can compute a rooted spanning tree of $G_i$ in $O(D)$ time by performing a breadth-first search from each node, and using the identities of the nodes in $V$ for tie-breaking in choosing which breadth-first tree becomes the spanning tree. 
	Once such a spanning tree is computed, the total number of vertices and edges inside $G_i$ can be collected through the edges of the tree in $O(D)$ many rounds. The required identities of vertices in $V_i$ can also be generated by propagating new identities from the root vertex of the spanning tree in $O(D)$ many rounds. 
	
	We know that the diameter of the graph $G_i$ is at most $O(\log n \cdot \phi^{-2})$, by \Cref{eq:mixtime} and the properties of the expander decomposition, which proves the claim.
\end{proof}

\subsubsection{Results on Routing Algorithms in CONGEST}

We cover the results on routing demand instances from \cite{GhaffariKS17}, following the presentation of these results in \cite{ChangPSZ21}.

\begin{problem}[Routing in CONGEST]\label{prob:routing}
	In a graph $G = (V, E)$ in the CONGEST model, an instance of the routing problem is defined as:
	\begin{enumerate}
		\item There are $a$ source-destination pairs $(s_1, t_1), (s_2, t_2), \ldots, (s_a, t_a)$ for some $a \in \IN$. Each source $s_i $ has to send a message of length at most $O(\log n)$ to vertex $t_i$ for $i \in [a]$. Each vertex $v \in V$ knows all vertices $u$ such that $(v, u)$ is one of the source-destination pairs, and the message $v$ should send to $u$.\footnote{The vertices need not know any other information about the source-destination pairs.}
		\item Each vertex $v \in V$ is the source of at most $\deg(v)$ messages and is the destination of at most $\deg(v)$ messages. 
	\end{enumerate}
	There may be multiple instances of the routing problem that need to be routed sequentially.
	
	We use a tuple $(\Prou(n), \Qrou(n))$ to denote the runtime of an algorithm in CONGEST that on $n$-vertex graphs, uses $\Prou(n) \cdot \taumix$ rounds of preprocessing steps, and after these steps, can route each instance of the routing problem in
	 $\Qrou(n) \cdot \taumix$ time, 
	where $\taumix$ is the mixing time of the underlying graph $G$.  
\end{problem}

Let us state one existing routing algorithm that we employ. This way of presenting the results of \cite{GhaffariKS17} can be found in Section 2.1 of \cite{ChangPSZ21}.
\begin{proposition}[cf. Lemma 3.2, 3.3 and 3.4 of \cite{GhaffariKS17}]\label{prop:routing}
There exists constants $\crout, \cquery > 1$ such that, for any integer $k \geq 1$, there exists a routing algorithm with:
\begin{align*}
	\Prou(n) &\leq k \cdot n^{4/k} \cdot (\log n)^{\crout \cdot k}  \cdot O(1), \\
	 \Qrou(n) &\leq (\log n)^{\cquery \cdot k} \cdot O(1),
\end{align*}
that succeeds with high probability.
Note that there is also a $\taumix$ factor in the total number of rounds based on the definitions in \Cref{prob:routing}. 
\end{proposition}

Another routing algorithm is from \cite{GhaffariL18}, which is better than the ones we state from \cite{GhaffariKS17} for a certain ranges of number of instances of routing, although it does not admit a tradeoff in preprocessing and routing requests. 

To conclude our background section, we summarize the problems we consider and the functions corresponding to their runtimes. This may help the reader keep track of these functions while perusing the next subsection. 

\renewcommand{\arraystretch}{1.5}

\begin{center}
	\begin{tabular}{|c|p{10cm}|}
		\hline
		\textbf{Function}  & \textbf{Definition} \\
		\hline
		$\Texp(\eps, \phi, n)$ &The minimum number of rounds taken by a CONGEST algorithm that finds an expander decomposition with parameters $\eps, \phi$ of the input $n$-vertex graph. \\
		\hline
		$\Prou(n)$ &In the routing problem, $P(n) \cdot \taumix$ is the total number of rounds required for preprocessing. \\
		\hline 
		$\Qrou(n)$ &In the routing problem, after pre-processing, $Q(n) \cdot \taumix$ is the total number of rounds required to handle one instance of the routing demand problem. \\
		\hline
	\end{tabular}
\end{center}

\subsection{From Circuits to CONGEST Algorithms}

In this section, we will show how CONGEST lower bounds can also imply circuit lower bounds for an explicit function.
For $n, m \geq 1$, let $f_{n,m} : \{0,1\}^{2m \log n} \rightarrow \{0,1\}$  denote the function that takes as input 
a representation of a graph $G = (V, E)$ where each edge is given by $ 2 \log n$ bits denoting the identities of the two end points of the edge in $\log n$ bits each. The function outputs $1$ if the input graph has a triangle, and outputs $0$ otherwise. 

We show that we can simulate a circuit that computes $f_{n,m}$ in the CONGEST model. 

\begin{lemma}\label{lem:ckts-congest}
Given a family of circuits $\mathcal{C} = \{C\}_{n, m \geq 1}$ that compute function $f_{n,m}$ with constant fan-in and fan-out, depth $d(n, m)$ and total number of gates at most $m \cdot s(n, m)$, 
for any graph $G = (V, E)$, for any parameters $\eps \in (0,1/3)$ and $\phi \in (0,1)$, there is an algorithm in CONGEST that solves triangle detection in 
\[
	\Paren{\Texp(\eps, \phi, n) + O(\log n \cdot \phi^{-2})  \cdot \paren{\Prou(n) + \Qrou(n) \cdot s(n, m) \cdot d(n, m) }} \cdot O\paren{\frac{\log n}{\log (1/3\eps)}}
\]
many rounds.
\end{lemma}

\begin{proof}
First, we perform expander decomposition on graph $G$ with parameters $\eps, \phi$ in $\Texp(\eps, \phi, n)$ many rounds. Each cluster performs the following procedures for simulating circuits in parallel. 

\paragraph{Choosing the circuit.}
We use the result in \Cref{clm:find-aux} to make sure each vertex in any cluster $V_i$ knows all the information listed in the statement of \Cref{clm:find-aux}. We use $\nid(v)$ to denote the identities from \Cref{clm:find-aux}.

In each cluster $V_i$, a vertex $v \in V_i$ is said to be a \textbf{good node} if the number of neighbors of $v$ inside cluster $V_i$ is at least as many as the number of neighbors of $v$ outside cluster $V_i$. All other vertices are called \textbf{bad nodes}.

We let graph $G_i$ be the subgraph of $G$ with the union of all the edges incident on the good nodes of $V_i$ (including ones which go out of $V_i$), and all the edges between two vertices in $V_i$. Using an argument similar to the one in \Cref{clm:find-aux}, the total number of these edges can be known to all vertices in $V_i$ in $O(\log n \cdot \phi^{-2})$ rounds. Let $m_i$ be the total number of these edges.

Each cluster $V_i$ simulates the circuit $C_{n, m_i}$, and each vertex in cluster $V_i$ knows the value of $n$ and $m_i$. Hence, all the vertices in the cluster know which circuit to simulate.

\paragraph{Assignment of gates to vertices.}
Each gate $g$ of the circuit $C_{n, m_i}$ is assigned to one of the nodes in $V_i$ that is responsible for computing the output of gate $g$, as well as sending this output to all the gates in the next layer of the circuit that $g$ feeds into. The assignment of gates is as follows:
\begin{enumerate}[label=$(\roman*)$]
	\item Each gate $g$ of the input layer corresponds to some edge $(u,v)$ of graph $G$. If both $u, v \in V_i$, all these $2 \log n$ gates are assigned to either $u$ or $v$, lexicographically breaking ties based on $\nid(u)$ and $\nid(v)$. Moreover, which gate among the input layer corresponds to which edge can be fixed by arguments similar to the proof of \Cref{clm:find-aux}: the root of the spanning tree can communicate to each of its neighbors a partition of the input gates corresponding to the number of edges in each of the subtrees, and this can be done recursively. This will take $O(\log (n) \cdot \phi^{-2})$ rounds, also based on the upper bound on the diameter of the cluster.
	
	\item Each gate in any layer which is not the input layer, is assigned to a vertex in $V_i$ so that the total number of gates in any layer given to any vertex $v$ is at most $2 \cdot s(n, m_i) \cdot \deg_i(v)$, where $\deg_i(v)$ is the total number of neighbors of $v$ inside the cluster as follows:
	\begin{enumerate}
		\item The gates assigned to each vertex $v$ form $2 \cdot s(n, m_i)$ many disjoint contiguous segments in layer $j$. 
		\item There is a partition of $[m_i]$ into $\card{V_i}$ many contiguous segments, and each segment is given to a vertex $v \in V_i$. Let $[g_{v, 1}, g_{v, 2}]$ be the segment given to vertex $v$, and we require that $g_{v, 2} - g_{v, 1} + 1 = \deg_i(v)$. 
		\item All the gates numbered $(x-1) \cdot (m_i) + y$ is assigned to vertex $v$ for each $x \in [2\cdot s(n, m_i)]$, $g_{v, 1} \leq y \leq g_{v, 2}$ if such a gate exists. 
	\end{enumerate}
The total number of gates in each layer is at most $m_i \cdot s(n, m_i)$, so, all the gates are assigned to some vertex. 
Each vertex $v$ can also be made to know the integers $g_{v, 1}$ and $g_{v, 2}$ with arguments similar to assigning gates in the input layer: the root picks a contiguous segment of length equal to its degree for itself starting from $1$, then partitions the rest of the gates among its subtrees based on the number of edges. This step is performed recursively in each subtree. Totally in $O(\log(n) \cdot \phi^{-2})$ rounds, each vertex $v$ knows $g_{v, 1}$ and $g_{v, 2}$ and hence, all the gates that are assigned to it. 
\end{enumerate}
Next, we show how the vertices can communicate to each other the output of the gates they are assigned to. 

\paragraph{Routing gate outputs.}
The vertices simulate the circuit layer by layer starting from the input layer.  We prove that this simulation is possible by induction. For the base case, we know that the value of all the gates in the input layer is known to the vertices to which the gates are assigned to trivially. 
For the induction hypothesis, let us assume that the output of all the gates in layer $j - 1$ is known to the vertices that they are assigned to. We will show that the output of all these gates in layer $j-1$, once found, can be routed to the vertices that are assigned the gates in layer $j$ that they feed into.

\emph{Only} for the purpose of routing the output of gates in each layer, we assume that each vertex $v \in V_i$ has $\deg_i(v)$ many identities:  all integers from $g_{v, 1}$ to $g_{v, 2}$ (including $g_{v ,1}$ and $g_{v, 2}$). The identities of each vertex can be specified by $O(\log m_{i})$ bits $g_{v, 1}, g_{v, 2} \leq m_i$. The vertices can send all their identities to their neighbors in $1$ additional round. 

The vertices use the preprocessing routine for the routing algorithm, which can be run in $\Prou(n)$ rounds.  Next, for each gate $g$, the vertex $v$ that $g$ is assigned to needs to send the output of $g$ to $O(1)$ many vertices, which are assigned the gates in layer $j$ that $g$ feeds into.  The identities of these vertices are known to $v$ as it knows circuit $C_{n,m_i}$ (if the gate number is $g$, $(g-1) \mod m_i + 1$ is the identity of the vertex that $g$ is assigned to) and the output of each gate is also only $1$ bit.
  The total number of gates assigned to each vertex $v$ is at most $s(n, m_i) \cdot 2 \cdot \deg_i(v)$, this can be viewed as $O(s(n, m_i))$ many instances of the routing problem, each of which takes $\Qrou(n) \cdot O(\log n \cdot \phi^{-2})$.\footnote{Here, we assign multiple identites to each vertex, to have a total of $\deg_i(v)$ many identites each, and we claim the routing algorithm can still be run correctly. This can be easily seen by performing the routing algorithm on another graph with $m_i$ vertices: a replacement product of the original graph in cluster $V_i$ with a family of $O(1)$-regular low-degree expanders of conductance $\Omega(1)$. Each vertex in the new graph will have a unique identity from $[m_i]$, and any round of a CONGEST algorithm can be simulated in the original graph in cluster $V_i$. Refer to Section 4 and Proposition 4.2 of \cite{AssadiSW19} for details on replacement products and the fact that the mixing time increases by at most an $O(1)$ factor, respectively.  } 
  
This routing process is repeated till the last layer of the circuit, and if a triangle is found in graph $G_i$ the vertex with the output gate of the circuit outputs the existence of the triangle, and the algorithm terminates. 

\paragraph{Recursion.}
If there is no triangle found in any of the clusters, all the edges $(u,v)$ where $u, v$ are two good nodes inside the same cluster, i.e., $u, v \in V_i$ for some $i \in [\ell]$ are removed. (Note the definition of a good node is with respect to the cluster it belongs to.) First, we argue that we remove at least $(1-3\eps) \cdot m$ edges. 

We know by the parameters on the expander decomposition that the total number of inter-cluster edges is at most $\eps \cdot m$. We will argue that each edge $(u,v)$ with $u, v \in V_i$ that is not removed can be charged to an inter-cluster edge, so that each inter-cluster edge receives at most two charges. This will prove that the total number of edges left is at most $3 \cdot \eps \cdot m$. 
Let $(u_1, v), (u_2, v), \ldots, (u_a, v)$ be all the edges incident to some bad node $v$ in cluster $V_i$ with $u_j \in V_i$ for all $j \in[a]$. We know that $v$ has at least $a$ many edges to vertices outside cluster $V_i$ by the definition of bad nodes. We charge each of these edges to one such inter-cluster edge incident on $v$, so that each inter-cluster edge receives at most one charge. Globally, each inter-cluster edge $(x, y)$ receives at most two charges, as it receives at most one charge from each of its end-points. 

We have shown that at least half the edges are removed, and we recurse on the remaining network by computing another expander decomposition and circuit simulation. This recursion is done till there are only $O(1)$ edges left in the graph, wherein the existence of a triangle can be detected trivially in $O(1)$ rounds. The total number of levels of recursion is at most $\log m/(\log (1/3\eps))$. 

\paragraph{Correctness.}
It is easy to see that if the algorithm outputs the existence of a triangle, there is a triangle in the graph, as the outputs of the circuits are correct. 
The edges we remove at each level of the recursion are incident on two good nodes from the same cluster $V_i$. The only way such an edge $(u,v)$ can be a part of a triangle is if there exists some vertex $w \in V$ with $(u, w), (v, w) \in E$. However, when the circuit is simulated in $V_i$, the edges $(u, w)$ and $(v, w)$ would be included as a part of the input, as $u, v$ are good nodes. Hence the existence of a triangle would be detected by the circuit. 
Therefore, when the algorithm declares that there is no triangle after all the levels of the recursion, no triangle exists in the graph. 

\paragraph{Total number of rounds.}
The number of levels of the recursion is $O((\log n)/(\log (c_2/3)))$. Assuming $\Texp(\eps, \phi, n), \Prou(n), \Qrou(n), s(n, m)$ and $d(n, m)$ are increasing functions in $n, m$, each level of the recursion takes:
\begin{enumerate}[label=$(\roman*)$]
	\item $\Texp(\eps, \phi, n)$ rounds to compute the expander decomposition. 
	\item $O(\log n \cdot \phi^{-2})$ rounds to propagate the required information to all the nodes in each cluster before circuit simulation. 
	\item $\Prou(n) \cdot O(\log n \cdot \phi^{-2})$ pre-processing time before circuit simulation via routing requests (as $\taumix = O(\log n \cdot \phi^{-2})$).
	\item $O(s(n,m)) \cdot \Qrou(n) \cdot O(\log n \cdot \phi^{-2})$ rounds to compute one layer of the circuit, repeated for $d(n,m)$ many layers. 
\end{enumerate}
This completes the proof.
\end{proof}

The following result now establishes the barrier for proving CONGEST lower bounds. 

\begin{proposition}\label{clm:exist-lb}
	If there exists a large enough constant $c > 0$ and integer $n_0$, such that for all $n \geq n_0$ triangle detection in CONGEST requires $n^{c/\log \log n}$ rounds, then there is an explicit function $f:\{0,1\}^{N} \rightarrow \{0,1\}$ such that 
	the following is true. There exists integer $N_0$ such that for all $N \geq N_0$, there exists no family of circuits computing $f_N$ with depth $O(\log N)$, constant fan-in, constant fan-out and size $N^{1+c'/\log \log N}$ for some constant $c'$ 
	depending on $c$ only.
\end{proposition}
\begin{proof}
The proof is a direct application of \Cref{lem:ckts-congest} and the results of \Cref{prop:exp,prop:routing}. 

Firstly, in \Cref{prop:exp}, we set $\phi = \log^{-\ccond }(n)/6$. This gives:
\begin{align*}
	&\Texp(\phi \cdot \log^{\ccond}(n), \phi, n) \leq (\log n)^{\cdecomp} \cdot \phi^{-4} \\
	&\implies \Texp(1/6, (\log n)^{-\ccond}, n) \leq (\log n)^{\cdecomp + 4\ccond}.
\end{align*}
This also sets the mixing time of the clusters as $O(1) \cdot (\log n)^{2\ccond +1}$, by \Cref{eq:mixtime}.

Next, we set parameter $k$ in \Cref{prop:routing} to be $k= \log \log n/100 (\cquery + \crout)$. This gives:
\begin{align*}
	&\Prou(n) \leq O(\log \log n)\cdot (n^{400(\cquery + \crout)/\log \log n}  \cdot (\log n)^{\log \log n}) \leq n^{(2 + 400(\cquery + \crout))/(\log \log n)},\\
	&\Qrou(n) \leq O(1) \cdot (\log n)^{\log \log n} \leq n^{1/\log \log n},  
\end{align*}
for large enough $n$. 

Let $N := 2m \log n$, and $f_N$ be the function that takes as input $m$ edges, each represented by $ 2\log n$ bits for both the end-points, and computes whether there is a triangle in the input graph. Given a circuit of depth $O(\log N) = O(\log n)$, and size $s(n,m)$, 
by \Cref{lem:ckts-congest}, there is an algorithm that solves triangle detection in CONGEST in
\begin{align*}
	&(\Texp(1/6, (\log n)^{-\ccond}, n) + O(\log^{2\ccond + 1} n)  \cdot (\Prou(n) + \Qrou(n) \cdot s(n, m) \cdot d(n, m))) \cdot O(\log n) \\
	&= n^{O(1/\log \log n)} \cdot s(n,m). \tag{by our equations above, and depth $O(\log n)$}
\end{align*}
By our premise, we get $s(n,m) \geq n^{c'/\log \log n}$ for some constant $c'$, which is a function of $c$. This proves that the size of the circuits is at least $m \cdot n^{c'/\log \log n} \geq m^{1+ \Omega(1/\log \log m)}$, for an explicit function on at least
 $N = \Theta(m\log{m})$ bits.
\end{proof}

\paragraph{Barrier aspect of~\Cref{clm:exist-lb}.} We can now provide the main message of this appendix. 
\Cref{clm:exist-lb} implies a \emph{barrier} because of the following: the current lower bounds in circuit complexity on the size of a  circuit with constant fan-in and fan-out computing an explicit function on $N$ bits using arbitrary Boolean operator gates in $O(\log N)$ depth is $(3 + 1/86)N - o(N)$ (see \cite{GolovnevKW21}). Obtaining any superlinear lower bound on the size of such circuits has been a major open problem in complexity theory for decades. 
A lower bound of $n^{\Omega(1/\log\log{n})}$ rounds in CONGEST for triangle detection then implies an $N^{1+\Omega(1/\log\log{N})}$ size lower bounds for an explicit problem (detecting a triangle) on above circuits, 
way above the longstanding $\Theta(N)$ threshold of existing lower bounds.

\begin{remark}\label{rem:cktbounds}	
	\Cref{clm:exist-lb} identifies $n^{\Omega(1/\log\log{n})}$ rounds as a barrier for proving distributed CONGEST lower bounds for triangle detection. However, the arguments here can imply even more. 
	
	It is plausible that the runtime of the expander routing in \Cref{prob:routing} (which is the bottleneck here) can be improved much beyond what is currently known. 
	In fact, it is possible that these problems can be solved in $\poly\!\log{\!(n)}$ rounds in CONGEST. 
	Thus, in absence of a lower bound that rules out possibility of such an algorithm for expander routing, any lower bound 
	for distributed triangle detection of the form $\Omega(\log^c{(n)})$ rounds, for some sufficiently large constant $c > 1$, has the \emph{potential} of breaking current circuit complexity barriers. 
	
	 As a result, we believe that at this stage, namely in absence of a concrete lower bound on \Cref{prob:routing} for the parameters mentioned earlier, the \emph{limit} of one can realistically hope for an unconditional lower bound for distributed triangle
	  detection is $\poly\!\log{(n)}$ and even $\Omega(\log n)$ rounds. 
\end{remark}

\clearpage

\renewcommand{\arraystretch}{1.5}

\section{List of Random Variables in the Main Proofs}\label{sec:list-rv}

We compile a list of random variables used throughout our proofs and their meanings here. 

\begin{itemize}
\item First, we begin with the random variables associated with inputs from distribution $\cG_r(n_r)$.
\begin{center}
\begin{tabular}{|c|p{10cm}|}
	\hline
\textbf{Random Variable}  & \textbf{Definition} \\
	\hline
$\rG_{r-1}$ & The inner graph sampled in the hard distribution for $r$-rounds \\
$\ridlayer{X}$ & The identities chosen in $G_r$ for inner vertices in layer $X$ of $G_{r-1}$ \\
$\rids$ & The identities chosen for all the inner vertices \\
$\rid{x_i}$ & The identity chosen for inner vertex $x_i$ \\
$\rtypeval{x, w}$ & The type of vertex pair $x, w \in G$ which lies in $[r+1] \cup \{0\}$ \\
$\rNinner{x}{i}$ & The input given to inner vertex $x_i$ in $G_{r-1}$ \\
$\rNbhood{x}{Y}{i}$	 & The $n_r$-length vector given to $x_i$ of the list of its neighbors to layer $Y$ \\
$\rNbhood{x}{Y}{i}[\ell]$ &The type of vertex pair $(x_i, y_{\ell})$ for $\ell \in [n_r]$ \\
$\rNbtogether{x}{i}$ & The $2n_{r-1}$ length vector containing both $\rNbhood{x}{Y}{i}$ and $\rNbhood{x}{Z}{i}$ \\
$\rNbtogether{x}{i}[S]$ & The type of vertex pair $(x, w)$ for $w \in S \subseteq Y \cup Z$ \\
\hline
\end{tabular}
\end{center}

\item Next, we talk about the random variables associated with the messages. 
\begin{center}
	\begin{tabular}{|c|p{10cm}|}
		\hline
		\textbf{Random Variable}  & \textbf{Definition} \\
		\hline
$\rmsg{x}{Y}{i}$ &The $n_r$-length vector denoting list of messages that inner vertex $x_i$ sends to neighboring channels in layer $Y$ \\
$\rmsg{x}{Y}{i}[\ell]$ &The message sent by inner vertex $x_i$ to $y_{\ell}$ for $\ell \in [n_r]$ in the first round, with $\typeval{x_i, y_{\ell}} \leq r$  \\
$\rmsgrest{x}{i}$ & The messages received by inner vertex $x_i$ from its outer neighbors \\
$\rMinner{x_i}$ & The $2n_{r-1}$ length list of messages sent by inner vertex $x_i$ to other inner vertices in first round \\
$\rMinner{x_i}(w)$ &The message sent by inner vertex $x_i$ to some other inner vertex $w$ in the first round \\
$\rNrest{x_i}$ &The random variables that inner vertex $x_i$ samples in \Cref{box:r-1-step3-input-rest} comprised of the input in each group that is not fixed, and the messages sent by all outer vertices to $x_i$ \\
		\hline
	\end{tabular}
\end{center}

\item Finally, we give the random variables associated with the round elimination protocol. 
\begin{center}
	\begin{tabular}{|c|p{10cm}|}
		\hline
		\textbf{Random Variable}  & \textbf{Definition} \\
		\hline
		$\rNpub{x}$ & The random variable containing the types of channels to inner vertex $x$ which are sampled with public randomness in \Cref{box:r-1-step1-ids} \\
		$\rMpub{x}$ & The random variable containing all the messages sent by inner vertex $x$ that are sampled publicly in \Cref{box:r-1-step1-ids} \\
		$\rcJup{x}$ & The random variable associated with the collection of $\alpha_r$ many subsets of $Y \cup Z$ for inner vertex $x$ in \Cref{box:cGtilr-definition} from \ref{item:disjt-sampling-1}-\ref{item:J-appears} \\
		$\rcJall$ &The random variable containing collections $\rcJup{x}$ for each inner vertex $x$, used to break correlation between public messages and inner input\\
		$\rcKupdown{x \to Y}{t, i}$ & For inner vertex $x$, other layer $Y$, type $t \in [r] \cup \{0\}$ and value $i \in[n_{r-1}]$, collection of $\beta_r$ many subsets of $Y \cup Z$ sampled in \ref{item:disjt-sampling-1}-\ref{item:K-appears} of \Cref{box:cGtilr-definition} \\
		$\rcKall$ & The random variable containing collections $\rcKupdown{x \to Y}{t, i}$ for inner vertex $x$, layer $Y$ with $x \notin Y$, type $t \in [r] \cup \{0\}$ and value $i \in [n_{r-1}]$ for breaking correlation between messages to inner vertices and inner inputs  \\
		$\rLupdown{x \to Y}{t, i}$ & For inner vertex $x$, other layer $Y$, type $t \in [r] \cup \{0\}$ and value $i \in[n_{r-1}]$, a set of $\gamma_r$ elements of $Y$ sampled in \ref{item:disjt-sampling-1}-\ref{item:L-appears} of \Cref{box:cGtilr-definition} \\
		$\rLall$ & The random variable containing sets $\rLupdown{x \to Y}{t, i}$ for  inner vertex $x$, layer $Y$ with $x \notin Y$, type $t \in [r] \cup \{0\}$ and value $i \in [n_{r-1}]$ to break correlations among messages to inner vertices  \\
		$ \raux$ & Auxiliary random variables $\rcJall, \rcKall$ and $\rLall$ defined for breaking correlation \\
		\hline
	\end{tabular}
\end{center}
\end{itemize}

\clearpage


\section{A Schematic Organization of the Main Proofs}\label{sec:schematic-org}

This a schematic organization of the flow of main components of the proofs in our paper (we also include~\Cref{clm:Dreal-Drealtil-close} which is minor component compared to the rest but is technically needed
to complete the picture here). Each arrow points to the main component(s) used in the proof of originating component. 

\bigskip

\hspace{-2cm}
\begin{tikzpicture}
    \node[draw, rounded corners, minimum width=6cm, minimum height=2cm, align=center] (box1)
    {
       {\textbf{\Cref{res:main}}} \\
        (Main result) \\[0.5cm] 
        Introduced in {\Cref{sec:intro}} \\ 
        Proved in \Cref{sec:hard-dist}
    };
    
    \node[draw, rounded corners, minimum width=6cm, minimum height=2cm, align=center] (box2) [below=1.5cm of box1]
    {
        \textbf{\Cref{thm:hard-dist}} \\
        (Lower bound for  the hard distribution)  \\[0.5cm] 
        Introduced in {\Cref{sec:hard-dist}} \\ 
        Proved in \Cref{sec:proof-hard-dist}
    };
    
    \node[draw, rounded corners, minimum width=6cm, minimum height=2cm, align=center] (box3) [below=1.5cm of box2]
    {
        \textbf{\Cref{lem:round-elim}} \\
        (Round elimination protocol)  \\[0.5cm] 
        Introduced in {\Cref{sec:proof-hard-dist}} \\ 
        Proved in \Cref{subsec:proof-round-elim}
    };
    
      \node[draw, rounded corners, minimum width=6cm, minimum height=2cm, align=center] (box4) [below=1.5cm of box3]
    {
        \textbf{\Cref{lem:round-elim-distance}} \\
        (Distance of $\cDreal$ and $\cDfake$) \\[0.5cm] 
        Introduced in {\Cref{sec:fake-protocol}} \\ 
        Proved in \Cref{subsec:round-elim-distance}
    };
    
     \node[draw, rounded corners, minimum width=6cm, minimum height=2cm, align=center, dashed] (box-temp) [left=0.5cm of box4]
    {
        \textbf{\Cref{clm:Dreal-Drealtil-close}} \\
        (Distance of $\cDreal$ and $\cDrealtil$) \\[0.5cm] 
        Introduced in {\Cref{subsec:dist-input-message}} \\ 
        Proved in \Cref{subsec:dist-input-message}
    };
    
          \node[draw, rounded corners, minimum width=6cm, minimum height=2cm, align=center] (box5) [below left=1.5cm and 0.5cm of box4]
    {
        \textbf{\Cref{lem:hybrid-1-inner-msgs}} \\
        (Distance of $\cDrealtil$ and $\cH_1$) \\[0.5cm] 
        Introduced in \Cref{subsec:setup-analysis} \\ 
        Proved in \Cref{subsec:msgs-low-correlation}
    };
    
              \node[draw, rounded corners, minimum width=6cm, minimum height=2cm, align=center] (box6) [below=1.5cm of box4]
    {
        \textbf{\Cref{lem:cH1-cH2-public-msgs}} \\
        (Distance of $\cH_1$ and $\cH_2$) \\[0.5cm] 
          Introduced in \Cref{subsec:setup-analysis} \\ 
        Proved in \Cref{subsec:public-msgs}
    };
    
              \node[draw, rounded corners, minimum width=6cm, minimum height=2cm, align=center] (box7) [below right=1.5cm and 0.5cm of box4]
    {
        \textbf{\Cref{lem:cH2-cDfake-inner-msgs}} \\
        (Distance of $\cH_2$ and $\cDfake$) \\[0.5cm] 
          Introduced in \Cref{subsec:setup-analysis} \\ 
        Proved in \Cref{subsec:inner-msgs}
    };

   	 \draw[->, line width=1pt]
    		(box1) -- (box2);
		
	\draw[->, line width=1pt]	
		(box2) -- (box3);
		
	 \draw[->, line width=1pt]	
		(box3) -- (box4);
		
			 \draw[->, line width=1pt]	
		(box4) -- (box5);

			 \draw[->, line width=1pt]	
		(box4) -- (box6);

			 \draw[->, line width=1pt]	
		(box4) -- (box7);

			 \draw[->, line width=1pt, dashed]	
		(box4) -- (box-temp);

\end{tikzpicture}

\clearpage

\section{Background on Information Theory}\label{app:info}

We now briefly introduce some definitions and facts from information theory that are used in our proofs. We refer the interested reader to the text by Cover and Thomas~\cite{CoverT06} for an excellent introduction to this field, 
and the proofs of the statements used in this Appendix. 

For a random variable $\rA$, we use $\supp{\rA}$ to denote the support of $\rA$ and $\distribution{\rA}$ to denote its distribution. 
When it is clear from the context, we may abuse the notation and use $\rA$ directly instead of $\distribution{\rA}$, for example, write 
$A \sim \rA$ to mean $A \sim \distribution{\rA}$, i.e., $A$ is sampled from the distribution of random variable $\rA$. 

\begin{itemize}[leftmargin=10pt]
\item We denote the \emph{Shannon Entropy} of a random variable $\rA$ by
$\en{\rA}$, which is defined as: 
\begin{align}
	\en{\rA} := \sum_{A \in \supp{\rA}} \Pr\paren{\rA = A} \cdot \log{\paren{1/\Pr\paren{\rA = A}}} \label{eq:entropy}
\end{align} 
\noindent
\item The \emph{conditional entropy} of $\rA$ conditioned on $\rB$ is denoted by $\en{\rA \mid \rB}$ and defined as:
\begin{align}
\en{\rA \mid \rB} := \Ex_{B \sim \rB} \bracket{\en{\rA \mid \rB = B}}, \label{eq:cond-entropy}
\end{align}
where 
$\en{\rA \mid \rB = B}$ is defined in a standard way by using the distribution of $\rA$ conditioned on the event $\rB = B$ in \Cref{eq:entropy}.

\item The \emph{mutual information} of two random variables $\rA$ and $\rB$ is denoted by
$\mi{\rA}{\rB}$ and is defined:
\begin{align}
\mi{\rA}{\rB} := \en{A} - \en{A \mid  B} = \en{B} - \en{B \mid  A}. \label{eq:mi}
\end{align}
\noindent
\item The \emph{conditional mutual information} $\mi{\rA}{\rB \mid \rC}$ is $\en{\rA \mid \rC} - \en{\rA \mid \rB,\rC}$ and hence by linearity of expectation:
\begin{align}
	\mi{\rA}{\rB \mid \rC} = \Ex_{C \sim \rC} \bracket{\mi{\rA}{\rB \mid \rC = C}}. \label{eq:cond-mi}
\end{align} 
\end{itemize}

\subsection{Useful Properties of Entropy and Mutual Information}\label{sec:prop-en-mi}

We shall use the following basic properties of entropy and mutual information throughout. 

\begin{fact}\label{fact:it-facts}
  Let $\rA$, $\rB$, $\rC$, and $\rD$ be four (possibly correlated) random variables.
   \begin{enumerate}
  \item \label{part:uniform} $0 \leq \en{\rA} \leq \log{\card{\supp{\rA}}}$. The right equality holds
    iff $\distribution{\rA}$ is uniform.
  \item \label{part:info-zero} $\mi{\rA}{\rB \mid \rC} \geq 0$. The equality holds iff $\rA$ and
    $\rB$ are \emph{independent} conditioned on $\rC$.
    \item \label{part:info-entropy} $\mi{\rA}{\rB \mid \rC} \leq \en{\rB}$ for any random variables $\rA, \rB, \rC$. 
  \item \label{part:cond-reduce} \emph{Conditioning on a random variable reduces entropy}:
    $\en{\rA \mid \rB,\rC} \leq \en{\rA \mid  \rB}$.  The equality holds iff $\rA \perp \rC \mid \rB$.
  \item \label{part:chain-rule} \emph{Chain rule for mutual information}: $\mi{\rA,\rB}{\rC \mid \rD} = \mi{\rA}{\rC \mid \rD} + \mi{\rB}{\rC \mid  \rA,\rD}$.
  \item \label{part:data-processing} \emph{Data processing inequality}: for a function $f(\rA)$ of $\rA$, $\mi{f(\rA)}{\rB \mid \rC} \leq \mi{\rA}{\rB \mid \rC}$. 
   \end{enumerate}
\end{fact}


\noindent
We also use the following two standard propositions, regarding the effect of conditioning on mutual information.

\begin{proposition}\label{prop:info-increase}
  For random variables $\rA, \rB, \rC, \rD$, if $\rA \perp \rD \mid \rC$, then, 
  \[\mi{\rA}{\rB \mid \rC} \leq \mi{\rA}{\rB \mid  \rC,  \rD}.\]
\end{proposition}
 \begin{proof}
  Since $\rA$ and $\rD$ are independent conditioned on $\rC$, by
  \itfacts{cond-reduce}, $\HH(\rA \mid  \rC) = \HH(\rA \mid \rC, \rD)$ and $\HH(\rA \mid  \rC, \rB) \ge \HH(\rA \mid  \rC, \rB, \rD)$.  We have,
	 \begin{align*}
	  \mi{\rA}{\rB \mid  \rC} &= \HH(\rA \mid \rC) - \HH(\rA \mid \rC, \rB) = \HH(\rA \mid  \rC, \rD) - \HH(\rA \mid \rC, \rB) \\
	  &\leq \HH(\rA \mid \rC, \rD) - \HH(\rA \mid \rC, \rB, \rD) = \mi{\rA}{\rB \mid \rC, \rD}. \qed
	\end{align*}
	
\end{proof}

\begin{proposition}\label{prop:info-decrease}
  For random variables $\rA, \rB, \rC,\rD$, if $ \rA \perp \rD \mid \rB,\rC$, then, 
  \[\mi{\rA}{\rB \mid \rC} \geq \mi{\rA}{\rB \mid \rC, \rD}.\]
\end{proposition}
 \begin{proof}
 Since $\rA \perp \rD \mid \rB,\rC$, by \itfacts{cond-reduce}, $\HH(\rA \mid \rB,\rC) = \HH(\rA \mid \rB,\rC,\rD)$. Moreover, since conditioning can only reduce the entropy (again by \itfacts{cond-reduce}), 
  \begin{align*}
 	\mi{\rA}{\rB \mid  \rC} &= \HH(\rA \mid \rC) - \HH(\rA \mid \rB,\rC) \geq \HH(\rA \mid \rD,\rC) - \HH(\rA \mid \rB,\rC) \\
	&= \HH(\rA \mid \rD,\rC) - \HH(\rA \mid \rB,\rC,\rD) = \mi{\rA}{\rB \mid \rC,\rD}. \qed
 \end{align*}

\end{proof}

\subsection{Measures of Distance Between Distributions}\label{sec:prob-distance}

We use two main measures of distance (or divergence) between distributions, namely the \emph{Kullback-Leibler divergence} (KL-divergence) and the \emph{total variation distance}. 

\paragraph{KL-divergence.} For two distributions $\mu$ and $\nu$ over the same probability space, the \textbf{Kullback-Leibler (KL) divergence} between $\mu$ and $\nu$ is denoted by $\kl{\mu}{\nu}$ and defined as: 
\begin{align}
\kl{\mu}{\nu}:= \Ex_{a \sim \mu}\Bracket{\log\frac{\mu(a)}{{\nu}(a)}}. \label{eq:kl}
\end{align}
We also have the following relation between mutual information and KL-divergence. 
\begin{fact}\label{fact:kl-info}
	For random variables $\rA,\rB,\rC$, 
	\[\mi{\rA}{\rB \mid \rC} = \Ex_{(B,C) \sim {(\rB,\rC)}}\Bracket{ \kl{\distribution{\rA \mid \rB=B,\rC=C}}{\distribution{\rA \mid \rC=C}}}.\] 
\end{fact}


\paragraph{Total variation distance.} We denote the \textbf{total variation distance} between two distributions $\mu$ and $\nu$ on the same 
support $\Omega$ by $\tvd{\mu}{\nu}$, defined as: 
\begin{align}
\tvd{\mu}{\nu}:= \max_{\Omega' \subseteq \Omega} \paren{\mu(\Omega')-\nu(\Omega')} = \frac{1}{2} \cdot \sum_{x \in \Omega} \card{\mu(x) - \nu(x)}.  \label{eq:tvd}
\end{align}
\noindent
We use the following basic properties of total variation distance. 
\begin{fact}\label{fact:tvd-small}
	Suppose $\mu$ and $\nu$ are two distributions for $\event$, then, 
	$
	{\mu}(\event) \leq {\nu}(\event) + \tvd{\mu}{\nu}.
$
\end{fact}


We also have the following (chain-rule) bound on the total variation distance of joint variables.

\begin{fact}\label{fact:tvd-chain-rule}
	For any distributions $\mu$ and $\nu$ on $n$-tuples $(X_1,\ldots,X_n)$, 
	\[
		\tvd{\mu}{\nu} \leq \sum_{i=1}^{n} \Exp_{X_{<i} \sim \mu} \tvd{\mu(X_i \mid X_{<i})}{\nu(X_i \mid X_{<i})}. 
	\]
\end{fact}

We also have the following ``over conditioning'' property. 

\begin{fact}\label{fact:tvd-over-conditioning}
	For any random variables $\rX,\rY,\rZ$, 
	\[
		\tvd{\rX}{\rY} \leq \tvd{\rX\rZ}{\rY\rZ} = \Exp_{Z} \tvd{(\rX \mid \rZ=Z)}{(\rY \mid \rZ=Z)}.  
	\]
\end{fact}
\begin{proof}
First, we prove the equality between the second term and the third term in the statement. 
\begin{align*}
	\tvd{\rX\rZ}{\rY\rZ} &=\frac12 \cdot \sum_{W, Z} \Pr[(W, Z)] \card{\Pr[\rX\rZ = (W, Z)] - \Pr[\rY\rZ=(W, Z)]}  \\
	&=\frac12 \cdot \sum_{W, Z}  \card{\Pr[\rZ = Z] \cdot (\Pr[\rX = W \mid  \rZ = Z] - \Pr[\rY = W \mid \rZ= Z])} \\
	&= \sum_{Z} \Pr[\rZ =Z] \cdot \frac12 \sum_{W}  \card{\Pr[\rX = W \mid  \rZ = Z] - \Pr[\rY = W \mid \rZ= Z]} \\
	&= \sum_{Z} \Pr[\rZ =Z] \cdot\tvd{(\rX \mid \rZ= Z)}{(\rY \mid \rZ = Z)} \\
	&= \Exp_{Z} \tvd{(\rX \mid \rZ= Z)}{(\rY \mid \rZ = Z)}.
\end{align*}

Now we prove the inequality between the first term and the third term. 
\begin{align*}
		\tvd{\rX}{\rY} &= \frac12 \cdot \sum_{W} \card{\Pr[\rX = W] - \Pr[\rY = W]} \\
		&= \frac12 \cdot \sum_{W} \card{\sum_{Z} \Pr[\rZ = Z](\Pr[\rX = W \mid \rZ = Z] - \Pr[\rY = W \mid \rZ = Z])} \\
		& \leq  \frac12 \cdot \sum_{W}  \sum_{Z} \Pr[\rZ = Z] \card{\Pr[\rX = W \mid \rZ = Z] - \Pr[\rY = W \mid \rZ = Z]} \\
		&= \sum_{Z}   \Pr[\rZ = Z] \cdot \paren{\frac12 \cdot \sum_{W} \card{\Pr[\rX = W \mid \rZ = Z] - \Pr[\rY = W \mid \rZ = Z]}}  \\
		&= \Exp_{Z} \tvd{\rX \mid \rZ= Z}{\rY \mid \rZ = Z}. \qedhere
\end{align*}
\end{proof}


The following Pinsker's inequality bounds the total variation distance between two distributions based on their KL-divergence.

\begin{fact}[Pinsker's inequality]\label{fact:pinskers}
	For any distributions $\mu$ and $\nu$, 
	$
	\tvd{\mu}{\nu} \leq \sqrt{\frac{1}{2} \cdot \kl{\mu}{\nu}}.
	$ 
\end{fact}

%
%

\end{document}